\documentclass[a4paper,envcountsame]{llncs}
\usepackage{ifthen}
\usepackage{amsmath,amsfonts,amssymb}
\usepackage{complexity}
\usepackage{color}
\usepackage{tikz}
\usetikzlibrary{shapes}
\usepackage{dsfont}
\usetikzlibrary{decorations.pathreplacing}
\usetikzlibrary{calc}
\usetikzlibrary{arrows}
\usetikzlibrary{positioning}
\usetikzlibrary{arrows}
\usepackage{graphicx}
\usepackage{caption}
\usepackage{psfrag}
\usepackage{wrapfig}
\usepackage{cutwin}
\usepackage{pdfpages}
\usepackage{xspace}
\usepackage{paralist}
\tikzstyle{bad state}=[draw,diamond, inner sep=0.1mm, minimum size=5mm]
\tikzstyle{good state}=[draw, rectangle, inner sep=0.1mm, minimum size=4.5mm]
\tikzstyle{good state rand}=[draw, rectangle, double, inner sep=0.1mm, minimum size=4.5mm]
\tikzstyle{stochastic state}=[draw,circle, inner sep=0.1mm, minimum size=5mm]
\tikzstyle{dice state}=[draw,circle, inner sep=0.1mm, minimum size=2mm]
\tikzstyle{target state}=[draw,circle,fill=lightgray, inner sep=0.1mm, minimum size=3mm]
\tikzstyle{dead state}=[draw,circle, inner sep=0mm, minimum size=6mm]
\tikzstyle{arr}=[-latex', rounded corners]

\newcommand{\vecprime}[1]{\vec{#1}^{\prime}}
\newcommand{\vecsup}[2]{\vec{#1}^{#2}}

\newcommand{\achievable}{\mathcal{A}}

\newcommand{\path}{\lambda}
\newcommand{\paths}[1]{\Omega_{{#1}}}

\newcommand{\pdist}[1]{\mathcal{D}(#1)}
\renewcommand{\Pr}{\mathbb{P}}
\newcommand{\Exp}{\mathbb{E}}
\newcommand{\mydef}{\stackrel{\mbox{\rm {\tiny def}}}{=}}
\def\Rset{\mathbb{R}}
\newcommand{\allstratsone}{\Pi}
\newcommand{\stratone}{\ensuremath{\pi}\xspace}
\newcommand{\allstratstwo}{\Sigma}
\newcommand{\strattwo}{\ensuremath{\sigma}\xspace}
\newcommand{\game}{\ensuremath{\mathcal{G}}\xspace}
\newcommand{\states}{S}

\tikzstyle{state1}=[draw, rectangle, inner sep=0mm, minimum size=1mm]
\tikzstyle{state2}=[draw,diamond, inner sep=0mm, minimum size=1.5mm]
\tikzstyle{state3}=[draw,circle, inner sep=0mm, minimum size=1.2mm]
\tikzstyle{state3}=[draw,rectangle, double, inner sep=0mm, minimum size=1mm]

\newcommand{\statesone}{\states_{\Box}}
\newcommand{\statestwo}{\states_{\Diamond}}
\newcommand{\statesprob}{\states_{\bigcirc}}
\newcommand{\actions}{\tfunction}
\newcommand{\tfunction}{\Delta}
\newcommand{\gametuple}{\game = \langle  \states, (\statesone, \statestwo, \statesprob),
			\tfunction \rangle }
\newcommand{\startpara}[1]{{%
\vskip5pt\noindent
{\bf #1.}}}
\newcommand{\dwc}{\textsf{dwc}}
\newcommand{\conv}{\textsf{conv}}
\newcommand{\appref}[1]{\ifthenelse{\isundefined{\techreport}}{\cite{report}}{Appendix~\ref{#1}}}
\definecolor{green}{rgb}{0,.7,0}

\newcommand{\set}[1]{\left\{ #1 \right\}}

\newcommand{\PONE}{\textsf{Player~1}\xspace}
\newcommand{\PTWO}{\textsf{Player~2}\xspace}

\newcommand{\LAST}{\textsf{last}}
\newcommand{\supp}{\textsf{supp}}
\newcommand{\length}{\textsf{len}}
\newcommand{\terminal}{\textsf{Term}}

\newcommand{\bbR}{\mathbb{R}}
\newcommand{\Qf}{\mathbb{R}}
\newcommand\restr[2]{{% we make the whole thing an ordinary symbol
  \left.\kern-\nulldelimiterspace % automatically resize the bar with \right
  #1 % the function
  \vphantom{\big|} % pretend it's a little taller at normal size
  \right|_{#2} % this is the delimiter
  }}
\def\lslope{\textsf{lslope}}
\def\rslope{\textsf{rslope}}

\def\openbegin{(}
\def\openend{)}

\def\adam{\textsf{Adam}\xspace}
\def\eve{\textsf{Eve}\xspace}

\newcommand\point[1]{\noindent\textbf{#1.} }

\newcommand{\thmhelperpre}[2]{\newcommand{\theoremlike}[1]{\par\medskip\penalty-250\refstepcounter{theorem}{\bfseries\noindent##1 \ref{#1}.}\itshape}\theoremlike{#2}}
\newcommand{\thmhelperpost}{%\par\medskip%
 \renewcommand{\theoremlike}[1]{\par\medskip\penalty-250\refstepcounter{theorem}{\bfseries\noindent##1 \thesection .\thetheorem.}\itshape}%
}

\newenvironment{reftheorem}[1]{\thmhelperpre{#1}{Theorem}}{\thmhelperpost}
\newenvironment{reflemma}[1]{\thmhelperpre{#1}{Lemma}}{\thmhelperpost}

\def\df{\delta}
\def\tot{\mathsf{totrew}}
\def\dis{\mathsf{disrew}}
\def\bound{{\vec{z}}}

\def\reward{{\vec{\varrho}}}
\def\singlereward{\varrho}
\newcommand{\firstdim}[1]{#1_1}
\newcommand{\seconddim}[1]{#1_2}

\tikzstyle{player1}=[draw=green!50!black,minimum size=6mm,thick,inner sep=3pt, top color=green!20, bottom color=green!40,draw=green!20!black]
\tikzstyle{player2}=[thick,inner sep=1pt, diamond, top color=red!20, bottom color=red!40,minimum size=6mm,draw=red!20!black]
\tikzstyle{player3}=[thick,inner sep=1pt, circle, minimum size=6mm,top color=blue!20, bottom color=blue!40,draw=blue!20!black]

\tikzstyle{smallplayer1}=[player1,inner sep=0pt,minimum size=4mm]
\tikzstyle{smallplayer2}=[player2,inner sep=0pt,minimum size=3mm]
\tikzstyle{smallplayer3}=[player3,inner sep=0pt,minimum size=4mm]

\title{Decidability Results for Multi-objective Stochastic Games}

\author{Romain Brenguier \and
 Vojt\v{e}ch Forejt}
\institute{Department of Computer Science, University of Oxford, UK}
\begin{document}
\maketitle

\begin{abstract}
  We study stochastic two-player turn-based games in which the objective of one player is to ensure several infinite-horizon total reward objectives, while the other player attempts to spoil at least one of the objectives.
  The games have previously been shown not to be determined, and an approximation algorithm for computing a Pareto curve has been given. The major drawback of the existing
  algorithm is that it needs to compute Pareto curves for finite horizon objectives (for increasing length of the horizon), and the size of these Pareto curves can grow unboundedly, even when the infinite-horizon
  Pareto curve is small.

  By adapting existing results, we first give an algorithm that computes the Pareto curve for determined games. Then, as the main result of the paper, we show that for the natural class
  of stopping games and when there are two reward objectives, the problem of deciding whether a player can ensure satisfaction of the objectives with given thresholds is decidable. The result relies on intricate and novel
  proof which shows that the Pareto curves contain only finitely many points.

As a consequence, we get that the two-objective discounted-reward problem for unrestricted class of stochastic games is decidable.
\end{abstract}

\section{Introduction}
%\vojta{slope: add left/right}

Formal verification is an area of computer science which deals with establishing properties of systems by mathematical
means. Many of the systems that need to be modelled and verified contain controllable decisions, which can be influenced
by a user, and behaviour which is out of the user's control. The latter can be further split into events whose presence can
be quantified, such as failure rate of components, and events which are considered to be completely adversarial, such as
acts of an attacker who wants to break into the system.

%For such models, 
Stochastic turn-based games are used as a modelling formalism for such systems~\cite{DBLP:journals/jcss/ChatterjeeH12}.
Formally, a stochastic game comprises three kinds
of states, owned by one of three players: \PONE, \PTWO, and the stochastic player.
In each state, one or more transitions to successor states are available. 
At the beginning of a play, a token is placed on a distinguished initial
state, and the player who controls it picks a transition and the token is moved to the corresponding
successor state. This is repeated ad infinitum and a path, comprising an infinite sequence of
states, is obtained. \PONE and \PTWO have a free choice of transitions, and the recipe for picking them is called
a strategy. The stochastic player is bound to pick each transition with a fixed probability that is associated with it.

\vskip -0.9pt %This seems to be needed in order to avoid spacing issues in 2nd llncs page
The properties of systems are commonly expressed using rewards, where numbers corresponding to gains or losses are assigned to states
of the system. The numbers along the infinite paths are then summed, giving the total reward of an infinite path, intuitively expressing the energy consumed
or the profit made along a system's execution. Alternatively, the numbers can be summed with a discounting $\delta<1$, giving
discounted reward.
It formalises the fact that immediate gains matter more than future gains, and it is particularly important in economics where money
received early can be invested and yield interest.

Traditionally, the aim of one player is to make sure the expected (discounted) total reward exceeds a given bound, while
the other player tries to ensure the opposite. 
We study the {\em multi-objective problem} in which each state is given a tuple of numbers, for example corresponding
to both the profit made on visiting the state, and the energy spent. Subsequently, we give a bound on both profit and energy, and \PONE
attempts to ensure that the expected total profit and expected total energy exceed (or do not exceed) the given bound, while \PTWO tries to spoil this by making sure
that at least one of the goals is not met.

The problem has been studied in~\cite{CFKSW13}, where it has been shown that Pareto optimal strategies might not exist,
and the game might not be determined (for some bounds neither of the players have $\varepsilon$-optimal strategies).
A value iteration
algorithm has been given for approximating the Pareto curve of the game, i.e. the bounds \PONE can ensure. The algorithm successively computes,
for increasing $n$, the sets of bounds \PONE can ensure if the length of the game is restricted to $n$ steps. The approach has two major
drawbacks. Firstly, the algorithm cannot decide, for given bounds, if \PONE can achieve them. Secondly, it does
not scale well since the representation of the sets can grow with increasing $n$, even if the ultimate Pareto curve is small.

The above limitations show that it is necessary to
design alternative solution approaches.
One of the promising directions is to characterise the shape of the set of achievable bounds, for computing it efficiently.
The value iteration of~\cite{CFKSW13} allows us to show that the sets are convex, but no further observations can be made, in particular
it is not clear whether the sets are convex polyhedra, or if they can have infinitely many extremal points. The main result of our paper shows
that for two-objective case and stopping games, the sets are indeed convex polyhedra, which directly leads to a decision algorithm. We believe
that our proof technique is of interest on its own. It proceeds by assuming that there is an accumulation point on the Pareto curve, and then establishes
that there must be an accumulation point in one of the successor states such that the {\em slope} of the Pareto curves in the accumulation points are equal.
This allows us to obtain a cycle in the graph of the game in which we can ``follow'' the accumulation points and eventually revisit some of them
infinitely many times. By further analysing slopes of points on the Pareto curves that are close to the accumulation point, we show that
there are two points on the curve that are sufficiently far from each other yet have the same slope, which contradicts the assumption that
they are near an accumulation point.

Our results also yield novel important contributions for non-stochastic games. Although there have recently been several works on non-stochastic
games with multiple objectives, they a priori restrict to deterministic strategies, by which the associated problems become fundamentally different. It is easy to
show that enabling randomisation of strategies extends the bounds \PONE can achieve, and indeed, even in other areas of game-theory randomised
strategies have been studied for decades: the fundamental theorem of game theory is that every finite game admits a \emph{randomised} Nash equilibrium~\cite{nash50}.

\startpara{Related work}
In the area of stochastic games, single-objective problems are well studied. For reachability objectives the games are determined
and the problem of existence of an optimal strategy achieving a given value is in NP$\cap$co-NP~\cite{Condon1992}; same holds for total reward objectives.
In the multi-objective setting, \cite{CFKSW13} gives a value iteration algorithm
for the multi-objective total reward problem. Although value iteration converges to the correct result, it does so only in infinite
number of steps. It is further shown in~\cite{CFKSW13} that when \PONE is restricted to only use deterministic strategies, the problem becomes undecidable; the proof relies fundamentally
on the strategies being deterministic and it is not clear how it can be extended to randomised strategies.
The work of \cite{BKTW14} extends the equations of \cite{CFKSW13} to expected energy objectives, and mainly concerns itself
with a variant of multi-objective mean-payoff reward, where the objective is a
``satisfaction objective'' requiring that there is a set of runs of a given probability on which all mean payoff rewards
exceed a given bound (i.e., expected values are not considered). 
\cite{BKTW14} only studies existence of finite-memory strategies and the probability bound~1; this restriction has very recently been lifted by
\cite{Krish16}, which shows that even unrestricted satisfaction objective problem is coNP-complete. 

In non-stochastic games, multi-objective optimisation has been studied for multiple mean-payoff objectives and energy games~\cite{DBLP:journals/iandc/VelnerC0HRR15}. A comprehensive analysis of the complexity of synthesis of optimal strategies has been given~\cite{CRR14}, and it has been shown that a variant of the problem is undecidable~\cite{Velner2015}. The work of~\cite{DBLP:conf/cav/BrenguierR15} studies the complexity of problems related to exact computation of Pareto curves for multiple mean-payoff objectives. In~\cite{DBLP:conf/fsttcs/HunterR14}, interval objectives are studied for total, mean-payoff and discounted reward payoff functions. The problems for interval objectives are a special kind of multi-objective problems that require the payoff to be within a given interval, as opposed to the standard single-objective setting where the goal is to exceed a given bound.
 As mentioned earlier, all the above works for non-stochastic games a priori restrict the players to use deterministic strategies, and hence the problems exhibit completely different properties than the problem we study.

\startpara{Our contribution}
We give the following novel decidability results. Firstly, we show that the problem for {\em determined} stochastic games is decidable.
Then, as the main result of the paper, we show that for non-determined games which also satisfy the stopping assumption and for two objectives,
the set of achievable bounds forms a convex polyhedron.
  This immediately leads to an algorithm for computing Pareto curves, and we obtain the following novel results as corollaries.
  \begin{itemize}
   \item Two-objective discounted-reward problem for stochastic games is decidable.
   \item Two-objective total-reward problem for stochastic stopping games is decidable.
  \end{itemize}
Although we phrase our results in terms of stochastic games, to our best knowledge, the above results also yield
novel decidability results for multi-objective {\em non-stochastic games} when randomisation of strategies is allowed.

\startpara{Outline of the paper}
In Sec.~\ref{sec:determined}, we show a simple algorithm that works for determined games and show how to decide whether a stopping game is determined.
In Sec.~\ref{sec:general}, we give decidability results for two-objective stopping games.

\section{Preliminaries on stochastic games}%{ (taken from MFCS13 paper)}
\label{sec:prelims}
We begin this section by introducing the notation used throughout the paper.
Given a vector $\vec{v}\in \Qf^n$, we use $\vec{v}_i$ to refer to its $i$-th component,
where $1\le i \le n$.
%Note that in dimension~$n=2$, we will often refer to $x$ for $v_1$ and $y$ for $v_2$.
%, and We define the norm $\| \vec{x} \| \mydef \sum^n_{i=1} |x_i|$.
%Given a number $y\in \bbR$, we use $\vec{x} \pm y$ to denote the vector $(x_1 \pm y, x_2 \pm y,\ldots, x_n \pm y)$.
The comparison operator $\le$ on vectors is defined to be the componentwise ordering: $\vec{u} \le \vec{v} \Leftrightarrow \forall i \in [1 , n].\ \vec{u}_i \le \vec{v}_i$.
We write $\vec{u} < \vec{v}$ when $\vec{u} \le \vec{v}$ and $\vec{u} \ne \vec{v}$.
Given two vectors $\vec{u},\vec{v}\in \Qf^n$, the {\em dot product} of $\vec{u}$ and $\vec{v}$ is defined by $\vec{u}\cdot \vec{v} = \sum_{i=1}^n \vec{u}_i\cdot \vec{v}_i$.

The sum of two sets of vectors $U,V\subseteq \Qf^n$ is defined by
$U+V = \{\vec{u}+\vec{v} \,|\, \vec{u}\in U, \vec{v}\in V\}$.
Given a set $V\in \Qf^n$,
we define the \emph{downward closure} of $V$ as $\dwc(V)\mydef\{\vec{u} \mid \exists \vec{v} \in V \,.\, \vec{u}\leq \vec{v}\}$, and
we use $\conv(V)$ for the \emph{convex closure} of $V$, i.e. the set of all $\vec{v}$ for which there are $\vec{v}{}^1,\ldots \vec{v}^{\,n} \in V$ and $w_1\ldots w_n \in [0,1]$
such that $\sum_{i=1}^n w_i = 1$ and $\vec{u} = \sum_{i=1}^n w_i\cdot \vecsup{v}{i}$.
An \emph{extremal point} of a set $X \subseteq \Qf^n$ is a vector $x\in X$ that is not a convex combination of other points in $X$, i.e. $x \not \in \conv(X \setminus \{ x \})$.

A function $f \colon \mathbb{R} \to \mathbb{R}$ is concave whenever for all $x,y\in \mathbb{R}$ and $t\in [0,1]$ we have $f(t\cdot x + (1-t)\cdot y) \ge t\cdot f(x) + (1-t)\cdot f(y)$.
Given $x\in \mathbb{R}$, the \emph{left slope} of $f$ in $x$ is defined by $\lslope(f,x) \mydef \lim_{x' \to x^-} \frac{f(x) - f(x')}{x - x'}$.
  Similarly the \emph{right slope} is defined by $\lim_{x' \to x^+} \frac{f(x) - f(x')}{x - x'}$.
  Note that if $f$ is concave then both limits are well-defined, because by concavity $\frac{f(x) - f(x')}{x - x'}$ is monotonic in $x'$; nevertheless, the left and right slope might still not be equal.

A point $\vec{p}\in \mathbb{R}^2$ is an \emph{accumulation point} of $f$  if $f(\firstdim{\vec{p}}) = \seconddim{\vec{p}}$ and for all $\varepsilon > 0$, there exists $x\neq\firstdim{\vec{p}}$ such that $(x,f(x))$ is an extremal point of $f$ and $| \firstdim{\vec{p}} - x | <\varepsilon$. Moreover, $\vec{p}$ is a \emph{left (right) accumulation point}
if in the above we in addition have $x < \firstdim{\vec{p}}$ (resp. $x > \firstdim{\vec{p}}$).
We sometimes slightly abuse notation by saying that $x$ is an extremal point when $(x,f(x))$ is an extremal point, and similarly for accumulation points.

A \emph{discrete probability distribution} (or just {\em distribution}) over a (countable) set $S$ is
a function $\mu : S \to [0, 1]$ such that $\sum_{s \in S} \mu(s) =  1$.
We write $\pdist{S}$ for the set of all distributions over $S$,
and use $\supp(\mu) = \set{s \in S \mid  \mu(s) > 0}$ for the \emph{support set}
  of $\mu \in \pdist{S}$.

We now define turn-based stochastic two-player
games together with the concepts of strategies and paths of the game.
We then present the objectives that are studied in this paper and the associated decision problems.

%\subsection{Stochastic games} In this section we introduce turn-based stochastic two-player games.
\startpara{Stochastic games}
  A \emph{stochastic (two-player) game} is a tuple $\gametuple$ where
    $\states$ is a finite set of states partitioned into
    sets $\statesone$, $\statestwo$, and $\statesprob$;
    $\tfunction : \states\times\states \rightarrow [0,1]$ is
    a probabilistic transition function
    such that $\tfunction(s,t)\in\{0,1\}$ if $s \in \statesone \cup
    \statestwo$ and
    $\sum_{t\in\states} \tfunction(s,t) = 1$
    if $s\in\statesprob$.

$\statesone$ and $\statestwo$ represent the sets of states controlled
by \PONE{} and \PTWO{}, respectively, while $\statesprob$ is the
set of stochastic states.
For a state $s\in\states$, the set of successor states is
denoted by $\actions(s)\mydef\{t\in\states\mid\tfunction(s,t){>}0\}$.
We assume that $\actions(s) \neq \emptyset$ for all $s \in \states$.
A state from which no other states except for itself are reachable is called
{\em terminal}, and the set of terminal states is denoted by $\terminal \mydef \{s\in\states\mid\tfunction(s){=}\{s\}\}$.

%% \begin{figure}[htb]
%%   \centering
%%   \begin{tikzpicture}[scale=1.2]
%%     \draw (0,0) node[stochastic state] (S0) {$s_0$};
%%     \draw (1,0.8) node[bad state] (S1) {$s_1$};
%%     \draw (1,-0.8) node[good state] (S2) {$s_2$};
%%     \draw (2,0.8) node[good state] (S3) {$s_3$};
%%     \draw (2,-0.8) node[good state] (S4) {$s_4$};
%%     \draw (3,0) node[dead state] (S5) {$s_5$};
%%     \path[-latex'] (S0) edge node[below,pos=0.8]{$\frac{1}{2}$} (S1) (S0) edge node[above,pos=0.8]{$\frac{1}{2}$} (S2) (S1) edge (S3) (S2) edge[bend left] (S3) (S2) edge (S4) (S1) edge[bend right] (S0) (S3) edge(S5) (S4) edge(S5) (S5) edge[loop right] (S5);
%%     \draw (S0.-90) node[below] {$(1,0)$};
%%     \draw (S3.-90) node[below] {$(1,0)$};
%%     \draw (S4.-90) node[below] {$(0,1)$};
%%   \end{tikzpicture}
%%   \caption{Example of a stochastic two-player game. 
%%     Vectors below the states denote a reward function when it is not $(0,0)$. 
%%     Weights on the edges represent the probability of the transition when it is not $1$.
%%   }\label{fig:game}
%% \end{figure}

\startpara{Paths}
An \emph{infinite path} $\path$ of a stochastic game $\game$ is a
sequence $(s_i)_{i \in\mathbb{N}}$ of states such that
$s_{i+1}\in\actions(s_i)$ for all $i \geq 0$.
A \emph{finite path} is a prefix of such a sequence.
For a finite or infinite path $\path$ we write $\length(\path)$ for the number of states in  the path.
For $i < \length(\path)$ we write $\path_i$ to refer to the $i$-th state $s_{i-1}$ of $\path=s_0s_1\ldots$ and $\path_{\le i}$ for the prefix of $\path$ of length $i+1$.
For a finite path $\path$ we write $\LAST(\path)$ for the
last state of the path.
For a game $\game$ we write $\paths{\game}^+$ for the set of
all finite paths, and $\paths{\game}$ for the set of all infinite paths,
and $\paths{\game,s}$ for the set of infinite paths starting in state $s$.
We denote the set of paths that reach a state in
$T\subseteq S$ by $\Diamond T \mydef \{
\path\in\paths{\game} \,|\,
\exists i \,.\, \path_i\in T\}$.

%We say that a distribution ${\mu \in \pdist{S}}$ is a \emph{Dirac distribution} if $\mu(s) = 1$ for some $s \in S$.
%
%We represent a distribution $\mu \in \pdist{S}$ on a set $S = \set{s_1, \dots, s_n}$ as a map $[s_1 \mapsto \mu(s_1), \dots, s_n\mapsto \mu(s_n)]$ and omit the elements of $S$ outside $\supp(\mu)$ to simplify the presentation.
% If the context is clear we sometimes identify a Dirac distribution $\mu$ with the unique element in $\supp(\mu)$.

\startpara{Strategies}
We write $\paths{\game}^\Box$ and $\paths{\game}^\Diamond$ for the finite paths that end with a state of $\statesone$ and $\statestwo$, respectively.
A \emph{strategy} of \PONE{} is a function
$\stratone : \paths{\game}^\Box {\rightarrow} \pdist{S}$
%which is defined for $\path \in \paths{\game}^+$ only if $\LAST(\path) \in \statesone$, 
such that $s\in\supp(\stratone(\path))$ only if $\tfunction(\LAST(\path),s)=1$.
%
%%% NOT SURE WHETHER WE NEED THIS OR NOT:
%A strategy $\stratone$ is a {\em finite-memory} strategy if there is a finite automaton ${\cal A}$ over the alphabet $S$ such that $\stratone(\path)$ is determined by $\LAST(\path)$ and the state of ${\cal A}$ in which it ends after reading the word $\path$.
 We say that $\stratone$ is \emph{memoryless} if $\LAST(\path) {=} \LAST(\path')$ implies $\stratone(\path) {=} \stratone(\path')$, and \emph{deterministic} if $\pi(\path)$ is Dirac for all $\path \in \paths{\game}^+$, i.e. $\pi(\path)(s)=1$ for some $s\in S$.
%It is \emph{persistent}~\cite{MT02} if for each history $\lambda$ compatible with $\stratone$, if there exists an index $i$ such that $\lambda_i = \LAST(\lambda)$ then $\stratone(\lambda) = \stratone(\lambda_{\le i})$.
%
%If $\stratone$ is a memoryless strategy for \PONE{} then we identify it with the mapping $\stratone\colon\statesone\to\pdist{\states}$.
A strategy $\strattwo$ for \PTWO{} is defined similarly replacing $\paths{\game}^\Box$ with $\paths{\game}^\Diamond$.
We denote by $\allstratsone$ and $\allstratstwo$ the sets of all
strategies for \PONE{} and \PTWO{}, respectively.

\startpara{Probability measures}
A stochastic game $\game$, together with  a
strategy pair $(\stratone,\strattwo) \in \allstratsone \times \allstratstwo$
and an initial state $s$, induces an infinite Markov chain on the game~(see e.g. \cite{ChenFKSTU12}).
We denote the probability measure of this Markov chain by
$\Pr^{\stratone,\strattwo}_{\game,s}$.
The expected value of a measurable function
$g\colon\states^\omega{\to}{\bbR_{\pm\infty}}$
is defined as
$
\Exp_{\game,s}^{\stratone,\strattwo}[g] \mydef
\int_{\paths{\game,s}} g\,d\Pr_{\game,s}^{\stratone,\strattwo}$.
We say that a game $\game$ is a {\em stopping game} if, for every strategy pair $(\stratone,\strattwo)$,
a terminal state is reached with probability\ $1$, i.e. $\Pr_{\game,s}^{\stratone,\strattwo}(\Diamond \terminal) = 1$ for all $s$.

%% \startpara{Reachability objectives}\romain{Should we do total payoff as well?}
%% Given a game \game and an initial state~$s$, we say that a strategy~\stratone \emph{achieves} objective $\Diamond T \ge p$ if for all strategy~\strattwo of the adversary $\Pr^{\stratone,\strattwo}_{\game,s}(\Diamond T) \ge p$.
%% Given a tuple of targets $(T_i)_{1 \le i \le n}$, we say that it achieves the vector $(v_i)_{1 \le i \le n}$ if for all index $i$, \stratone achieves $\Diamond T_i \ge v_i$.

\startpara{Total reward}
A reward function $\singlereward \colon \states \to \mathbb{Q}$ assigns a reward  to each state of the game.
We assume the rewards are $0$ in all terminal states.
The \emph{total reward} of a path~$\path$ is $\singlereward(\path) \mydef \sum_{j \ge 0} \singlereward(\path_j)$.
Given a game \game, an initial state~$s$, a vector of $n$ rewards $\reward$ and a vector of $n$ bounds $\bound\in \Qf^n$, we say that a pair of strategies~$(\stratone,\strattwo)$
\emph{yields} an objective $\tot(\reward,\bound)$ if $\Exp^{\stratone,\strattwo}_{\game,s}[\reward_i] \ge \bound_i$ for all $1\le i \le n$.
A strategy $\stratone\in\allstratsone$ {\em achieves} $\tot(\reward,\bound)$ if for all $\strattwo$ we have that
$(\stratone,\strattwo)$ yields $\tot(\reward,\bound)$; the vector $\bound$ is then called {\em achievable}, and we use
$\achievable_s$ for the set of all achievable vectors.
A strategy $\strattwo\in\allstratstwo$ {\em spoils} $\tot(\reward,\bound)$ if for no $\stratone \in \allstratsone$,
the tuple
$(\stratone,\strattwo)$ yields $\tot(\reward,\bound)$.
Note that lower bounds (objectives $\Exp^{\stratone,\strattwo}_{\game,s}[\reward_i] \le \bound_i$) can be modelled by upper bounds after multiplying all rewards and bounds by $-1$.

%The upper Pareto curve in $s$ is the set of all
%minimal $\bound$ such that for all $\varepsilon>0$ there is $\strattwo\in \allstratstwo$ that spoils the objective $\tot(\reward, \bound+\varepsilon)$.
A (lower) Pareto curve in $s$ is the set of all
maximal $\bound$ such that for all $\varepsilon>0$ there is $\stratone\in \allstratsone$ that achieves the objective $\tot(\reward, \bound-\varepsilon)$.
We use $f_s$ for the Pareto curve, and for the two-objective case we treat it as a function, writing $f_s(x)=y$ when $(x,y)\in f_s$.
We say that a game is
{\em determined} if for all states, every bound can be spoiled or lies in the downward closure of the Pareto curve%
\footnote{The reader might notice that in some works, games are said to be determined when each vector can be either achieved
by one player, or spoiled by the other. This is not the case
of our definition, where the notion of determinacy is {\em weaker} and only requires ability to spoil or achieve up to arbitrarily small $\varepsilon$.}. Note that the downward closure of the Pareto curve equals the closure of $\achievable_s$.

% for all strategy~\strattwo of the adversary $\Exp^{\stratone,\strattwo}_{\game,s}[\reward_i] \ge t$.
%We say that it achieves the vector $(v_i)_{1 \le i \le n}$ if for all index $i$, \stratone achieves $\reward_i \ge v_i$.
%Such a vector $v$ is said to be \emph{achievable}.
%It is a \emph{Pareto optimal} (or simply \emph{optimal}) vector if it is achievable and there is no $v' \ge v$ that is achievable and different from $v$.

\startpara{Discounted reward}
Discounted games play an important role in game theory.
In these games, the rewards have a discount factor~$\df \in \openbegin0,1\openend$ meaning that the reward received after $j$ steps is multiplied by $\df^j$, and so a discounted
reward of a path $\lambda$ is then $\singlereward(\path,\df) = \sum_{j \ge 0} \singlereward(\path_j)\cdot \df^j$.
We define the notions of achieving, spoiling and Pareto curves for discounted reward $\dis(\reward,\df,\bound)$ in the same way as for total reward. Since the problems for discounted reward
can easily be encoded using the total reward framework (by adding before each state a stochastic state from which with probability $(1 - \delta)$ we transition to a terminal state),
from now on we will concentrate on total reward, unless specified otherwise.

%% \begin{figure}[htb]
%%   \centering
%%   \begin{tikzpicture}[scale=1.5]
%%     \draw[-latex'] (0,0) -- (1.8,0) node[below] {$x$};
%%     \draw[-latex'] (0,0) -- (0,1.4)  node[left] {$f_{s_0}(x)$};
%%     \draw(0,0) node[below left] {$s_0$};
%%     \draw[blue,thick] (1,1) -- (1.5,0.5);
%%     \draw[dotted] (1,0) node[below] {$1$} -- (1,1);
%%     \draw[dotted] (0,1) node[left] {$1$} -- (1,1);
%%     \draw[dotted] (1.5,0) node[below] {$1.5$} -- (1.5,0.5);
%%     \draw[dotted] (0,0.5) node[left] {$0.5$} -- (1.5,0.5);
    
%%   \end{tikzpicture}\hfill
%%   \begin{tikzpicture}[scale=1.5]
%%     \draw[-latex'] (0,0) -- (1.8,0) node[below] {$x$};
%%     \draw[-latex'] (0,0) -- (0,1.4) node[left] {$f_{s_1}(x)$};
%%     \draw(0,0) node[below left] {$s_1$};
%%     \draw[blue,thick] (1,0) -- (0,1);
%%     \draw (1,0) node[below] {$1$};
%%     \draw (0,1) node[left] {$1$};
%%   \end{tikzpicture}\hfill
%%   \begin{tikzpicture}[scale=1.5]
%%     \draw[-latex'] (0,0) -- (1.8,0) node[below] {$x$};
%%     \draw[-latex'] (0,0) -- (0,1.4) node[left] {$f_{s_2}(x)$};
%%     %\draw(0,0) node[below left] {$s_2$};
%%     \draw (1,0) node[below] {$1$};
%%     \draw (1,0) node[fill=blue,inner sep=2pt] {};
%%   \end{tikzpicture}

%%   \caption{Pareto curves for the game of \figurename~\ref{fig:game} for the states $s_0$, $s_1$ and $s_2$ respectively.}\label{fig:pareto}
%% \end{figure}

\startpara{The problems}
In this paper we study the following decision problems.
\begin{definition}[Total-reward problem]
Given a stochastic game $\game$, an initial state $s_0$, and vectors of reward functions $\reward$
%with targets $T_1$ and $T_2$ 
and thresholds $\bound$, is $\tot(\reward,\bound)$ achievable from $s_0$?
\end{definition}
\begin{definition}[Discounted-reward problem]
Given a stochastic game $\game$, an initial state $s_0$, vectors of reward functions $\reward$
and thresholds $\bound$, and a discount factor $\df \in (0,1)$, is
%with targets $T_1$ and $T_2$ 
$\dis(\reward,\df,\bound)$ achievable from $s_0$?
\end{definition}

\noindent
In the particular case when $n$ above is $2$, we speak about {\em two-objective} problems.

\startpara{Simplifying assumption}
In order to keep the proofs simple, we will assume that each non-terminal state has exactly two successors and that only the states controlled by \PTWO have weights different from $0$.
Note that any stochastic game can be transformed into an equivalent game with this property in polynomial time, so we do not lose generality by this assumption.

\renewcommand\windowpagestuff{%
\begin{minipage}{0.95\textwidth}
%  \vspace{5cm}
  \includegraphics[width=0.95\textwidth,natwidth=541,natheight=311]{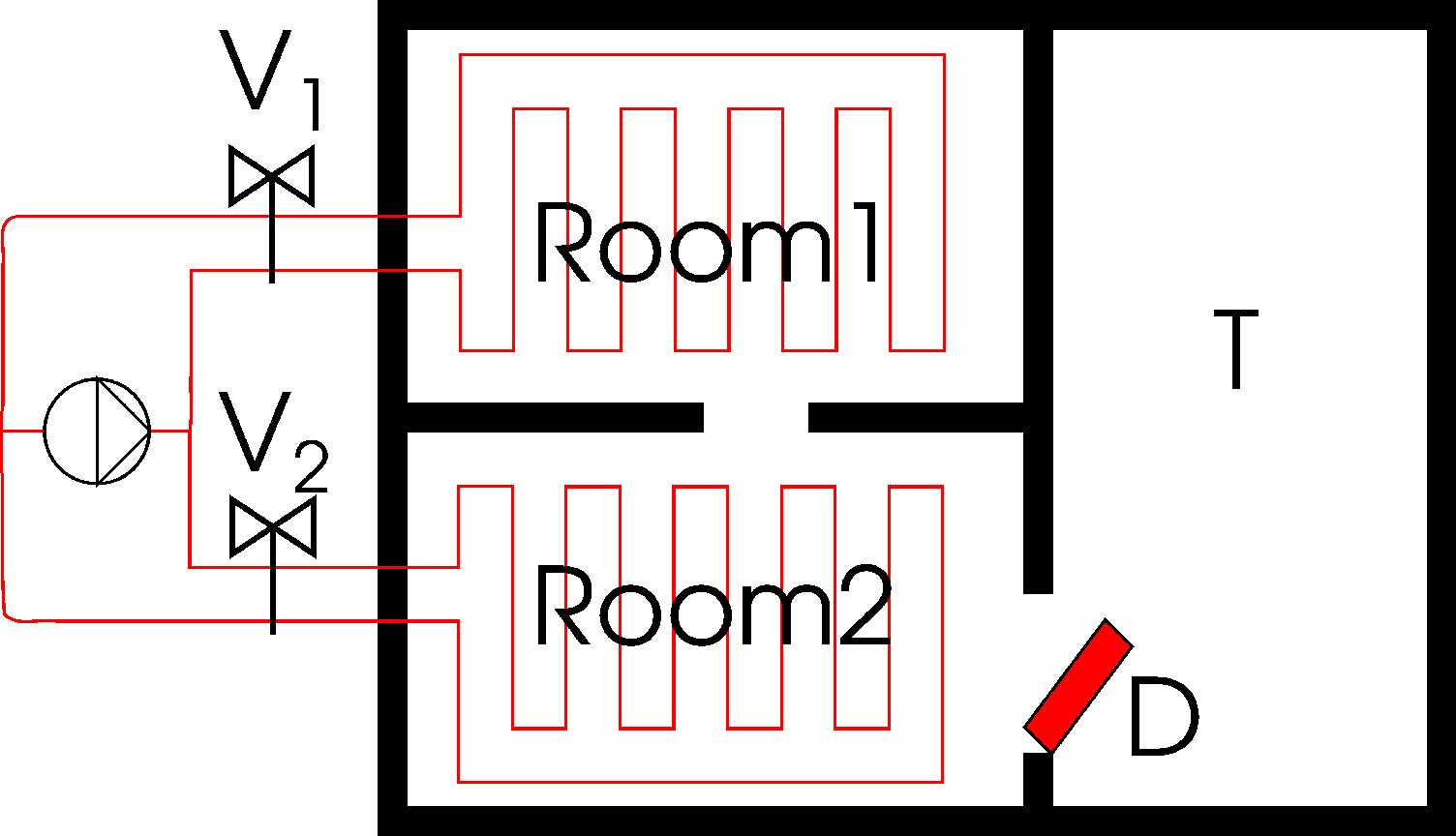}
  \captionof{figure}{A house with controllable floor heating in two rooms.}
  \label{fig:drawing}
%  \vspace{-8mm}
\end{minipage}
}

\medskip
\noindent
{\em Example 3 (Floor heating problem).}
As an example illustrating the definitions, as well as possible applications of our results,
we consider a simplified version of the smart-house case study presented in \cite{LMMST16} with
a difference that we model both user comfort and energy consumption.
\PONE, representing a controller, decides which rooms are heated, while the \PTWO
represents the configuration of the house, for instance which door and windows are open, which cannot be influenced by the controller.
The temperature in another room changes based on additional probabilistic factors.
We illustrate this example in \figurename~\ref{fig:drawing} and a simple model as a stochastic game is given in \figurename~\ref{fig:heating}~(left).
We have to control the floor heating of two rooms in a house, by opening at most one of the valves $V_1$ and $V_2$ at a time.

\vskip 2pt
\opencutleft
\begin{cutout}{3}{0pt}{7cm}{10}%
The state of each room is either cold or hot, for instance in state $H,C$, the first room is warm while the second one is cold, and the third room has unknown temperature.
Weights on the first dimension represent the energy consumption of the system while the second represent the comfort inside the house.
%\end{minipage}
\PTWO controls whether the door $D$ between the second room and a third one is open or not.
The temperature $T$ in the other room of the house is controlled by stochastic transitions.
For instance in the initial state $(C,C)$, the controller can choose either to switch on the heating in room 1 or room 2. 
Then the second player chooses whether the door is opened or not and stochastic states determine the contribution of the other rooms: for instance from $(H,C)$ if the second player chooses that the door is opened then depending on whether the temperature of the other room is low or high, room 2 can either stay cold or get heated through the door, and the next state in that case is $(H,H)$ which is the terminal state.
The objective is to optimise energy consumption and comfort until both rooms are warm.
The Pareto curve for a few states of the game is depicted in \figurename~\ref{fig:heating} (right).
\end{cutout}

%\begin{example}Consider an example stochastic two-player game is given in \figurename~\ref{fig:game}. The game has two associated reward functions,given by the tuples of numbers next to the states. The Pareto curves for the game are depicted in \figurename~\ref{fig:pareto}.\end{example}

\subsection{Equations for lower value}\label{sec:eqns}
We recall the results of \cite{CFKSW13,BKTW14} showing 
that for stopping games the sets of achievable points $\achievable_s$ are the unique solution to the sets of equations defined as follows:
%by three iterator used for approximating the Pareto curve and characterizing it as a fixpoint are of three kinds:
\[
X_s =
\begin{cases}
 \dwc(\{ (0,\dots,0)\} ) & \text{if $s\in \terminal$}\\
 \dwc(\conv(\bigcup_{t\in \tfunction(s)} X_t)) & \text{if $s\in \statesone$}\\
 \reward(s) + \dwc(\bigcap_{t\in \tfunction(s)} X_t) &\text{if $s\in \statestwo$}\\
 \dwc(\sum_{t\in\tfunction(s)} \tfunction(s,t) \cdot X_t) & \text{if $s\in \statesprob$}
\end{cases}
\]

The equations can be used to design a value-iteration algorithm that iteratively computes sets $X^i_s$ for increasing $i$: As a base step we have $X^0_s = \dwc(\vec{0})$ (where $\vec{0}=(0,\ldots,0)$); we then substitute $X^i_s$ for $X_s$ on the right-hand side of the equations, and obtain $X^{i+1}_s$ as $X_s$ on the left-hand side. The sets $X^i_s$ so obtained converge to the least fixpoint of the equations above~\cite{CFKSW13,BKTW14}. As we will show later, the sets $X^i_s$ might be getting increasingly complex even though the actual solution $X_s$ only comprises two extremal points.

\begin{figure}[t!]
\begin{tikzpicture}[x=1.6cm,y=1.3cm,xscale=0.75]
  \draw (-0.1,0) node[player1] (CC1) {C,C}; 
  \draw(CC1.-90) node[below] {$(-1,0)$};
  \draw (1,1) node[player2] (HC2) {H,C};
  \draw (1,-1) node[player2] (CH2) {C,H};
  %\draw (CH2.-90) node[below] {-1,-1};
  \draw (1.7,0.5) node[player3] (HC3) {};
  \draw (1.7,-0.5) node[player3] (CH3) {};
  \draw (4.5,1) node[player3] (HCS3) {};
  \draw (3,-1) node[player1] (CH1) {C,H};
  \draw (4.5,0) node[player2] (D1) {H,H};
  %\draw (3.5,-0.5) node[player2] (D2) {};
  \draw (3,0) node[player1] (HH1) {H,H}; 
  \draw (3,1) node[player1] (HC1) {H,C};
  \draw (4.5,-1) node[player3] (CHS3) {}; 
  \draw[-latex'] (-0.6,0) -- (CC1);
  \draw[-latex'] (CC1) --  (HC2);
  \draw (HC2.-90) node[below] {$(0,-1)$};
  \draw (CH2.-90) node[below] {$(0,-1)$};
  \draw[-latex'] (HC2) -- (HC3);
  \draw[-latex'] (HC2) -- node[above] {\scriptsize{$D$ closed}} (HC1);
  \draw[-latex'] (CC1) -- (CH2);
  \draw[-latex'] (CH2) -- (CH3); 
  \draw[-latex'] (CH2) -- node[below] {\scriptsize{$D$ closed}} (CH1);
  \draw[-latex'] (HC1) --  (HCS3);
  %\draw (HC1.-130) node[below] {$(0,-1)$};
  \draw[-latex'] (HC1)  --  (D1) ;
  %\draw (HH1.-130) node[below] {$(-1,0)$};
  \draw[-latex'] (HCS3) edge[bend right=25] (HC2);
  \draw[-latex'] (HCS3) edge[bend left=10] (HH1.10);
  \draw[-latex'] (HC3) -- node[above, sloped] {\scriptsize{$T$ low}} (HC1);
  \draw[-latex'] (HC3) -- node[above, sloped] {\scriptsize{$T$ high}} (HH1);
  \draw[-latex'] (CH3) -- node[above, sloped] {\scriptsize{$T$ low}}(CC1);
  \draw[-latex'] (CH3) -- node[above, sloped] {\scriptsize{$T$ high}} (CH1);
  \draw[-latex'] (CH1)  -- (D1);
  \draw[-latex'] (D1) -- (HH1);
  %\draw (CHS3.-90) node[below] {$(0,-1)$};
  \draw[-latex'] (CH1) --  (CHS3);
  %\draw (D2) node[below right] {$(-1,0)$};
  \draw (D1.-70) node[below] {$(-1,0)$};
  \draw[-latex'] (CHS3) edge[bend right=10] (HH1.-10);
  \draw[-latex'] (CHS3) edge[bend left=25] (CH2);
  \draw[-latex',rounded corners] (HH1) -- ++(-0.15,-0.55) -- ++(0.3,0) -- (HH1);
\end{tikzpicture}
%\hspace{0.4cm}
\hfill
\begin{tikzpicture}[scale=1.5]
  %\draw (-1,1) node[player1]  {\scriptsize H,C};
  %\draw (-0,1) node[player2]  {\scriptsize H,C};
  \draw[-latex'] (-2.2,0) -- (0.5,0) node[above] {energy};
  \draw[-latex'] (0,-2.2) -- (0,0.5) node[above] {comfort};

  \draw[dotted] (-1,0) -- (-1,-2.2);
  \draw[dotted] (-2,0) -- (-2,-2.2);
  \draw[dotted] (0,-1) -- (-2.2,-1);
  \draw[dotted] (0,-2) -- (-2.2,-2);

%  \draw[-latex'] (-2,0) -- (0.2,0);
  %\draw[-latex'] (0,-2.5) -- (0,0.2);
  \draw(-2,0) %node[inner sep=0pt,minimum size=2mm,fill=green] {} 
  node[above]{$-2$};%{$,0$};

  \draw(0,-1) node[right]{$-1$};
  
%  \draw(-1,-2) node[inner sep=0pt,minimum size=2mm,fill=green] {} node[right]{$-1,-2$};

  %\draw(-1.66,0) node[inner sep=0pt,minimum size=2mm,fill=red] {} node[above]{$-\frac{5}{3}$};%{$,0$};
  %\draw(-1,-1.33) node[inner sep=0pt,minimum size=2mm,fill=red] {} node[right]{$-1,-\frac{4}{3}$};

  \draw(-1,0) node[smallplayer2] {\tiny H,H};
  \draw (-1,0.1)node[above]{$-1$};%{$,0$};
  %\draw(-0.33,-1.33) node[inner sep=0pt,minimum size=2mm,fill=blue] {} node[right]{$-\frac{1}{3},-\frac{4}{3}$};
  %\draw(-1,0) node[inner sep=0pt,minimum size=2mm,fill=blue] {} node[above]{$-1,0$};
  \draw(0,-2) % node[inner sep=0pt,minimum size=2mm,fill=green] {} 
node[right]{$-2$};%{0,-2$};
  \draw[thick,color=green] (-1,0) --node[smallplayer1,color=black] {\tiny H,C}  (0,-1);
  \draw[thick,color=red] (-1,-1) --node[smallplayer2,color=black] {\tiny H,C}  (0,-2);

  %\draw[thick,color=green] (-1.33,-0.33) -- node[player1,black,inner sep=0pt,yshift=-1mm,xshift=-1mm] {\scriptsize C,H}(-1,-0.66);

  \draw (0,0) node[smallplayer1] {\tiny H,H};
  %% \begin{scope}[xshift=2.5cm]
  %% \draw (-0.5,1) node[player1]  {\scriptsize C,H};
  %% \draw[-latex'] (-1.4,0) -- (0.2,0);
  %% \draw[-latex'] (0,-2.5) -- (0,0.2);
  %% \end{scope}

  %% \begin{scope}[xshift=5cm]
  %% \draw (-0.5,1) node[player2]  {C,H};
  %% \draw[-latex'] (-1.4,0) -- (0.2,0);
  %% \draw[-latex'] (0,-2.5) -- (0,0.2);
  %% \end{scope}
  \draw[thick,color=green] (-2,-1) -- node[smallplayer1,color=black] {\tiny C,C} (-1,-2);

  \begin{scope}[xshift=7.5cm]
  %\draw (-1,0) node[smallplayer1,below left] {\tiny C,H};
    %\draw[very thick,color=red] (-1.5,-1.5) -- node[smallplayer2,color=black] {\tiny C,H} (-1,-2) ;
  \end{scope}

\end{tikzpicture}
\caption{A stochastic two-player game modelling the floor heating problem.
   Vectors under states denote a reward function when it is not $(0,0)$. All probabilistic
   transitions have probability $\frac{1}{2}$. Pareto curves of a few states of the game are depicted on the right.
\label{fig:heating}}
\end{figure}
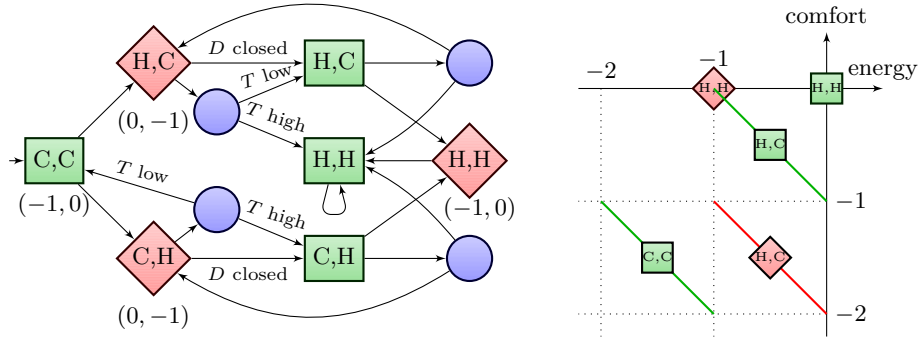

\section{Determined games}\label{sec:determined}
In this section we present a simple algorithm which works under the assumption that the game is determined.
For stopping games, we then give a procedure to decide whether a game is determined.
%We also characterise the complexity of these problems.

\begin{theorem}
  There is an algorithm working in exponential time, which given a determined stochastic two-player game, computes its Pareto-curve.
\end{theorem}

For the proof of the theorem we will make use of the following:
\begin{theorem}[{\cite[Thm.~7]{CFKSW13}}]\label{theorem:spoiler-md}
  Suppose \PTWO has a strategy $\strattwo$ such that for all $\stratone$ of \PONE there is at least one $1\le i\le n$ with
  $\Exp^{\stratone,\strattwo}_{\game,s}(\reward_i) < \bound_i$. Then \PTWO has a memoryless deterministic strategy
  with the same properties.
%  In stochastic two-player games, memoryless deterministic strategies are sufficient for \PONE to achieve a disjunctive objective (i.e. of the form $\bigvee_{i} \Pr^{\stratone,\strattwo}_{\game,s}(\Diamond T_i) < p_i$).
\end{theorem}
From the above theorem we obtain the following lemma.

\begin{lemma}\label{lemma:identify-md} The following two statements are equivalent for determined games:
\begin{itemize}
 \item A given point $\bound$ lies in the downward closure of the Pareto curve for $s$.
 \item For all memoryless deterministic strategies $\strattwo$ of \PTWO, there is a strategy $\stratone$ of \PONE such that $(\stratone,\strattwo)$ yield $\tot(\reward,\bound)$.
\end{itemize}
\end{lemma}
Thus, to compute the Pareto curve for a determined game $\game$, it is sufficient to consider all memoryless deterministic strategies $\strattwo_1,\strattwo_2,\ldots,\strattwo_m$ of \PTWO and use~\cite{EKVY08} to compute the Pareto curves $f^{\strattwo_i}_s$ for the games $\game^{\strattwo_i}$
induced by $\game$ and $\strattwo_i$ (i.e. $\game^{\strattwo_i}$ is obtained from $\game$ by turning all $s\in\statestwo$ to stochastic vertices and stipulating $\tfunction(s,t) = \strattwo_i(s)$ for all successors $t$ of $s$; in turn, $\game^{\strattwo_i}$ is a Markov decision process), and obtain the Pareto curve for $\game$ as the pointwise minimum
$V_s:= \min_{1\le i\le m} f^{\strattwo_i}_s$.

To decide if a stopping game is determined, it is sufficient to take the downward closures of solutions $V_s$
and check if they satisfy the equations from Sec.~\ref{sec:eqns}. Since in stopping games the solution of the equations is unique,
if the sets are a solution they are also the Pareto curves and the game is determined. If any of the
equations are not satisfied, then $V_s$ are not the Pareto curves and the game is not determined.
Note that for non-stopping games the above approach does not work: even if the sets do not change by applying one step of value iteration,
it is still possible that the solution is not the least fixpoint, and so we cannot infer any conclusion.

\section{Games with two objectives}
\label{sec:general}

We start this section by showing that the existing value iteration algorithm presented in Sec.~\ref{sec:eqns}
might iteratively compute sets $X^i_s$ with increasing number of extremal points,
although the actual resulting set $X_s$ (and the associated Pareto curve $f_s$) is very simple. Consider the game from Fig.~\ref{fig:inf} (left).
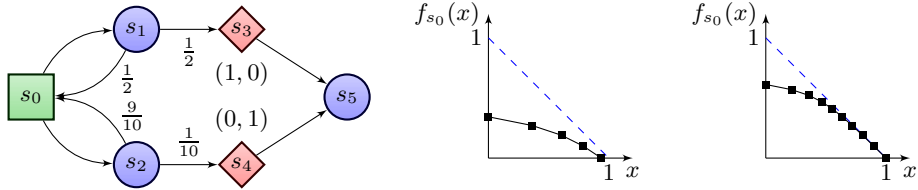
\begin{figure}[t]
  \begin{tikzpicture}[x=1.4cm,y=1.2cm]
    \draw (0,0) node[player1] (S0) {$s_0$};
    \draw (1,0.75) node[player3] (S1) {$s_1$};
    \draw (1,-0.75) node[player3] (S2) {$s_2$};
    \draw (2,0.75) node[player2] (S3) {$s_3$};
    \draw (2,-0.75) node[player2] (S4) {$s_4$};
    \draw (3,0) node[player3] (S5) {$s_5$};
    \path[-latex'] (S0) edge[bend left] (S1)
                   (S0) edge[bend right] (S2)
		   (S1) edge[bend left] node[right,pos=0.4]{\phantom{.}$\frac{1}{2}$} (S0)
		   (S2) edge[bend right] node[right,pos=0.4]{\phantom{.}$\frac{9}{10}$} (S0)
		   (S1) edge node[below]{$\frac{1}{2}$} (S3)
		   (S2) edge node[above]{$\frac{1}{10}$} (S4)
		   (S3) edge (S5)
		   (S4) edge (S5);
    \draw (S3.-90) node[below] {$(1,0)$};
    \draw (S4.90) node[above] {$(0,1)$};
  \end{tikzpicture}
\hfill
  \begin{tikzpicture}[scale=1.6]
    \draw[-latex'] (0,0) -- (1.2,0) node[below] {$x$};
    \draw[-latex'] (0,0) -- (0,1.2)  node[left] {$f_{s_0}(x)$};
    \draw[blue,dashed] (0,1) -- (1,0);
    \foreach \iter in {0,...,4} {
	\pgfmathsetmacro{\xx}{(9/10)^(\iter)*(1-(2/1)^(-(4-\iter))};
	\pgfmathsetmacro{\yy}{1-(10/9)^(-\iter)};
    	\draw (\xx,\yy) node [fill,inner sep=0mm,minimum size=1mm] (n\iter) {};
    }
    \foreach \iter in {1,...,4} {
	\pgfmathtruncatemacro{\itm}{\iter-1};
    	\path [-] (n\itm.center) edge (n\iter.center);
    }
    \node at (1,-0.12) {$1$};
    \node at (-0.1,1) {$1$};
  \end{tikzpicture}
\hfill
  \begin{tikzpicture}[scale=1.6]
    \draw[-latex'] (0,0) -- (1.2,0) node[below] {$x$};
    \draw[-latex'] (0,0) -- (0,1.2)  node[left] {$f_{s_0}(x)$};
    \draw[blue,dashed] (0,1) -- (1,0);
    \foreach \iter in {0,...,9} {
	\pgfmathsetmacro{\xx}{(9/10)^(\iter)*(1-(2/1)^(-(9-\iter))};
	\pgfmathsetmacro{\yy}{1-(10/9)^(-\iter)};
    	\draw (\xx,\yy) node [fill,inner sep=0mm,minimum size=1mm] (n\iter) {};
    }
    \foreach \iter in {1,...,9} {
	\pgfmathtruncatemacro{\itm}{\iter-1};
    	\path [-] (n\itm.center) edge (n\iter.center);
    }
    \node at (1,-0.12) {$1$};
    \node at (-0.1,1) {$1$};
  \end{tikzpicture}
\caption{An example showing that value iteration might produce Pareto curves with unboundedly many extremal points.\label{fig:inf}}
\end{figure}
Applying the value-iteration algorithm given by the equations from Sec.~\ref{sec:eqns} for $n$ steps gives a Pareto curve in $s_0$ with
$n-1$ extremal points. Each extremal point corresponds to a strategy $\stratone_i$ that in $s_0$ chooses to go to $s_2$ when the number of visits of $s_0$ is less than $i$, and after
that chooses to go to $s_1$. The upper bounds of the sets $X^n_s$ for $n=5$ and $n=10$ are drawn in Fig.~\ref{fig:inf} (centre and right, respectively) using solid line, and their
extremal points are marked with dots. The Pareto curve $f_s$ is drawn with dashed blue line, and it consists of two extremal points, $(0,1)$ and $(1,0)$.

We now proceed with the main result of this section, the decidability of the two-objective strategy synthesis problem for
stopping games. The result can be obtained from the following theorem.

\begin{theorem}\label{thm:stopping-no-accum}
  If \game is a stopping stochastic two-player game with two objectives, and $s$ a state of \game then the Pareto curve $f_s$ has only finitely many extremal points.
\end{theorem}

The above theorem can be used to design the following algorithm. For a fixed number $k$, we create
a formula $\varphi_k$ over $(\Rset,+,\cdot,\le)$ which is true if and only if for each $s\in S$ there are points
$\vecsup{p}{s,1},\ldots, \vecsup{p}{s,k}$ such that the sets $V_s \mydef \dwc(\conv(\{\vecsup{p}{s,1},\ldots \vecsup{p}{s,k}\}))$
satisfy the equations from Sec.~\ref{sec:eqns}. Using~\cite{Tarski:reals-arithmetic} we can then successively check validity of $\varphi_k$ for increasing $k$,
and Thm.~\ref{thm:stopping-no-accum} guarantees that we will eventually get a formula which is valid, and it immediately gives us the Pareto curve.
We get the following result as a corollary.
\begin{corollary}
 Two-objective total reward problem is decidable for stopping stochastic games, and two-objective discounted-reward problem is decidable for stochastic games.
\end{corollary}

\startpara{Outline of the proof of Thm.~\ref{thm:stopping-no-accum}}
The proof of Thm.~\ref{thm:stopping-no-accum} proceeds by assuming that there are infinitely many extremal points on the Pareto curve, and then deriving a contradiction.
Firstly, because the game is stopping, an upper bound on the expected total reward that can be obtained with respect to a single total reward objective is $M:=\sum_{i=0}^\infty (1-p_{min}^{|S|}) \cdot \singlereward_{max}^{|S|}$ where $p_{min} = \min \{\Delta(s,s') \mid \Delta(s,s')>0\}$ is the smallest transition probability, and
$\singlereward_{max} = \max_{i\in\{1,2\}} \max_{s\in S} \reward_i(s)$ is the maximal reward assigned to a state. Thus, the Pareto curve
is contained in a compact set, and this implies that there is an accumulation point on it.
In Sec.~\ref{sec:mapping}, we show that we can follow one accumulation point~$\vec{p}$ from one state to one of its successors, while preserving the same left slope. 
Moreover, in the neighbourhood of the accumulation point the rate at which the right slope decreases is quite similar to the decrease in the successors, in a way that is made precise in Lem.~\ref{lem:slope-convex-union}, \ref{lem:slope-intersection}, and \ref{lem:accumulation-stochastic-two}.
This is with the exception of some stochastic states for which the decrease strictly slows down when going to the successors: we will exploit this fact to get a contradiction.
We construct a transition system $T_{s_0,\vec{p}}$, which keeps all the paths obtained by following the accumulation point $\vec{p}$ from $s_0$.
We show that if \game is a stopping game, then we can obtain a path in $T_{s_0,\vec{p}}$ which visits stochastic states for which the decrease of the right slope strictly slows down.
This relies on results for {\em inverse betting games}, which are presented in Sec.~\ref{subsec:inverse-betting}.
Since this decrease can be repeated and there are only finitely many reachable states in $T_{s_0,\vec{p}}$, we show in Sec.~\ref{sec:contra} that the decrease of the right slope must be zero somewhere, meaning that the curve is constant in the neighbourhood of an accumulation point, which is a contradiction.

We will rely on the properties of the equations from Sec.~\ref{sec:eqns} and the left and right slopes of
the Pareto curve. Note that we introduced the notion of slope only for two-dimensional sets,
and so our proofs only work for two dimensions. 
Generalisations of the concept of slopes exist for higher dimensions, but simple generalisation of our lemmas would not be valid,
as we will show later.
Hence, in the remainder of this section, we focus on the two-objective case. For the simplicity of presentation, we will present
all claims and proofs for {\em left} accumulation points. The case of right accumulation points is analogous.

\subsection{Mapping accumulation points to successor states}\label{sec:mapping}
We start by enumerating some basic but useful properties of the Pareto curve and its slopes. 
First notice that it is a continuous concave function and we can prove the following:
\begin{lemma}\label{lem:general-lemma-bis}
  Let $f$ be a continuous concave function defined on $[a,b]$.
  \begin{enumerate}
  \item If $a < x < x' \le b$ are two reals for which $\lslope(f)$ is defined, 
      then $\lslope(f,x) \ge \rslope(f,x) \ge \lslope(f,x')$.
    \label{lem:slope-decrease-bis}
    %This is illustrated in \figurename~\ref{fig:general-lemma}.
  \item
    If $\openbegin x, x' \openend$ contains an extremal point of $f$ then $\lslope(f_s,x) \ne \lslope(f_s,x')$.\label{lem:extremal-slopes-bis}  
    %$(x_1,y_1)$ with $x_1 \in \openbegin x,x'\openend$) then $\lslope(f_s,x) \ne \lslope(f_s,x')$.\label{lem:extremal-slopes}  
  \item If $x\in \openbegin a ,b ]$, then $\lim_{x' \to x^-} \lslope(f,x') = \lim_{x' \to x^-} \rslope(f,x') = \lslope(f,x)$. \label{lem:limit-slope-bis}
  \end{enumerate}
\end{lemma}

To prove Thm.~\ref{thm:stopping-no-accum}, we will use the equations from
Sec.~\ref{sec:eqns} to describe how accumulation points on a Pareto curve for $s$ ``map'' to accumulation
points on successors. 
%The reader can refer to examples in Fig.~\ref{fig:succ} to reinforce the intuition behind the following lemmas.
\begin{figure}[t]
    \begin{tikzpicture}[scale=2]
      \draw[-latex'] (0,0) -- (1.3,0);
      \draw[-latex'] (0,0) -- (0,1.3);
      \draw[blue!20,very thick] (0,1) arc (90:45:0.5) coordinate (bla);
      \draw[blue!20,very thick] (1.087,0) arc(0:42:0.31) -- (bla);
      \draw[dotted] (0,1) arc (90:10:0.5) -- (0.6,0);
      \draw[dashed] (0,0.6) arc (90:50:1.5) arc (53:0:0.31);
      \draw (0.2,0.955) node [fill,inner sep=0mm,minimum size=1mm] {} node[above] {$\vec{p}$};
    \end{tikzpicture}
\hfill
    \begin{tikzpicture}[scale=2]
      \draw[-latex'] (0,0) -- (1.3,0);
      \draw[-latex'] (0,0) -- (0,1.3);
      \draw[blue!20,very thick] (0,0.6) arc (90:70.2:1.5) -- (0.6,0);
      \draw[dashed] (0,1) arc (90:10:0.5) -- (0.6,0);
      \draw[dotted] (0,0.6) arc (90:50:1.5) arc (53:0:0.31);
      \draw (0.2,0.582) node [fill,inner sep=0mm,minimum size=1mm] {} node[above] {$\vec{p}$};
      %\draw[very thick,dashed,green] (0,0.6) -- (0.1,0.6) -- (0.5,0.5) -- node[left] {$s_0$} (0.6,0.1) -- (0.62,0);
    \end{tikzpicture}
\hfill
    \begin{tikzpicture}[scale=2]
      \draw[-latex'] (0,0) -- (1.3,0);
      \draw[-latex'] (0,0) -- (0,1.3);

      \draw[blue!20,very thick] (0,0.95) arc (90:45:0.35cm)
      -- (0.85,0.25) arc (45:0:0.35cm) -- (0.95,0);
      \draw[dashed] (0,1.2) -- (1.2,0) ;
      \draw[dotted] (0.7,0) arc (0:90:0.7cm);
      \draw (1.2,0) node [fill,inner sep=0mm,minimum size=1mm] {} node[below] {$\vec{r}$};
      \draw (30:0.7) node [fill,inner sep=0mm,minimum size=1mm] {} node[below] {$\vec{q}$};
      \draw (0.6,0) + (30:0.35) node [fill,inner sep=0mm,minimum size=1mm] {} node[below] {$\vec{p}$};
      %% \draw (0.2,1) node [fill,inner sep=0mm,minimum size=1mm] {} node[right] {$\vec{r}$};
      %% \draw (0.495,0.495) node [fill,inner sep=0mm,minimum size=1mm] {} node[below] {$\vec{q}$};
      %%\draw (0.347,0.748) node [fill,inner sep=0mm,minimum size=1mm] {} node[below] {$\vec{p}$};
    \end{tikzpicture}
\hfill
    \begin{tikzpicture}[scale=2]
      \draw[blue!20,very thick] (0.2,0.4) -- (0.6,0.4);
      \node at (0,0.4) {$s_0$};
      \draw[dotted] (0.2,0.2) -- (0.6,0.2);
      \node at (0,0.2) {$s_1$};
      \draw[dashed] (0.2,0) -- (0.6,0);
      \node at (0,0) {$s_2$};
    \end{tikzpicture}
\caption{An example of Pareto curve in a state $s_0$ with two successors $s_1$ and $s_2$, for the case of $s_0\in \statesone$ (left), $s_0\in \statestwo$ (centre), and $s_0\in \statesprob$ with uniform probabilities on transitions (right). In each case, the curve in $s_0$ has infinitely many accumulation points.\label{fig:succ}}
\end{figure}
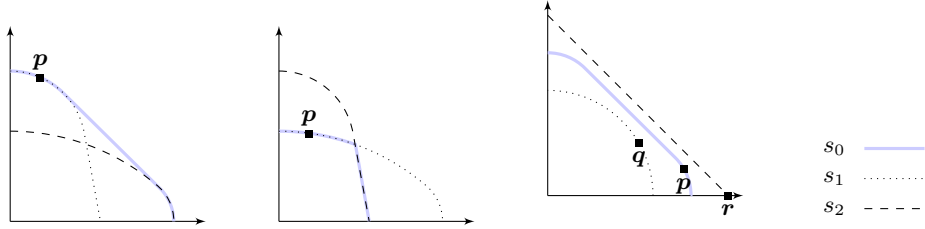

\begin{lemma}\label{lem:slope-convex-union}
  Let $s_0$ be a \PONE{} state with two successors $s_1$ and $s_2$, and let $\vec{p}$ be a left \mbox{accumulation} point of $f_{s_0}$.
  Then there is  $\eta(s_0,\vec{p}) > 0$ such that
  for all $\varepsilon \in (0, \eta(s_0,\vec{p}\ )\openend$,
  %for all $x' \in [\firstdim{\vec{p}}- \eta(s_0,\vec{p}),\firstdim{\vec{p}}]$,
  there is $s' \in \{ s_1, s_2 \}$ such that:
  \begin{inparaenum}
  \item $\vec{p}$ is a left accumulation point in $f_{s'}$; \label{lem:slope-convex-union-a}
  \item $\lslope(s_0,\firstdim{\vec{p}}) = \lslope(s',\firstdim{\vec{p}})$; \label{lem:slope-convex-union-b}
  \item $f_{s_0}(\firstdim{\vec{p}} - \varepsilon) \ge f_{s'}(\firstdim{\vec{p}} - \varepsilon)$ and $\rslope(f_{s_0},\firstdim{\vec{p}} - \varepsilon)  \ge \rslope(f_{s'},\firstdim{\vec{p}} - \varepsilon)$. \label{lem:slope-convex-union-c}
  \end{inparaenum}
\end{lemma}
\begin{proof}[Sketch]%[Proof (Sketch).]
The point
\begin{inparaenum}
 \item follows from the fact that every extremal point in the Pareto curve for $s_0$ must be an extremal point in one of the successors. 
   This is illustrated in \figurename~\ref{fig:succ}~(left): $p$ which is an extremal point for $s_0$ is also an extremal point for $s_1$.
The point
 \item follows because from a sequence of extremal points $(\vecsup{p}{i})_{i\ge 0}$ on the Pareto curve of $s_0$ that converge to $\vec{p}$, we can select a 
  subsequence that gives extremal points on $s'$ that converge to the left accumulation point $\vec{p}$ on
  $s'$. Finally, to prove
 \item we use the fact that the right slope of $f_{s_0}$ is always between those of $f_{s_1}$ and of $f_{s_2}$.
   %TODO
\end{inparaenum}
\qed
\end{proof}

\begin{lemma}\label{lem:slope-intersection}
  %there is $0 < \varepsilon < \eta$ such that $\lslope(s,x-\varepsilon) \ge \lslope(s',x-\varepsilon)$ and  $0 < \varepsilon' < \eta$ such that $\lslope(s,x-\varepsilon') \le \lslope(s',x-\varepsilon')$.
  Let $s_0$ be a \PTWO{} state with two successors $s_1$ and $s_2$, and let $\vec{p}$ be a left \mbox{accumulation} point of $f_{s_0}$.
  There is $\eta(s_0,\vec{p}) > 0$ such that for all $\varepsilon \in (0,\eta(s_0,\vec{p})\ \openend$, there is $s' \in \{s_1,s_2\}$, such that:
  \begin{inparaenum}
  \item $\vec{p} - \reward(s_0)$ is a left accumulation point in $f_{s'}$; \label{lem:slope-intersection-a}
  \item $\lslope(s_0,\firstdim{\vec{p}}) = \lslope(s',\firstdim{\vec{p}}-\reward_1(s_0))$; \label{lem:slope-intersection-b}
  \item  $f_{s_0}(\firstdim{\vec{p}}-\varepsilon) = f_{s'}(\firstdim{\vec{p}}-\varepsilon-\reward_1(s_0))$ and $\rslope(f_{s_0},\firstdim{\vec{p}}-\varepsilon) = \rslope(f_{s'},\firstdim{\vec{p}}-\varepsilon-\reward_1(s_0))$.\label{lem:slope-intersection-c}
  %\item $f_{s_0} \equivx f_{s'}$. \label{lem:slope-intersection-c}
  \end{inparaenum}
\end{lemma}
%\review{in the statement of lemma, 2. $x$ should be $p_1$}\vojta{what about $x$ in the proof?}
\begin{proof}[Sketch]%[Proof (Sketch).]
A crucial observation here is that $f_{s_0}(\firstdim{\vecsup{p}{i}})$ is either $\reward_2(s_0) + f_{s_1}(\firstdim{\vecsup{p}{i}}-\reward_1(s_0))$ or $\reward_2(s_0) + f_{s_2}(\firstdim{\vecsup{p}{i}}-\reward_1(s_0))$.
This is illustrated in \figurename~\ref{fig:succ}~(center): $f_{s_0}(\firstdim{\vec{p}}) = \reward_2(s_0) + f_{s_1}(\firstdim{\vec{p}}-\reward_1(s_0))$ (there $\reward(s_0) = (0,0)$).
Hence when we take a sequence $(\firstdim{\vecsup{p}{i}})_{i\in \mathbb{N}}$, for some $\ell\in \{1,2\}$ the value $f_{s_0}(\firstdim{\vecsup{p}{i}})$ equals $\reward_2(s_0) + f_{s_\ell}(\firstdim{\vecsup{p}{i}}-\reward_1(s_0))$ infinitely many times. From this we get a converging sequence of points in $s_\ell$,
and obtain that the left slopes are equal in $s_0$ and $s_\ell$. By further arguing that in any
left neighbourhood of $\firstdim{\vecsup{p}{i}}-\reward_1(s_0)$ we can find infinitely many points with different left slopes, we obtain that there are also infinitely many extremal
points in the neighbourhood and hence $\firstdim{\vecsup{p}{i}}-\reward_1(s_0)$ is a left accumulation point.

As for the last item, the important observation here is that if at some point $\vec{p}'$, $f_{s_1}$ is strictly below $f_{s_2}$ then the right slope of $f_{s_0}$ corresponds to that of $f_{s_1}$, and if $f_{s_1}$ equals $f_{s_2}$ then the right slope of $f_{s_0}$ corresponds to the minimum of the right slopes of $f_{s_1}$ and $f_{s_2}$ (it is also interesting to note that the left slope corresponds to the maximum of the two).
\qed
\end{proof}

\begin{lemma}\label{lem:accumulation-stochastic-two}
  Let $s_0$ be a stochastic state with two successors $s_1$ and $s_2$, and $\vec{p}$ a left accumulation point of $f_{s_0}$.
  There are points $\vec{q}$ and $\vec{r}$ on $f_{s_1}$ and $f_{s_2}$ respectively such that $\vec{p} = \tfunction(s_0,s_1) \cdot \vec{q} + \tfunction(s_0,s_2) \cdot \vec{r}$.
  Moreover:
  \begin{enumerate}
  \item there is $(s',\vec{t}) \in \{ (s_1,\vec{q}), (s_2,\vec{r}) \}$ such that $\vec{t}$ is a left accumulation point of $f_{s'}$ and
    $\lslope(f_{s_0},\firstdim{\vec{p}}) = \lslope(f_{s'},\firstdim{\vec{t}})$;\label{lem:accumulation-stochastic-two-lap}
  \item there is $\eta(s_0,\vec{p}) > 0$ such that for all $\varepsilon \in (0,\eta(s_0,\vec{p}\ ) \openend$:
    \label{lem:accumulation-stochastic-two-eps}
    \begin{itemize}
      \item there are $\varepsilon_1 \ge 0, \varepsilon_2 \ge 0$ such that $\rslope(f_{s_0},\firstdim{\vec{p}} - \varepsilon) \ge \rslope(f_{s_1},\firstdim{\vec{q}} - \varepsilon_1)$, $\rslope(f_{s_0},\firstdim{\vec{p}} - \varepsilon) \ge \rslope(f_{s_2},\firstdim{\vec{r}} - \varepsilon_2)$, and $\varepsilon = \tfunction(s_0,s_1) \cdot \varepsilon_1 + \tfunction(s_0,s_2) \cdot \varepsilon_2$;
    %$\rslope(f_{s_0},\firstdim{\vec{p}} - \varepsilon) \ge \rslope(f_{s'},\firstdim{\vecprime{p}} - \varepsilon)$;
      \item if $\vec{r}$ is not a left accumulation point in $f_{s_2}$, or $\lslope(f_{s_0},\firstdim{\vec{p}}) \ne \lslope(f_{s_2},\firstdim{\vec{r}})$, then 
        $f_{s_0}(\firstdim{\vec{p}} - \varepsilon) = \tfunction(s_0,s_1)\cdot f_{s_1}\left(\frac{\firstdim{\vec{p}} - \varepsilon - \tfunction(s_0,s_2) \cdot \firstdim{\vec{r}}}{\tfunction(s_0,s_1)}\right) + \tfunction(s_0,s_2) \cdot \seconddim{\vec{r}}$;
      \item symmetrically, if $\vec{q}$ is not a left accumulation point in $f_{s_1}$, or $\lslope(f_{s_0},\firstdim{\vec{p}}) \ne \lslope(f_{s_1},\firstdim{\vec{q}})$, then 
        $f_{s_0}(\firstdim{\vec{p}} - \varepsilon) = \tfunction(s_0,s_1)\cdot \seconddim{\vec{q}} + \tfunction(s_0,s_2) \cdot f_{s_1}\left(\frac{\firstdim{\vec{p}} - \varepsilon - \tfunction(s_0,s_1) \cdot \firstdim{\vec{q}}}{\tfunction(s_0,s_2)}\right)$.
    \end{itemize}
  \end{enumerate}
\end{lemma}
%\review{Statement of lemma. Why is p an extremal point? Is that also part of the definition of accumulation points}
%\romain{I don't see were it appears in the statement but only in the proof and I removed the sentence  where we were mentioning this, which seemed unecessary}
\begin{proof}[Sketch]%[Proof (Sketch).]
We use the fact that for every extremal point $\vecprime{p}$ there are unique extremal points $\vecprime{q}$ and $\vecprime{r}$ on $f_{s_1}$ and $f_{s_2}$, respectively, such that
$\vecprime{p} = \tfunction(s_0,s_1) \cdot \vecprime{q} + \tfunction(s_0,s_2) \cdot \vecprime{r}$. 
%Because $\vec{p}$ is also extremal, we immediately get the first item of the lemma.

To prove item \ref{lem:accumulation-stochastic-two-lap}, we show that for all extremal point $\vecprime{p}$, $\lslope(s_0,\vecprime{p}) = \min (\lslope(s_1,\vecprime{q}), \lslope(s_2,\vecprime{r}))$, which can be surprising at first glance since one could have expected a weighted sum of the left slopes.
This fact is illustrated in \figurename~\ref{fig:succ}~(right): $\lslope(s_0,\vecprime{p}) = \lslope(s_1,\vecprime{q}) \le \lslope(s_2,\vecprime{r})$.
The inequality $\lslope(s_0,\vec{p}) \le \lslope(s_1,\vec{q})$ (and similarly $\lslope(s_0,\vec{p}) \le \lslope(s_2,\vec{r})$), follows from concavity of $f_{s_0}$: because for all
$\varepsilon>0$ the inequality $f_{s_0}(\firstdim{\vec{p}} - \varepsilon) \ge \tfunction(s_0,s_1) \cdot f_{s_1}(\firstdim{\vec{q}} - \frac{\varepsilon}{\tfunction(s_0,s_1)}) + \tfunction(s_0,s_2) \cdot f_{s_2}(\firstdim{\vec{r}})$
holds true, from which we obtain $\lim_{\varepsilon \to 0^+} \frac{f_{s_0}(\firstdim{\vec{p}})- f_{s_0}(\firstdim{\vec{p}} - \varepsilon)}{\varepsilon} \le \lim_{\varepsilon \to 0^+} \frac{f_{s_1}(\firstdim{\vec{q}}) - f_{s_1}(\firstdim{\vec{q}} - \varepsilon)}{\varepsilon}$. Showing that the left slope is at least the minimum of the successors' slopes is significantly more demanding and technical, and we give the proof in the appendix.

%For item 2., we consider a sequence $\vecsup{p}{i}$ of extremal points converging to $\vecsup{p}{i}$, and by the above characterisation of extremal points get that they can be expressed as $\vecsup{p}{i} = \tfunction(s_0,s_1) \cdot \vecsup{q}{i} + \tfunction(s_0,s_2) \cdot \vecsup{r}{i}$ for appropriately chosen $\vecsup{q}{i}$ and $\vecsup{r}{i}$. Then it is sufficient to observe that either the sequence $\vecsup{q}{i}$ or $\vecsup{r}{i}$ consists of infinitely many points, from which we can take a subsequence converging to $\vec{q}$ or $\vec{r}$.
Proving the second point, is based on the observation that a point on the Pareto curve $f_{s_0}$ is a combination of points of $f_{s_1}$ and $f_{s_2}$ that share a common tangent: in other words they maximize the dot product with a specific vector on their respective curves.
From this observation it is possible to link the right slopes of these curves.

The last two points hold because with the assumption, extremal points that converge to $\vec{p}$ from the left can be obtained as a combination from a fixed $\vec{r}$ and points on $f_{s_2}$.
\qed
\end{proof}
%Note that the second item in Lemma~\ref{lem:accumulation-stochastic-two} might appear to be counter-intuitive, because one might expect the left slope to be strictly between 
%the left slopes of the successors. 

Now we will prove that there are no left accumulation points on the Pareto curve.
To do that, we will try to follow one in the game: if there is a left accumulation point in one state then at least one of its successors also has one, as the above lemmas show.
By using the fact that the left slopes of left accumulation points are preserved we show that the number of reachable combinations $(s, \vec{p})$, where $s\in \states$ and $\vec{p}$ is a left accumulation point, is finite.
We then look at points slightly to the left of the accumulation points, their distance to the accumulation point and right slopes are also mostly preserved except in stochastic states, where if only one successor has a left accumulation point, the decrease of the right slope accelerate (by Lem.~\ref{lem:accumulation-stochastic-two}.\ref{lem:accumulation-stochastic-two-eps}).
By using the fact that in stopping games we can ensure visiting such stochastic states, we will show that for some states the right slope is constant on the left neighbourhood of the left accumulation point, which is a contradiction.

Assume we are given a state $s_0$ and a left accumulation point $\vecsup{p}{0}$ of $f_{s_0}$.
We construct a transition system $T_{s_0,\vecsup{p}{0}}$ where the initial state is $(s_0,\vecsup{p}{0})$, and the successors of a given configuration $(s,\vec{p})$ are the states $(s',\vecprime{p})$ such that $s'$ is a successor of $s$, and $\vecprime{p}$ is a left accumulation point of $s$ with the same left slope on $f_{s'}$ as $\vec{p}$ on $f_{s}$.
Lem.~\ref{lem:slope-convex-union}, \ref{lem:slope-intersection}, and \ref{lem:accumulation-stochastic-two}, ensure that all the reachable states have at least one successor.

\begin{lemma}\label{lem:finite-transition-system}
  For all reachable states $(s,\vec{p})$ and $(s',\vecprime{p})$ in the transition system $T_{s_0,\vecsup{p}{0}}$, if $s = s'$, 
  %for some indexes $i, j \in \mathbb{N}$, 
  then $\vec{p} = \vecprime{p}$.
\end{lemma}
\begin{proof}
  %Since the game is finite, some state must appear infinitely often, in particular for some indexes $i$, $j$, $s_i = s_j$.
  Assume $s = s'$.
  %, we prove that $p$ and $p'$ are the same point.
  By construction of $T_{s_0,\vecsup{p}{0}}$,
  %the transition system, 
  the left slope in $s$ of $\vec{p}$ and~$\vecprime{p}$ is the same: $\lslope(s,\firstdim{{\vec{p}}}) = \lslope(s_0,\firstdim{\vecsup{p}{0}} ) = \lslope(s_2,\firstdim{\vecprime{p}} )$.
  Assume towards a contradiction that $\vec{p} < \vecprime{p}$; the proof would work the same for $\vecprime{p} < \vec{p}$.
  Since $\vecprime{p}$ is a left accumulation point, there is an extremal point in $\openbegin \firstdim{\vec{p}} ,\firstdim{\vecprime{p}} \openend$.
  Lem.~\ref{lem:general-lemma-bis}.\ref{lem:extremal-slopes-bis} tells us that $\lslope(s_1,\firstdim{\vec{p}} ) \ne \lslope(s_2,\firstdim{\vecprime{p}} )$ which is a contradiction.
  Hence $\vec{p} = \vecprime{p}$.
\qed
\end{proof}
As a corollary of this lemma, the number of states that are reachable in $T_{s_0,\vecsup{p}{0}}$ is finite and bounded by $|S|$.
%(although a state cannot appear on the same path with two different points, it can appear on different paths with different points).

\subsection{Inverse betting game}\label{subsec:inverse-betting}
To show a contradiction, we will follow a path with left accumulation points.
We want this path to visit stochastic states which have only one successor in $T_{s_0,\vecsup{p}{0}}$.
For that, we will prove a property of an intermediary game that we call an inverse betting game.

An \emph{inverse betting game} is a two player game, given by $\langle V_\exists,V_\forall, E, (v_0,c_0), w \rangle$ where
$V_\exists$ and $V_\forall$ are the set of vertices controlled by \eve and \adam, respectively,
$\langle V_\exists \cup V_\forall, E \rangle$ is a graph whose each vertex has two successors,
$(v_0, c_0) \in V \times \mathbb{R}$ is the initial configuration, and
$w \colon E \to \mathbb{R}$ is a weight function such that for all $v \in V$:
\(\sum_{v' \mid (v,v') \in E} w(v,v') = 1\).

A configuration of the game is a pair $(v,c) \in V \times \mathbb{R}$ where $v$ is a vertex and $c$ a credit.
The game starts in configuration $(v_0, c_0)$ and is played by two players \eve and \adam.
At each step, from a configuration $(v,c)$ controlled by \eve, \adam suggests a valuation $d \colon E \to \mathbb{R}$ for the outgoing edges of $v$ such that $\sum_{v' \mid (v,v') \in E} w(v,v') \cdot d(v,v') = c$.
\eve then chooses a successor $v'$ such that $(v,v')\in E$ and the game continues from configuration $(v', d(v,v'))$.
From a configuration $(v,c)$ controlled by \adam, \adam choses a successor $v'$ of $v$ and keeps the same credit, hence the game continues from $(v',c)$.

Intuitively, \adam has some credit, and at each step he has to distribute it by betting over the possible successors.
Then \eve choses the successor and \adam gets a credit equal to its bet divided by the probability of this transition.
The game is \emph{inverse} because \eve is trying to maximize the credit of \adam.

\begin{theorem}\label{thm:inverse-betting}
  Let $\langle V_\exists,V_\forall, E, (v_0,c_0), w \rangle$ be an inverse betting game.
  Let $T \subseteq V_\exists \cup V_\forall$ be a target set and $B \in \mathbb{R}$ a bound.
  If from every vertex $v\in V$, \eve has a strategy to ensure visiting $T$ then she has one to ensure visiting it with a credit $c \ge 1$ or to exceed the bound, that is, she can force a configuration in $(T \times [c_0,+\infty\openend) \cup (V \times [B, +\infty\openend)$.
\end{theorem}

Our next step is transforming the transition system $T_{s_0,\vecsup{p}{0}}$ into such a game.
Consider the inverse betting game~$\mathcal{B}$
%played by two players \eve and \adam 
on the structure given by $T_{s_0,\vecsup{p}{0}}$ where $V_\exists = \statesprob$ are the states controlled by \eve, $V_\forall = \statesone \cup \statestwo$ are controlled by \adam, $w((s,\vec{p}) , (s',\vecprime{p})) = \tfunction(s,s')$ is a weight on edges and the initial configuration is $((s_0,\vecsup{p}{0}),\varepsilon_0)$.
Let $U_{s_0,\vecsup{p}{0}}$ the set of terminal states and of stochastic states that have only one successor in $T_{s_0,\vecsup{p}{0}}$.
We show that in the inverse betting game obtained from a stopping game \game, \eve can ensure visiting $U_{s_0,\vecsup{p}{0}}$.
% apply this theorem for the betting game obtained from a stopping game $\game$.
%We write

\begin{lemma}\label{lem:stopping-implies-one-successor}
  If $\game$ is stopping, there is a strategy for \eve in $\mathcal{B}$ such that from every vertex $v\in V$, all outcomes visit $U_{s_0,\vecsup{p}{0}}$.
%a stochastic state with only one successor. 
\end{lemma}
\begin{proof}
  Assume towards a contradiction that this is not the case, then by memoryless determinacy of turn-based reachability games (see e.g.~\cite{GTW03}) there is a vertex~$v$ and a memoryless deterministic strategy $\sigma_\adam$ of \adam, such that no outcomes of $\sigma_\adam$ from $v$ visit $U_{s_0,\vecsup{p}{0}}$.
  Let $\stratone$ and $\strattwo$ be the strategies of \PONE and \PTWO respectively corresponding to $\sigma_\adam$.
  Formally, if $h \in \paths{\game}^\Box$ then $\stratone(h) = \sigma_\adam(h)$ and if $h \in \paths{\game}^\Diamond$ then $\strattwo(h) = \sigma_\adam(h)$.
  We prove that all outcomes~$\path$ in \game of $\stratone,\strattwo$ from $v$ are outcomes of $\sigma_\adam$ in $\mathcal{B}$.
  This is by induction on the prefixes $\path_{\le i}$ of the outcomes.
  It is clear when $\path_{\le i}$ ends with states that are controlled by \PONE and \PTWO by the way we defined $\stratone$ and $\strattwo$, that $\path_{\le i+1}$ is also compatible with $\sigma_\adam$ in $\mathcal{B}$.
  For a finite path $\path_{\le i}$ ending with a stochastic state~$s$ in \game, two successors are possible.
  With the induction hypothesis that $\path_{\le i}$ is compatible with $\sigma_\adam$, and by the assumption on $\sigma_\adam$, $s$ does not belong to $U_{s_0,\vecsup{p}{0}}$.
  Therefore, both successors of $s$ are also in $T_{s_0,\vecsup{p}{0}}$, and $\path_{\le i+1}$ is compatible with $\sigma_\adam$ in $\mathcal{B}$.
  This shows that outcomes in \game of $(\stratone,\strattwo)$ are also outcomes of $\sigma_\adam$ in $\mathcal{B}$.
  %strategies of \PONE and \PTWO such that no states where we removed a stochastic edge is reached after $h$.
  %This means that set of paths compatible with these strategies are the same in $\mathcal{B}$ and the original game~$\game$.
  Therefore, \stratone and \strattwo ensure that from $v$, we visit no state of $U_{s_0,\vecsup{p}{0}}$ and thus no terminal state.
  This contradicts that the game is stopping.
\qed
\end{proof}
Putting Thm.~\ref{thm:inverse-betting} and Lem.~\ref{lem:stopping-implies-one-successor} together we can conclude the following:
\begin{corollary}\label{cor:inverse-betting}
  If \game is stopping then in $\mathcal{B}$, for any bound $B$, \eve has a strategy to ensure visiting $U_{s_0,\vecsup{p}{0}}$ with a credit $c \ge 1$ or making $c$ exceed $B$.
\end{corollary}

\subsection{Contradicting sequence}\label{sec:contra}
We define $\theta({s_0,\vecsup{p}{0}}) = \min \{ \eta(s,\vec{p}) \mid (s,\vec{p}) \text{ reachable in } T_{s_0,\vecsup{p}{0}} \}$, and
consider a sequence of points that are $\theta({s_0,\vecsup{p}{0}})$ close to $\vecsup{p}{0}$ and with a right slope that is decreasing at least as fast as that of their predecessors.

\begin{lemma}\label{lem:following-close}
  For stopping games, given $s_0\in \states$, $\vecsup{p}{0}\in \Rset^2$, and $\varepsilon_0>0$, such that $\varepsilon_0 < \theta({s_0,\vecsup{p}{0}})$, there is %exists
a finite sequence $\pi(s_0,\vecsup{p}{0},\varepsilon_0) = (s_i,\vecsup{p}{i},\varepsilon_i)_{i \le j}$ such that:
  \begin{itemize}
  \item $(s_i,\vecsup{p}{i})_{i \le j}$ is a path in $T_{s_0,\vecsup{p}{0}}$;
  \item for all $i\le j$, $\rslope(f_{s_i},\firstdim{\vecsup{p}{i}} - \varepsilon_i) \ge \rslope(f_{s_{i+1}},\firstdim{\vecsup{p}{i+1}} - \varepsilon_{i+1})$.
  \item either $\varepsilon_j \ge \theta({s_0,\vecsup{p}{0}})$ or $s_j \in U_{s_0,\vecsup{p}{0}}$ and $\varepsilon_j \ge \varepsilon_0$.
    %for all $i\in \mathbb{N}$, if $s_i \in U_{s_0,p_0}$, then $\varepsilon_{i+1} = \min \left\{ \frac{\varepsilon_i}{\tfunction(s_i,s_{i+1})} , \eta_{T'} \right\}$; \item for all $i\in \mathbb{N}$, if $s_i \not\in U_{s_0,p_0}$, then $\varepsilon_{i+1} = \varepsilon_i$;
  %\item for all $i\in \mathbb{N}$, if $s_i$ is controlled by \PONE or \PTWO, then $\varepsilon_{i+1} = \varepsilon_i$;  \item for all $i\in \mathbb{N}$, if $s_i$ is a stochastic state and $s_i \not\in U_{s_0,p_0}$, then $\varepsilon_{i+1} = \varepsilon_i$;
  \end{itemize}
\end{lemma}

The idea of the proof is that in $\mathcal{B}$, thanks to Lem.~\ref{lem:slope-convex-union}, \ref{lem:slope-intersection}, and \ref{lem:accumulation-stochastic-two}, \adam can always choose a successor such that $\rslope(f_{s_i},\firstdim{\vecsup{p}{i}} - \varepsilon_i) \ge \rslope(f_{s_{i+1}},\firstdim{\vecsup{p}{i+1}} - \varepsilon_{i+1})$.
Then thanks to Cor.~\ref{cor:inverse-betting}, there is a strategy for \eve to reach $(U_{s_0,\vecsup{p}{0}} \times [c_0,+\infty\openend) \cup (V \times [B, +\infty\openend)$.
By combining the two strategies, we obtain an outcome that satisfies the desired properties.

We use the path obtained from this lemma to show that no matter how small $\varepsilon_0$ we choose, $\varepsilon_i$ can grow to reach $\theta({s_0,\vecsup{p}{0}})$.
%\romain{explain how it is a consequence of the previous one}

\begin{lemma}\label{lem:constant-slope}
  For all states $s$ with a left accumulation point $\vec{p}$ and for all $0<\varepsilon < \theta({s,\vec{p}})$, there is some $(s', \vecprime{p})$ reachable in $T_{s,\vec{p}}$
  such that $\rslope(f_{s'},\firstdim{\vecprime{p}} - \theta({s,\vec{p}})) \le \rslope(f_s,\firstdim{\vec{p}} - \varepsilon)$.
\end{lemma}

Thanks to this lemma, we can now prove Thm.~\ref{thm:stopping-no-accum}.
  Assume towards a contradiction that there is a left accumulation point $\vec{p}$ in the state $s$.
  Let $m = \min \{ \lslope(f_{s'},\firstdim{\vecprime{p}} - \theta({s,\vec{p}})) \mid (s',\vecprime{p}) \text{ reachable in } T_{s,\vec{p}}\}$ and $(s',\vecprime{p})$ the configuration of $T_{s,\vec{p}}$ for which this minimum is reached (it is reached because the number of reachable configurations is finite: this is a corollary of Lem.~\ref{lem:finite-transition-system}).
  Because of Lem.~\ref{lem:constant-slope}, $\rslope(f_s,\firstdim{\vec{p}} - \varepsilon)$ is greater than $m$.
  By Lem.~\ref{lem:general-lemma-bis}.\ref{lem:limit-slope-bis}, when $\varepsilon$ goes towards $0$, $\rslope(f_s,\firstdim{\vec{p}} - \varepsilon)$ converges to $\lslope(f_s,\firstdim{\vec{p}} )$.
  This means that $\lslope(f_s,\firstdim{\vec{p}} ) \ge m$.
  Moreover, by construction of $T_{s,\vec{p}}$, we also have that $\lslope(f_{s'},\firstdim{\vecprime{p}} ) = \lslope(f_s,\firstdim{\vec{p}} )$, so $\lslope(f_{s'},\firstdim{\vecprime{p}} ) \ge m$.
Because the slopes are decreasing (Lem.~\ref{lem:general-lemma-bis}.\ref{lem:slope-decrease-bis}), $ m = \rslope(f_{s'},\firstdim{\vecprime{p}} - \theta({s,\vec{p}})) \ge \lslope(f_{s'},\firstdim{\vecprime{p}} ) \ge m$.
Hence, the left and right slopes of $f_{s'}$ are constant on $[\firstdim{\vecprime{p}} - \theta({s,\vec{p}}), \firstdim{\vecprime{p}} ]$, and Lem.~\ref{lem:general-lemma-bis}.\ref{lem:extremal-slopes-bis} implies that there are no extremal point in $\openbegin \firstdim{\vecprime{p}} - \theta({s,\vec{p}}) , \firstdim{\vecprime{p}} \openend$.
  This contradicts the fact that $\vecprime{p}$ is a left accumulation point: there should be an extremal point in any neighbourhood on the left of $\vecprime{p}$.
  Hence, $f_s$ contains no accumulation point.

\begin{remark}
  One might attempt to extend the proof of Thm.~\ref{thm:stopping-no-accum} to three or more objectives, but this does not seem to be easily doable.
  Although it is possible to use directional derivative (or pick a subgradient) instead of using left and right slope in such setting,
  an analogue of Lem.~\ref{lem:general-lemma-bis}.\ref{lem:extremal-slopes-bis}
  cannot be proved because in multiple dimensions, two accumulation points can share
  the same directional derivative, for a fixed direction. 
  It is also not easily possible to avoid this problem by following several directional
  derivatives instead of just one.
  This is because the slope in one direction may be inherited from one successor while the slope in another direction comes from another successor.
  We give more details and example of convex sets that would contradict generalisations of Lem.~\ref{lem:general-lemma-bis}.\ref{lem:extremal-slopes-bis} and Lem.~\ref{lem:slope-intersection} in Appendix~\ref{sec:difficulties}.
\end{remark}

\section{Conclusions}
We have studied stochastic games under multiple objectives, and have provided
decidability results for determined games and for stopping games
with two objectives.
Our results for non-determined games provide an important milestone towards
obtaining decidability for the
general case, which is a major task which will require further novel insights into the
problem. Another research direction concerns establishing an upper bound on the number
of extremal points of a Pareto curve; such result would allow us to give upper complexity
bounds for the problem.

\startpara{Acknowledgements} The authors would like to thank Aistis \v{S}imaitis and Clemens Wiltsche for their useful discussions on the topic.

\bibliographystyle{abbrv}
\bibliography{biblio}

\begin{thebibliography}{10}

\bibitem{BKTW14}
N.~Basset, M.~Kwiatkowska, U.~Topcu, and C.~Wiltsche.
\newblock Strategy synthesis for stochastic games with multiple long-run
  objectives.
\newblock In {\em TACAS}. Springer, 2015.

\bibitem{DBLP:conf/cav/BrenguierR15}
R.~Brenguier and J.~Raskin.
\newblock Pareto curves of multidimensional mean-payoff games.
\newblock In {\em CAV}, 2015.

\bibitem{Krish16}
K.~Chatterjee and L.~Doyen.
\newblock Perfect-information stochastic games with generalized mean-payoff
  objectives.
\newblock In {\em LICS}, 2016.
\newblock To appear.

\bibitem{DBLP:journals/jcss/ChatterjeeH12}
K.~Chatterjee and T.~A. Henzinger.
\newblock A survey of stochastic {\(\omega\)}-regular games.
\newblock {\em J. Comput. Syst. Sci.}, 78(2), 2012.

\bibitem{CRR14}
K.~Chatterjee, M.~Randour, and J.~Raskin.
\newblock Strategy synthesis for multi-dimensional quantitative objectives.
\newblock {\em Acta Inf.}, 51(3-4), 2014.

\bibitem{ChenFKSTU12}
T.~Chen, V.~Forejt, M.~Kwiatkowska, A.~Simaitis, A.~Trivedi, and M.~Ummels.
\newblock Playing stochastic games precisely.
\newblock In {\em CONCUR}, 2012.

\bibitem{CFKSW13}
T.~Chen, V.~Forejt, M.~Kwiatkowska, A.~Simaitis, and C.~Wiltsche.
\newblock On stochastic games with multiple objectives.
\newblock In {\em MFCS}. Springer, 2013.

\bibitem{Condon1992}
A.~Condon.
\newblock The complexity of stochastic games.
\newblock {\em Inf. Comput.}, 96(2), 1992.

\bibitem{EKVY08}
K.~Etessami, M.~Kwiatkowska, M.~Vardi, and M.~Yannakakis.
\newblock Multi-objective model checking of {Markov} decision processes.
\newblock {\em LMCS}, 4(4), 2008.

\bibitem{GTW03}
E.~Gr{\"a}del, W.~Thomas, and T.~Wilke.
\newblock {\em Automata, Logics, and Infinite Games: A Guide to Current
  Research}, volume 2500.
\newblock Springer, 2003.

\bibitem{DBLP:conf/fsttcs/HunterR14}
P.~Hunter and J.~Raskin.
\newblock Quantitative games with interval objectives.
\newblock In {\em FSTTCS}, 2014.

\bibitem{LMMST16}
K.~G. Larsen, M.~Mikucionis, M.~Mu{\~{n}}iz, J.~Srba, and J.~H. Taankvist.
\newblock Online and compositional learning of controllers with application to
  floor heating.
\newblock In {\em TACAS}, 2016.

\bibitem{nash50}
J.~F. Nash, Jr.
\newblock Equilibrium points in {\(n\)}-person games.
\newblock {\em Proc.\ National Academy of Sciences of the~{USA}}, 36(1), Jan.
  1950.

\bibitem{Tarski:reals-arithmetic}
A.~Tarski.
\newblock {\em A Decision Method for Elementary Algebra and Geometry}.
\newblock Univ.\ of California Press, Berkeley, 1951.

\bibitem{Velner2015}
Y.~Velner.
\newblock Robust multidimensional mean-payoff games are undecidable.
\newblock In {\em FoSSaCS}. Springer, 2015.

\bibitem{DBLP:journals/iandc/VelnerC0HRR15}
Y.~Velner, K.~Chatterjee, L.~Doyen, T.~A. Henzinger, A.~M. Rabinovich, and
  J.~Raskin.
\newblock The complexity of multi-mean-payoff and multi-energy games.
\newblock {\em Inf. Comput.}, 241, 2015.

\end{thebibliography}

\clearpage
\appendix
\section{Proofs for determined games}
\subsection{Proof of Lem.~\ref{lemma:identify-md}}
\def\MD{\textsf{MD}\xspace}
We write $\allstratstwo_\MD$ for the set of memoryless deterministic strategies of \PTWO. 
\begin{align*}
 &&  \forall \varepsilon > 0.\exists \stratone \in \allstratsone.\ \forall \strattwo\in \allstratstwo.\  & \bigwedge_i \Exp^{\stratone,\strattwo}(\reward_i) \ge \bound_i -\varepsilon \\
 \Leftrightarrow &&  \forall \varepsilon > 0. \neg \exists \strattwo \in \allstratstwo.\ \forall \stratone\in \allstratsone.\  & \bigvee_i \Exp^{\stratone,\strattwo}(\reward_i) < \bound_i -\varepsilon \tag{determinacy}\\
 \Leftrightarrow &&  \forall \varepsilon > 0. \neg \exists \strattwo \in \allstratstwo_\MD.\ \forall \stratone\in \allstratsone.\  & \bigvee_i \Exp^{\stratone,\strattwo}(\reward_i) < \bound_i -\varepsilon \tag{Thm.~\ref{theorem:spoiler-md}}\\
 \Leftrightarrow &&  \forall \varepsilon > 0. \forall \strattwo \in \allstratstwo_\MD.\ \exists \stratone\in \allstratsone.\  & \bigwedge_i \Exp^{\stratone,\strattwo}(\reward_i) \ge \bound_i -\varepsilon \tag{logical equivalences}\\
 \Leftrightarrow &&  \forall \strattwo \in \allstratstwo_{\MD}.\ \forall \varepsilon > 0.\ \exists \stratone\in \allstratsone.\ & \bigwedge_i \Exp^{\stratone,\strattwo}(\reward_i) \ge \bound_i - \varepsilon\tag{logical equivalences}\\
 \Leftrightarrow &&  \forall \strattwo \in \allstratstwo_{\MD}.\ \exists \stratone \in \allstratsone.\ & \bigwedge_i \Exp^{\stratone,\strattwo}(\reward_i) \ge \bound_i\tag{property of MDPs~\cite{EKVY08}}
\end{align*}
%In ($\ast$), the implication $\Rightarrow$ follows trivially because if $\Exp^{\stratone,\strattwo}(\reward_i) \ge \bound_i -\varepsilon + \varepsilon'$,
%then $\Exp^{\stratone,\strattwo}(\reward_i) \ge \bound_i -\varepsilon$. In the direction $\Leftarrow$, it suffices to take, for a fixed $\varepsilon$, the number $\varepsilon' \mydef \varepsilon/2$, and because $\varepsilon-\varepsilon' > 0$, we have that $\forall \strattwo \in \allstratstwo_\MD.\ \exists \stratone\in \allstratsone.\  \bigwedge_i \Exp^{\stratone,\strattwo}(\reward_i) \ge \bound_i -\varepsilon + \varepsilon'$ holds true.

\section{Proof for non-determined case}
\begin{lemma}\label{lem:f-lemma}
  Let $s$ be a state of \game.
  \begin{enumerate}
  \item\label{lem:f-defined} The domain of definition of $f_s$ is an interval of $\mathbb{R}$.
  \item\label{lem:f-concave} $f_s$ is decreasing and concave.
  \item\label{lem:f-continuity} If $f_s$ is defined on $\openbegin a,b \openend$ then $f_s$ is continuous on $\openbegin a,b\openend$.
  \end{enumerate}
\end{lemma}
%\end{reflemma}
\begin{proof}
  \point{\ref{lem:f-defined}} Assume that $f_s$ is defined for some $x,x'$.
  For all $x'' \in \openbegin x,x' \openend$, by convexity $(x'',\frac{x'' - x'}{x - x'}\cdot f_s(x) + \frac{x''-x}{x'-x} \cdot f_s(x'))$ can be achieved.
  We define $y = \sup \{ y \mid (x'',y) \text{ achievable} \}$, it is defined because of the previous remark, and by the properties of stopping games $(x'',y)$ is achievable and we now show that $(x'',y)\in f_s$.
  
  Assume there is an achievable $\vec{p}$ with $\firstdim{\vec{p}} > x''$.
  Then by convexity of the set of achievable points~\cite{CFKSW13},
  %\vojta{concavity elsewhere? of what?}, 
  any point in $[(x,f_s(x)),\vec{p}\,]$ is achievable.
  Because $(x,f_s(x))$ lies on the Pareto curve, we have $f_s(x) > \seconddim{\vec{p}}$ (otherwise $(x,f_s(x))$ is not maximal point satisfying the defining properties of the Pareto curve), therefore $(x'', \frac{x''-\firstdim{\vec{p}}}{x - \firstdim{\vec{p}}} \cdot f_s(x) + \frac{x''-x}{\firstdim{\vec{p}} - x} \cdot f_s(\firstdim{\vec{p}}))$ is achievable which contradicts the definition of $y$. Therefore, $f_s$ is defined in $x''$ and equals $\seconddim{\vec{p}}$.

  \point{\ref{lem:f-concave}}
  The fact that if $x' \ge x$, then $f_s(x) \ge f_s(x')$ comes by the definition of Pareto curves and the fact that $(x,f_s(x))$ lies on a Pareto curve. 
  We now prove that $f_s$ is concave.
  Let $x < x'$ in the domain of $f_s$ and $t \in [0,1]$.
  By convexity of the set $\achievable_s$~\cite{CFKSW13}, $t \cdot (x,f_s(x)) + (1 - t) \cdot (x',f_s(x'))$ can be ensured.
  Hence, there is a point $\vec{p}$ that is strictly greater than $t \cdot (x,f_s(x)) + (1 - t) \cdot (x',f_s(x'))$ that belongs to $f_s$.
  Since $f_s$ is defined in $t \cdot x + (1 - t) \cdot x'$ and it is decreasing this means that $f_s(t \cdot x + (1 - t) \cdot x') \ge t \cdot f_s(x) + (1 - t) \cdot f_s(x')$.
  This shows concavity.

  \point{\ref{lem:f-continuity}} 
  Since $f_s$ is concave, it is the negation of a convex function.
  A convex function defined on an open interval is continuous on this interval.
  Therefore, if $f_s$ is defined on $\openbegin a,b\openend$ it continuous on $\openbegin a,b\openend$.
  %% We now show continuity. 
  %% Assume towards a contradiction that for some $x\in [a,b\openend$, $f_s(x) \ne \lim_{\varepsilon \to 0^+} f_s(x - \varepsilon)$; the proof would work the same way if we assumed $f_s(x) \ne \lim_{\varepsilon \to 0^-} f_s(x - \varepsilon)$.
  %%   Since $f_s$ is defined in $b$\dots
  %%   %Consider the strategy $\stratone_n$ that ensures $(x - \frac{1}{n}, f_s(x - \frac{1}{n}))$. 
  %% %We can extract a subsequence $(\stratone'_n)_{n \in \mathbb{N}}$ of strategy from $(\stratone_n)_{n \in \mathbb{N}}$ such that for all $n \in \mathbb{N}$, strategies of $\stratone'_{n'}$ for $n' \ge n$ coincide on all paths of length smaller than $n$: this is because there are only finitely many ---> doesn't work for randomised strategies.
  %% \vojta{To finish}
  %% %\fbox{To finish} (I guess we have to consider a sequence of points, and the sequence of corresponding strategies, show there is an accumulation strategy, and that it ensures the right threshold, but maybe it is already proved somewhere).
\qed
\end{proof}

%\begin{figure}[htb]
%\centering{
%\begin{tikzpicture}[scale=2]
%  \draw[-latex'] (0,0) -- (2,0);
%  \draw[-latex'] (0,0) -- (0,1.5);
%
%  \draw[blue,very thick] (0.91,0) -- (0.9,0.3)  -- (0.75,0.75) -- (0.5,0.9) -- (0,1);
%
%  \draw (1.05,-0.15) node[right] {$\rslope(x)$} -- (0.9,0.3) -- (0.75,0.75) -- (0.6,1.2) ;
%  \draw (1.25,0.45) node[right] {$\lslope(x)$} -- (1,0.6) -- (0.75,0.75) -- (0.5,0.9) -- (0,1.2);
%  %\draw (0,1) -- (0.5,0.9) -- (1,0.8) -- (1.5,0.7) node[right] {$\slope(x - \varepsilon)$} ;
%
%  \draw (0.75,1.5) -- (0.75,0) node[below] {$x$};
%  \draw (0.75,0.75) node[right] {$\vec{p}$};
%
%  %\draw (0.5,1.5) -- (0.5,0) node[below] {$x - \varepsilon$};
%  %\draw[<->] (0.5,-0.3) -- node[below] {$R_{\varepsilon}$} (0.875,-0.3);
%\end{tikzpicture}
%}
%\end{figure}

\subsection{Proof of Lem.~\ref{lem:general-lemma-bis}}
We prove a bit more than what Lem.~\ref{lem:general-lemma-bis} gives in the main part of the paper:
%\begin{reflemma}{lem:general-lemma}
\begin{lemma}\label{lem:general-lemma}
  \begin{enumerate}
%  \item $f_s$ is defined on an interval of $[0,1]$ and is continuous. \label{lem:continuity}
  \item \label{lem:intermediary-slope}
    Let $f$ be a function (not necessarily concave) whose left slope is well defined on the interval $[a,b]$, then there exists $x \in [a,b]$ such that $\rslope(f,x) \ge \frac{f(a) - f(b)}{a - b}$. 
    Similarly, there exists $x' \in [a,b]$ such that $\lslope(f,x') \le \frac{f(a) - f(b)}{a - b}$.
  \item \label{lem:slope-domain}If $f$ is concave, then $x \mapsto \lslope(f,x)$ is defined on the same interval as $f$ except the left most point.
    %and $\lslope(f,x) \le 0$.
  \item \label{lem:slope-decrease} If $f$ is concave and $x < x'$ are two reals for which $\lslope(f)$ is defined, 
    %then $f_s(x) \ge f_s(x')$ and 
    then $\lslope(f,x) \ge \rslope(f,x) \ge \lslope(f,x')$.
  %\item If $x,x'$ are two reals where $f_s$ is defined, then there exists $x'' \in [x,x']$ such that $\lslope(s,x'') \ge \rslope(s,x'') \ge \frac{f_s(x) - f_s(x')}{x - x'}$. 
    %Similarly there exists $x''' \in [x,x']$ such that $\rslope(s,x'') \le \lslope(s,x''') \le \frac{f_s(x) - f_s(x')}{x - x'}$. \label{lem:intermediary-slope}
  \item \label{lem:extremal-slopes}
  % If $(x,y)$ and $(x',y')$ are two extremal point of $X_s$, then $\lslope(s,x) \ne \lslope(s,x')$. \fbox{not exactly...}\label{lem:extremal-slopes}
  If $f$ is concave and $\openbegin x, x' \openend$ contains an extremal point of $f$, then $\lslope(f,x) \ne \lslope(f,x')$.
  %(x_1,y_1)$ with $x_1 \in \openbegin x,x'\openend$) then $\lslope(f_s,x) \ne \lslope(f_s,x')$.
\item \label{lem:slope-extremals} If $f$ is concave and $\lslope(f,x) \ne \lslope(f,x')$ then $[x,x']$ contains an extremal point.
  \item  \label{lem:limit-slope}  If $f$ is a concave function whose left slope is defined in $x$, then $\lim_{x' \to x^-} \lslope(f,x') = \lim_{x' \to x^-} \rslope(f,x') = \lslope(f,x)$. 
  \end{enumerate}
%\end{reflemma}
\end{lemma}
\begin{proof}
 \point{\ref{lem:intermediary-slope}} 
 First we show that if $\rslope(f,x) < d$ for all $x\in [a,b]$ then for all $x\in \openbegin a,b]$, $x < f(a) + (x - a) \cdot d$.
 %First we show that if $\lslope(f,x) \le \frac{f(a) - f(b)}{a - b}$ for all $x\in [a,b]$ then for all $x\in [a,b]$, $x \le f(a) + (x - a) \cdot \frac{f(a) - f(b)}{a - b}$.
 To prove this, let $c$ be the right-most point of $[a,b]$ which minimizes $f(c) - f(a) - (c - a) \cdot d$.
% To prove this, let $c$ be the left-most point of $[a,b]$ that maximizes $f(c) - f(a) - (c - a) \cdot \frac{f(a) - f(b)}{a - b}$.
 Assume towards a contradiction that $c \ne b$.
 Since $\rslope(f,c) < d$, there exists $x \in \openbegin c,b \openend$ such that $\frac{f(x) - f(c)}{x - c} < d$.
 %Since $\lslope(f,c) \le \frac{f(a) - f(b)}{a-b}$, there exists $x \in ]a,c[$ such that $\frac{f(x) - f(c)}{x - c} \le \frac{f(a) - f(b)}{a - b}$.
 Since $(x - c)$ is positive this implies that:
 \begin{align*} 
   f(x) & < f(c) + (x - c) \cdot d \\
   f(x) - f(a) - (x - a) \cdot d & < f(c) - f(a) - (x - a - x + c) \cdot d \\
   f(x) - f(a) - (x - a) \cdot d & < f(c) - f(a) - (c - a) \cdot d
   %% f(x) & \ge f(c) + (x - c) \cdot  \frac{f(a) - f(b)}{a - b} \\
   %% f(x) - f(a) - (x - a) \cdot \frac{f(a) - f(b)}{a - b} & \ge f(c) - f(a) - (x - a - x + c) \cdot  \frac{f(a) - f(b)}{a - b} \\
   %% f(x) - f(a) - (x - a) \cdot \frac{f(a) - f(b)}{a - b} & \ge f(c) - f(a) - (c - a) \cdot  \frac{f(a) - f(b)}{a - b}
 \end{align*}
 Which contradicts the fact that $c$ is the right-most point of $[a,b]$ maximizing $f(c) - f(a) - (c - a) \cdot d$. %$\frac{f(a) - f(b)}{a - b}$.
 Hence, $b$ is the only point that maximizes this quantity, and it equals $0$ in $b$.
 We therefore have that for all $x \in [a,b\openend$, $f(x) < f(a) + (x - a) \cdot d$.
 %\frac{f(a) - f(b)}{a - b}$.

 In particular, assuming towards a contradiction that for all $x \in [a,b]$, $\rslope(f,x) < \frac{f(a) - f(b)}{a - b}$, we would have $f(a) < f(b) + (a - b) \cdot \frac{f(a) - f(b)}{a - b} = f(a)$.
 Hence, there exists some $x \in [a,b]$ such that $\rslope(f,x) \ge \frac{f(a) - f(b)}{a - b}$.

 Consider now the function $-f$.
 As a consequence of what we just proved, there is $x \in [a,b]$ such that $\lslope(-f,x) \ge \rslope(-f,x) \ge - \frac{f(a) - f(b)}{a - b}$.
 Since $\lslope(-f,x) = -\lslope(f,x)$, this implies that $\lslope(f,x) \le \frac{f(a) - f(b)}{a-b}$.

  \point{\ref{lem:slope-domain}}
  Assume that $f$ is defined on $[a,b]$ and let $x\in \openbegin a , b]$.
    Define the function $g$ by $g(x') = \frac{f(x') - f(x)}{x' - x}$, it is defined for $x'$ smaller than $x$ and close enough to $x$. 
  Moreover by concavity of $f$, $g$ is decreasing, therefore its left limit in $x$ is well defined, and equals the left slope of $f$ in $x$.

 \point{\ref{lem:slope-decrease}} 
 By concavity, for $x < x'$ and $0 < \varepsilon < \frac{x' - x}{2}$:
 \[
 \frac{f(x) - f(x-\varepsilon)}{\varepsilon} \ge \frac{f(x) - f(x+\varepsilon)}{\varepsilon} \ge \frac{f(x') - f(x'-\varepsilon)}{\varepsilon}\]
 Hence, it is the same for the limit when $\varepsilon$ moves towards $0$ and $\lslope(f,x) \ge \lslope(f,x')$.

 \point{\ref{lem:extremal-slopes}}
 Assume towards a contradiction that $\lslope(f,x) = \lslope(f,x')$, and $\vec{p}$ is an extremal point with $x < \firstdim{\vec{p}} < x'$.
 By item \ref{lem:slope-decrease} this means the slope is constant on $[x,x']$ and as a consequence of item \ref{lem:intermediary-slope}, it is equal to $\frac{f(x') - f(x)}{x' - x}$.
 Since it is an extremal point, $f(\firstdim{\vec{p}}) > f(x) + (\firstdim{\vec{p}} - x) \cdot \frac{f(x') - f(x)}{x' - x}$.
 By item \ref{lem:intermediary-slope}, there is $x'' \in [x,\firstdim{\vec{p}}]$ such that:
 \begin{align*} 
   \lslope(f,x'') & \ge \frac{f(\firstdim{\vec{p}}) - f(x)}{\firstdim{\vec{p}} - x} \\
   & > \frac{f(x) + (\firstdim{\vec{p}} - x) \cdot \frac{f(x') - f(x)}{x' - x} - f(x)}{\firstdim{\vec{p}} - x} \\
   & > \frac{(\firstdim{\vec{p}} - x) \cdot \frac{f(x') - f(x)}{x' - x}}{\firstdim{\vec{p}} - x} \\
   & > \frac{f(x') - f(x)}{x' - x} 
 \end{align*}
 Which is in contradiction with the fact that on $[x,x']$ the slope is constant and equal to $\frac{f(x') - f(x)}{x' - x}$.

 \point{\ref{lem:slope-extremals}}
 Assume towards a contradiction that $[x,x']$ contains no extremal point, then for all $x'' \in [x,x']$, $(x'',f(x''))$ is a convex combination of $(x,f(x))$ and $(x',f(x'))$.
 Hence, if $x'' \in \openbegin x,x']$, $\lslope(f,x'') = \frac{f(x) - f(x'')}{x - x''}$.
 This is in particular the case for $x'$.
 Since $\lslope(f,x) \ne \lslope(f,x')$, $x$ has a left slope different from all the points in $\openbegin x,x']$; we will prove that $(x,f(x))$ is an extremal point.
 If $x$ is a convex combination of two points of the curve given by $f$, then as all the points on its right are below the line segment $[(x,f(x)), (x',f(x'))]$, one point $\vec{p}\in f$
 %$(x_1,f_s(x_1))$ with $x_1 < x$ 
 with $\firstdim{\vec{p}} < x$ must be above the corresponding line.
 By concavity, this is also the case for all the $x'' \in \openbegin \firstdim{\vec{p}},x\openend$.
 This means that $\frac{f(x'') - f(x)}{x'' - x} \le \frac{f(x') - f(x)}{x' - x}$.
 The limit when $x''$ goes towards $x$ is also smaller than $\frac{f(x') - f(x)}{x' - x}$.
 It cannot be greater because by item~\ref{lem:slope-decrease} the slope decreases.
 This contradicts that $\lslope(f,x) \ne \lslope(f,x')$.
 This contradiction proves that $[x,x']$ contains an extremal point.

 \point{\ref{lem:limit-slope}}
 The fact that $\lim_{x' \to x^-} \lslope(f,x') \ge \lim_{x' \to x^-} \rslope(f,x') \ge \lslope(f,x)$ comes from the fact that the slope is decreasing (point~\ref{lem:slope-decrease}).
 We now show $\lslope(f,x) \le \lim_{x' \to x^-} \lslope(f,x')$ which implies the result.
% \fbox{change for $\rslope$}
 Assume towards a contradiction that $\lim_{x' \to x^+} \lslope(f,x') > \lslope(f,x)$. 
 Then there is $\varepsilon > 0$ such that $\frac{f(x) - f(x - \varepsilon)}{\varepsilon} < \lim_{x' \to x^-} \lslope(f,x')$.
 By point~\ref{lem:intermediary-slope}, there is $x'' \in [x - \varepsilon, x]$ such that $\lslope(f,x'') \le \frac{f(x) - f(x - \varepsilon)}{\varepsilon}  < \lim_{x' \to x^-} \lslope(f,x')$.
 Since the slope is decreasing, $\lim_{x' \to x^-} \lslope(f,x') \ge \lslope(f,x'')$ which gives a contradiction.
\qed
\end{proof}

\subsection{Evolution of the slope in \PONE{} states (Proof of Lem.~\ref{lem:slope-convex-union})}

\begin{lemma}\label{lem:bound-slope-convex-union}
  Let $s_0$ be a \PONE{} state with two successors $s_1$ and $s_2$.
  For all $x$ where the slopes of $f_{s_1}$ and $f_{s_2}$ are defined, $f_{s_0}$ and its slope are defined and:
  \[
  \begin{array}{r @{~\le~} c @{~\le~} l}
    \min(\lslope(f_{s_1},x), \lslope(f_{s_2},x)) & \lslope(f_{s_0},x) &  \max(\lslope(f_{s_1},x), \lslope(f_{s_2},x)) \\
    \min(\rslope(f_{s_1},x), \rslope(f_{s_2},x)) & \rslope(f_{s_0},x) & \max(\rslope(f_{s_1},x), \rslope(f_{s_2},x)) 
  \end{array}
  \]
\end{lemma}
\begin{proof}
  To lighten the notation we will write $f_i$ instead of $f_{s_i}$ in the following.

  Assume towards a contradiction that there is $x$ where $f_{1}$ and $f_{2}$ are defined and where $f_0$ is not defined.
  Let $y$ be the maximum such that $\vec{p}=(x,y)$ can be ensured from $s_0$.
  %Assume towards a contradiction that , 
  Since $f_0$ is not defined in $x$, there exists $(x',y')\in \achievable_{s_0}$ such that $x < x'$ and $y\le y'$.
  By downward closure of the set of achievable points, $(x,y')$ can be ensured from $s_0$, thus by definition of $y$, $y' \le y$ and therefore $y = y'$.
  Then there is $(x',y) \in \achievable_{s_0}$ with $x' > x$.
  %\review{The statement assume towards contradiction that $f_0$ is not defined in x … should come directly after the one about assume x is such that $f_1$ and $f_2$ …, Right now it becomes unclear what you are working towards contradicting (it could be read as you trying to contradict the previous statement about let y be the maximum …}
  %\review{how do you know that there is a $(x',y) \in \achievable_{s_0}$}?
  By the characterisation of the set of values that can be ensured, we have $(x',y) \in \conv(\achievable_{s_1} \cup \achievable_{s_2})$, so there are $\lambda_1,\lambda_2 \in [0,1]$ and 
  $\vec{q},\vec{r} \in \achievable_{s_1} \cup \achievable_{s_2}$ such that $\lambda_1 + \lambda_2 = 1$ and $(x',y) = \lambda_1 \cdot \vec{q} + \lambda_2 \cdot \vec{r}$.
  If $\firstdim{\vec{q}} > y$ then there is some $\vecprime{p}$ which can be ensured with $\firstdim{\vecprime{p}} > x$ and $\seconddim{\vecprime{p}} > y'$ which by downward closedness contradicts the fact that $y$ is the maximum such that $(x,y)$ can be ensured.
  The case $\seconddim{\vec{r}} > y$ is similar.
  In the remaining cases $y = \seconddim{\vec{q}} = \seconddim{\vec{r}}$. 
  Assume w.l.o.g. $\firstdim{\vec{q}} \ge x$, and therefore $\vec{q} \ge (x,y)$.
  Since $f_{1}$ is defined in $x$, $f_{1}(x) > \seconddim{\vec{q}}$.
  By characterisation of the set of values that can be ensured $y \ge f_{1}(x) > \seconddim{\vec{q}}$ which is a contradiction.

  This shows that if $f_{1}$ and $f_{2}$ are defined then $f_{0}$ also is.
  Since $\lslope(f_{0},x)$ is defined for all $x$ where $f_{0}$ except the left-most point, it is defined where $f_{1}$ and $f_{2}$ are defined except the left-most point of one of the two.
  In this left-most point one of the two slopes is not defined.
  This shows that if both left slopes are defined then $\lslope(f_{0},x)$ also is.
  The proof proceeds similarly for the right slopes.
  
  \medskip

  Let $m= \min(\lslope(f_1,x),\lslope(f_2,x))$ and $\vec{n} = (-m,1)$.
  Assume towards a contradiction that $\lslope(f_0,x) < m$.
  Then by definition of $\lslope$ there exists $x' < x$ such that $\frac{f_0(x) - f_0(x')}{x - x'} < m$.
  Since $x - x' > 0$, we have $f_0(x') > f_0(x) + (x' - x) \cdot m$.
  Therefore, $(x',f_0(x')) \cdot \vec{n} = f_0(x') - m x' > f_0(x) - m x = (x,f_0(x)) \cdot \vec{n}$.
  %which means that $(x',f_0(x'))$ is above the line $\ell = \{ (x,f_0(x)) + t \cdot (1,\min(\lslope(f_1,x),\lslope(f_2,x))) \mid t \in \mathbb{R} \}$
  By the characterization of the Pareto curve, there are $\vecprime{q}, \vecprime{r} \in \achievable_{s_1} \cup \achievable_{s_2}$, such that $(x',f_0(x'))$ is a convex combination of $\vecprime{q}$ and $\vecprime{r}$.
  %and $\lambda_1', \lambda_2' \in [0,1]$ such that $\lambda_1' + \lambda_2' = 1$, $x' = \lambda_1' \cdot x_1' + \lambda_2' \cdot x_2'$ and $f_{0}(x') = \lambda_1' \cdot f_{1}(x_1') + \lambda_2' \cdot f_{2}(x_2')$.
  Because their convex combination has greater dot product with $\vec{n}$ than $(x,f_0(x))$, it is also the case of one of the points $\vecprime{q}$ or $\vecprime{r}$.
  %$(x',f_0(x')), (x,f_0(x))$ 

  Without loss of generality, we assume it is $\vecprime{q}$.
  So we have $f_1(\firstdim{\vecprime{q}}) > f_0(x) + (\firstdim{\vecprime{q}} - x) \cdot m$.
  \begin{itemize}
  \item If $\firstdim{\vecprime{q}} < x$, then
    $\frac{f_0(x) - f_1(\firstdim{\vecprime{q}})}{x - \firstdim{\vecprime{q}}} < m \le \lslope(f_1,x)$.
    By the characterization of the Pareto curve $f_1(x) \le f_0(x)$, therefore 
    $\frac{f_1(x) - f_1(\firstdim{\vecprime{q}})}{x - \firstdim{\vecprime{q}}} < \lslope(f_1,x)$.
    By Lem.~\ref{lem:general-lemma}.\ref{lem:intermediary-slope}, there is $x'' \in [\firstdim{\vecprime{q}},x]$ such that $\lslope(f_1,x'') \le \frac{f_1(x) - f_1(\firstdim{\vecprime{q}})}{x - \firstdim{\vecprime{q}}} < \lslope(f_1,x)$.
    By Lem.~\ref{lem:general-lemma}.\ref{lem:slope-decrease}, $\lslope(f_1,x'') \ge \lslope(f_1,x)$ which is a contradiction.
  \item If $\firstdim{\vecprime{q}} \ge x$, then $x \in [x',\firstdim{\vecprime{q}}]$ and by the characterization of the Pareto curve $f_0(\firstdim{\vecprime{q}}) \ge f_1(\firstdim{\vecprime{q}})$, so $f_0(x'_1) > f_0(x) + (\firstdim{\vecprime{q}} - x) \cdot m$.
    This implies $(\firstdim{\vecprime{q}},f_0(\firstdim{\vecprime{q}})) \cdot \vec{n} > (x,f_0(x)) \cdot \vec{n}$.
    By concavity of $f_0$:
    \begin{align*}
      f_0(x) & \ge f_0(x') + (x - x') \cdot \frac{f_0(\firstdim{\vecprime{q}}) - f_0(x')}{\firstdim{\vecprime{q}} - x'} \\
      (x,f_0(x)) \cdot \vec{n} & \ge \left(1 - \frac{x-x'}{\firstdim{\vecprime{q}} - x'}\right) (x',f_0(x')) \cdot \vec{n} + \frac{x - x'}{\firstdim{\vecprime{q}} - x'} (\firstdim{\vecprime{q}},f_0(\firstdim{\vecprime{q}})) \cdot \vec{n}
    \end{align*}
    We obtain $(x,f_0(x)) \cdot \vec{n} > \left(1 - \frac{x-x'}{\firstdim{\vecprime{q}} - x'}\right)  (x,f_0(x)) \cdot \vec{n} + \frac{x - x'}{\firstdim{\vecprime{q}} - x'} (x,f_0(x)) \cdot \vec{n}  = (x,f_0(x)) \cdot \vec{n}$, which is a contradiction.
  \end{itemize}

  \medskip

  We now turn to the case of $\max$.
  By the characterization of the Pareto curve, there are $\vec{q}, \vec{r} \in \achievable_{s_1} \cup \achievable_{s_2}$ 
  such that $(x,f_0(x))$ is a convex combination of $\vec{q}$ and $\vec{r}$.
  Let $f',f'' \in \{f_1,f_2\}$, be such that $\seconddim{\vec{q}} \le f'(\firstdim{\vec{q}})$ and $\seconddim{\vec{r}} \le f''(\firstdim{\vec{r}})$.
  Let also $\lambda_1 \in [0,1]$ be such that $(x,f_0(x)) = \lambda_1 \cdot \vec{q} + (1 - \lambda_1) \cdot \vec{r}$.
  We first prove $\lslope(f_0,x) \le \lslope(f',\firstdim{\vec{q}})$ (and obtain similarly $\lslope(f_0,x) \le \lslope(f'',\firstdim{\vec{r}})$).
  %Assume towards a contradiction that $\lslope(f_0,x) > \lslope(f',x_1)$.
  By the characterisation of the Pareto curve, for all $\varepsilon \ne 0$, $f_0(x-\varepsilon) \ge \lambda_1 \cdot f'(\firstdim{\vec{q}} - \frac{\varepsilon}{\lambda_1}) + (1 - \lambda_1) \cdot f''(\firstdim{\vec{r}})$.
  So for all $\varepsilon > 0$, $\frac{f_0(x) - f_0(x-\varepsilon)}{\varepsilon} \le \lambda_1 \cdot \frac{f'(\firstdim{\vec{q}}) - f'(\firstdim{\vec{q}} - \frac{\varepsilon}{\lambda_1})}{\varepsilon}$.
  The limit when $\varepsilon$ goes towards $0$ is smaller than that of $\frac{f'(\firstdim{\vec{q}}) - f'(\firstdim{\vec{q}} - \varepsilon)}{\varepsilon}$.
  Hence, $\lslope(f_0,x) \le \lslope(f',\firstdim{\vec{q}})$.

  Now, since their convex combination contains $x$, one of $\firstdim{\vec{q}}$ and $\firstdim{\vec{r}}$ is greater than $x$.
  Assume without loss of generality that it is $\firstdim{\vec{q}}$.
  Then $\firstdim{\vec{q}} \ge x$ and by Lem.~\ref{lem:slope-decrease}, $\lslope(f',x) \ge \lslope(f',\firstdim{\vec{q}}) \ge \lslope(f_0,x)$.
  This shows $\lslope(f_0,x) \le \max(\lslope(f_1,x),\lslope(f_2,x))$.

  \medskip

  We now turn to the case of $\rslope$.
  This is in fact the same proof if we consider the function $g_0 \colon x \mapsto f_0(1 - x)$, which is also concave and for which $\lslope(g_0,x) = -\rslope(f_0,1-x)$: at no point did we use the fact that the slope was negative.
  So what we proved for $f_0$ is also valid for $g_0$, which means $\min(\rslope(f_{s_1},x), \rslope(f_{s_2},x)) \le \rslope(f_{s_0},x) \le \max(\rslope(f_{s_1},x), \rslope(f_{s_2},x)) $.
\qed
\end{proof}

%We define $C_\varepsilon(x)$ the slope at position $x-\varepsilon$.
We are now ready to prove Lem.~\ref{lem:slope-convex-union}.
\begin{reflemma}{lem:slope-convex-union}
  Let $s_0$ be a \PONE{} state with two successors $s_1$ and $s_2$, and let $\vec{p}$ be a left \mbox{accumulation} point of $f_{s_0}$.
  Then there is  $\eta(s_0,\vec{p}) > 0$ such that
  for all $\varepsilon \in (0, \eta(s_0,\vec{p}\ )\openend$,
  %for all $x' \in [\firstdim{\vec{p}}- \eta(s_0,\vec{p}),\firstdim{\vec{p}}]$,
  there is $s' \in \{ s_1, s_2 \}$ such that:
  \begin{inparaenum}
  \item $\vec{p}$ is a left accumulation point in $f_{s'}$; 
  \item $\lslope(s_0,\firstdim{\vec{p}}) = \lslope(s',\firstdim{\vec{p}})$;
  \item $f_{s_0}(\firstdim{\vec{p}} - \varepsilon) \ge f_{s'}(\firstdim{\vec{p}} - \varepsilon)$ and $\rslope(f_{s_0},\firstdim{\vec{p}} - \varepsilon)  \ge \rslope(f_{s'},\firstdim{\vec{p}} - \varepsilon)$. 
  \end{inparaenum}
%%%%% OLD VERSION
  %% Let $s_0$ be a \PONE{} state with two successors $s_1$ and $s_2$.
  %% If $\vec{p}=(x,y)$ is a left accumulation point, then there is some $s' \in \{s_1,s_2\}$ such that:
  %% \begin{inparaenum}
  %% \item $s'$ has a left accumulation point $\vec{p}$; 
  %% \item $\lslope(s_0,x) = \lslope(s',x)$; 
  %% \item and there is $\eta(s_0,\vec{p}) > 0$ such that for all $x' \in [x- \eta(s_0,\vec{p}),x]$, there is $s' \in \{ s_1, s_2 \}$ which satisfies conditions \ref{lem:slope-convex-union-a} and \ref{lem:slope-convex-union-b} and such that $f_{s_0}(x') \ge f_{s'}(x')$ and $\rslope(f_{s_0},x')  \ge \rslope(f_{s'},x')$. 
\end{reflemma}
\begin{proof}
  By the characterisation of the Pareto curve, the extremal points of the Pareto curve $f_{s_0}$ are included in the extremal points of $f_{s_1}$ and $f_{s_2}$.
  Let $(\vecsup{p}{i})_{i\in \mathbb{N}}$ be a sequence of extremal points of the Pareto curve in $s_0$ which converges to $\vec{p}$.
  We assume that the %$x$ 
  first coordinate of the sequence is increasing (note that we can always extract a sub-sequence which satisfies that).
  Each $\vecsup{p}{i}$ is either an extremal point of $s_1$ or of $s_2$, therefore the Pareto curve of one of the two contains infinitely many points $\vecsup{p}{i}$.
  We first show that for any $s'$ for which this holds, the first two points are true for $s'$.
  We will then show that the third point is satisfied by one such $s'$.

  \point{\ref{lem:slope-convex-union-a}} Since there is a subsequence of $\vecsup{p}{i}$ which are extremal points of $s'$ and which converges to $\vec{p}$ from the left, $\vec{p}$ is a left accumulation point of $s'$.
  
  \point{\ref{lem:slope-convex-union-b}} Moreover since $\lim_{x' \to \firstdim{\vec{p}}^-} \frac{f_{s_0}(x') - f_{s_0}(\firstdim{\vec{p}})}{x' - \firstdim{\vec{p}}}$ is well defined and $(\firstdim{p^i})_{i \in \mathbb{N}}$ converges to $\firstdim{\vec{p}}$, we have that:
  \begin{align*}
    \lslope(s_0,\firstdim{\vec{p}}) & = 
    \lim_{x' \to \firstdim{\vec{p}}^-} \frac{f_{s_0}(x') - f_{s_0}(\firstdim{\vec{p}})}{x' - \firstdim{\vec{p}}} \\
    & = \lim_{i \to \infty} \frac{f_{s_0}(\firstdim{p^i}) - f_{s_0}(\firstdim{\vec{p}})}{\firstdim{p^i} - \firstdim{\vec{p}}} \\ 
    & = \lim_{i \to \infty} \frac{f_{s'}(\firstdim{p^i}) - f_{s'}(\firstdim{\vec{p}})}{\firstdim{p^i} - \firstdim{\vec{p}}} \\ 
    & = \lslope(s',\firstdim{\vec{p}})
  \end{align*}

  \point{\ref{lem:slope-convex-union-c}}
  Let $x' < \firstdim{\vec{p}}$.
  If $\vec{p}$ is not a left accumulation point for $s_1$, let $x_1 = \sup\{ x_1 < \firstdim{\vec{p}} \mid (x_1,f_{s_1}(x_1)) \text{ extremal in } f_{s_1} \}$ and let $\vec{q}$ be an extremal point of $s_0$ with $\firstdim{\vec{q}} \in \openbegin x_1, \firstdim{\vec{p}}\openend$.
  Since extremal points of $s_0$ are included in those of $s_1$ and $s_2$, the extremal points on $[\firstdim{\vec{q}} , x]$ and $x$ are the same.
  Since moreover $(\firstdim{\vec{q}},f_0(\firstdim{\vec{q}}))$ and $(x,f_0(x))$ are extremal points of both $f_0$ and $f_2$, the curves $f_0$ and $f_{2}$ coincide on $[\firstdim{\vec{q}} , x]$.
  So the property we want is satisfied by $s' = s_2$ and $\eta(s_0,\vec{p}) = x - \firstdim{\vec{q}}$.
  
  Similarly, if $\vec{p}$ is not a left accumulation point for $s_2$, then $s' = s_1$ and some $\eta(s_0,\vec{p})$ witnesses the property.

  If $\vec{p}$ is a left accumulation point for $s_1$ but $\lslope(f_{s_1},\vec{p}) \ne \lslope(f_0,\vec{p})$, then by what was proven in point~\ref{lem:slope-convex-union-b}, we conclude that there is no infinite sequence of extremal points of $s_1$ that are also extremal points of $s_0$ and converge to $\vec{p}$.
  There is a left neighbourhood of $x$ where $f_1$ is below $f_2$ (otherwise some extremal point would be above and also be an extremal point of $f_0$).
  Hence, the curves $f_0$ and $f_{s_2}$ coincide on some neighbourhood of $x$.
  So the property we want is satisfied by $s' = s_2$ and some $\eta(s_0,\vec{p})$.
  
  Similarly, if  $\lslope(f_{s_1},\vec{p}) \ne \lslope(f_s,\vec{p})$, then $s' = s_1$ and some $\eta(s_0,\vec{p})$ witnesses the property.

  In the other cases, both $s_1$ and $s_2$ satisfy points~\ref{lem:slope-convex-union-a} and~\ref{lem:slope-convex-union-b}.
  By Lem.~\ref{lem:bound-slope-convex-union}, there is $s' \in \{s_1,s_2\}$, such that $\rslope(f_{s'},x') \le \rslope(f_{s_0},x')$.
  So the property is satisfied for any $\eta(s_0,\vec{p})$ such that $x - \eta(s_0,\vec{p})$ is still in the domain of $f_{s_0}$.
\qed
\end{proof}

\subsection{Evolution of the slope in \PTWO{} states (proof of Lem.~\ref{lem:slope-intersection})}

\begin{lemma}\label{lem:bound-slope-intersection}
  Let $s_0$ be a \PTWO{} state with two successors $s_1$ and $s_2$.
  For all $x$ where $f_{s_1}$ is defined and $f_{s_2}$ is defined, $f_{s_0}$ is defined in $x + \reward_1(s_0)$ and:
  %\[ \lslope(f_{s_0},x) \in \{ \lslope(f_{s_1},x), \lslope(f_{s_2},x) \} \] Moreover 
  \begin{inparaenum}
  \item
    if $f_{s_1}(x) < f_{s_2}(x)$ then $\lslope(f_{s_0},x+\reward_1(s_0)) = \lslope(f_{s_1},x)$ and $\rslope(f_{s_0},x+\reward_1(s_0)) = \rslope(f_{s_1},x)$; \label{lem:bound-slope-intersection-a}
  \item 
    if $f_{s_1}(x) = f_{s_2}(x)$ then $\lslope(f_{s_0},x+\reward_1(s_0)) = \max \{ \lslope(f_{s_1},x), \lslope(f_{s_2},x) \}$ and $\rslope(f_{s_0},x) = \min \{ \rslope(f_{s_1},x), \rslope(f_{s_2},x) \}$.\label{lem:bound-slope-intersection-b}
  \end{inparaenum}
\end{lemma}
\begin{proof}
  \point{\ref{lem:bound-slope-intersection-a}} Assume $f_{s_1}(x) < f_{s_2}(x)$.
  By continuity of the curves (Lem.~\ref{lem:f-lemma}.\ref{lem:f-continuity}), $f_{s_1}(x') < f_{s_2}(x')$ holds for all points $x'$ of some neighbourhood $[x- \eta,x+\eta]$.
  By the characterization of the Pareto curve in \PTWO{} states (see Sec.~\ref{sec:eqns}), we have $f_{s_0}(x'+\reward_1(s_0)) = \reward_2(s_0) + \min \{ f_{s_1}(x') , f_{s_2}(x')\}$.
  Hence, for all $x' \in [x - \eta, x+\eta]$, $f_{s_0}(x'+\reward_1(s_0)) = \reward_2(s_0) + f_{s_1}(x')$.
  So,
  %\begin{align*}
  $\lslope(f_{s_0},x + \reward_1(s_0)) = \lim_{x' \to x^-} \frac{\reward_2(s_0) + f_{s_0}(x') - f_{s_0}(x) - \reward_2(s_0)}{x' - x} = \lim_{x' \to x^-} \frac{f_{s_1}(x') - f_{s_1}(x)}{x' - x} = \lslope(f_{s_1},x)$.
  Similarly, $\rslope(f_{s_0},x + \reward_1(s_0)) = \lim_{x' \to x^+} \frac{f_{s_0}(x') - f_{s_0}(x)}{x' - x} = \lim_{x' \to x^+} \frac{f_{s_1}(x') - f_{s_1}(x)}{x' - x} = \rslope(f_{s_1},x)$.
%\end{align*} \fbox{proof for $\rslope$}

  \point{\ref{lem:bound-slope-intersection-b}} 
  If  $f_{s_1}(x) = f_{s_2}(x)$ and $\lslope(f_{s_1},x) = \lslope(f_{s_2},x)$ then for a sequence $(x_i)_{i \in \mathbb{N}}$ that converges to $x$ from the left, $\lim_{i \to \infty} \frac{f_{s_0}(x_i+\reward_1(s_0)) - f_{s_0}(x+\reward_1(s_0))}{x_i - x} = \lslope(f_{s_0},x+\reward_1(s_0))$.
  By the characterisation of the Pareto curve, $f_{s_0}(x_i+\reward_1(s_0)) \in \{\reward_2(s_0) + f_{s_1}(x_i), \reward_2(s_0) + f_{s_2}(x_i)\}$, so there are infinitely many $x_i$ in the sequence for which $\reward_2(s_0) + f_{s_1}(x_i) = f_{s_0}(x_i+\reward_1(s_0))$ or there are infinitely many $x_i$ in the sequence for which $\reward_2(s_0) + f_{s_2}(x_i) = f_{s_0}(x_i+\reward_1(s_0))$.
  Without loss of generality we assume it is for $s_1$.
  $\lslope(f_1,x) = \lim_{i \to \infty} \frac{f_{s_1}(x_i) - f_{s_1}(x)}{x_i - x} =$ 
  $\lim_{i \to \infty} \frac{f_{s_0}(x_i + \reward_1(s_0)) - f_{s_0}(x+\reward_1(s_0))}{x_i - x} = \lslope(f_{s_0},x+\reward_1(s_0))$.
  Hence, we proved that $ \lslope(f_{s_0},x+\reward_1(s_0)) \in \{ \reward_2(s_0) + \lslope(f_{s_1},x), \reward_2(s_0) + \lslope(f_{s_2},x) \}$.
  We could prove similarly, by considering a sequence $x_i$ that converges from the right, that $\rslope(f_{s_0},x+\reward_1(s_0)) \in \{ \reward_2(s_0) + \rslope(f_{s_1},x), \reward_2(s_0) + \rslope(f_{s_2},x) \}$.

  We will now show that if $f_{s_1}(x) = f_{s_2}(x)$ and $\lslope(f_{s_1},x) < \lslope(f_{s_2},x)$, then $\lslope(f_{s_0},x+\reward_1(s_0)) = \lslope(f_{s_2},x)$, which shows the property.
  Since:
  \[ \lim_{x' \to x^-} \frac{f_{s_1}(x') - f_{s_1}(x)}{x' - x} = \lslope(f_{s_1},x) < \lslope(f_{s_2},x) = \lim_{x' \to x^-} \frac{f_{s_2}(x') - f_{s_2}(x)}{x' - x},
  \]
  there is $\varepsilon > 0$, such that for all $x' \in [x - \varepsilon, x \openend$, $\frac{f_{s_1}(x') - f_{s_1}(x)}{x' - x} < \frac{f_{s_2}(x') - f_{s_2}(x)}{x' - x}$.
    This implies that $f_{s_1}(x') > f_{s_2}(x')$, because $x' - x < 0$ and therefore $f_{s_0}(x'+\reward_1(s_0)) = \reward_2(s_0) + f_{s_2}(x')$ (by characterization of the Pareto curve).
    Hence,
    \begin{align*}
      %\[
      \lslope(f_{s_0},x+\reward_1(s_0)) & = \lim_{x' \to x^-} \frac{f_{s_0}(x'+\reward_1(s_0)) - f_{s_0}(x+\reward_1(s_0))}{x' - x} \\
      & = \lim_{x' \to x^-} \frac{f_{s_2}(x') - f_{s_2}(x)}{x' - x} = \lslope(f_{s_2},x).
    \end{align*}

    The proof is quite similar for the right slope.
    Assume $\rslope(f_{s_1}(x)) < \rslope(f_{s_2},x)$.
    Since:
    \begin{align*}
      \lim_{x' \to x^+} \frac{f_{s_1}(x') - f_{s_1}(x)}{x' - x} & = \rslope(f_{s_1}(x)) < \rslope(f_{s_2},x) \\
      & = \lim_{x' \to x^+} \frac{f_{s_2}(x') - f_{s_2}(x)}{x' - x},
    \end{align*}
    there is $\varepsilon > 0$, such that for all $x' \in \openbegin x, x + \varepsilon \openend$, $\frac{f_{s_1}(x') - f_{s_1}(x)}{x' - x} < \frac{f_{s_2}(x') - f_{s_2}(x)}{x' - x}$.
    This implies that $f_{s_1}(x') < f_{s_2}(x')$, and therefore $f_{s_0}(x'+\reward_1(s_0)) = \reward_2(s_0) + f_{s_1}(x')$ (by characterisation of the Pareto curve).
    Hence,
    \begin{align*}
      \rslope(f_{s_0},x+\reward_1(s_0)) & = \lim_{x' \to x^+} \frac{f_{s_0}(x'+\reward_1(s_0)) - f_{s_0}(x+\reward_1(s_0))}{x' - x} \\
      & = \lim_{x' \to x^+} \frac{f_{s_2}(x') - f_{s_2}(x)}{x' - x} \\ 
      & = \rslope(f_{s_2},x).
    \end{align*}
\qed
\end{proof}

\begin{reflemma}{lem:slope-intersection}
  Let $s_0$ be a \PTWO{} state with two successors $s_1$ and $s_2$, and let $\vec{p}$ be a left \mbox{accumulation} point of $f_{s_0}$.
  There is $\eta(s_0,\vec{p}) > 0$ such that for all $\varepsilon \in (0,\eta(s_0,\vec{p}\ )\openend$, there is $s' \in \{s_1,s_2\}$, such that:
  \begin{inparaenum}
  \item $\vec{p} - \reward(s_0)$ is a left accumulation point in $f_{s'}$; 
  \item $\lslope(s_0,x) = \lslope(s',\firstdim{\vec{p}}-\reward_1(s_0))$; 
  \item  $f_{s_0}(\firstdim{\vec{p}}-\varepsilon) = f_{s'}(\firstdim{\vec{p}}-\varepsilon-\reward_1(s_0))$ and $\rslope(f_{s_0},\firstdim{\vec{p}}-\varepsilon) = \rslope(f_{s'},\firstdim{\vec{p}}-\varepsilon-\reward_1(s_0))$.
  %%% %there is $0 < \varepsilon < \eta$ such that $\lslope(s,x-\varepsilon) \ge \lslope(s',x-\varepsilon)$ and  $0 < \varepsilon' < \eta$ such that $\lslope(s,x-\varepsilon') \le \lslope(s',x-\varepsilon')$.
  %% Let $s_0$ be a \PTWO{} state with two successors $s_1$ and $s_2$.
  %% If $\vec{p}=(x,y)$ is a left accumulation point, then there is some $s' \in \{s_1,s_2\}$ such that all the following holds:
  %% \begin{inparaenum}
  %% \item $s'$ has a left accumulation point $\vec{p}-\reward(s_0)$; 
  %% \item $\lslope(s_0,x) = \lslope(s',x-\reward_1(s_0))$; 
  %% \item there is $\eta(s_0,\vec{p}) > 0$ such that for all $x' \in [x - \eta(s_0,\vec{p}) , x]$, there is $s' \in \{s_1,s_2\}$, such that $f_{s_0}(x') = \reward_2(s_0) + f_{s'}(x'-\reward_1(s_0))$ and $\rslope(f_{s_0},x') =  \rslope(f_{s'},x'-\reward_1(s_0))$.
%%    %\reward_2(s_0) +\vojta{Does $\reward_2(s_0)$ belong here?}.
  %%% %\item $f_{s_0} \equivx f_{s'}$. \label{lem:slope-intersection-c}
  \end{inparaenum}
\end{reflemma}
\begin{proof}  
  We first show that for each $x$ where $f_{s_0}$ is defined, there is $s' \in \{s_1,s_2\}$ such that $\lslope(s_0,x) = \lslope(s',x-\reward_1(s_0))$.
  Let $(x_i)_{i\in \mathbb{N}}$ be a sequence converging towards $x$.
  By the characterization of the Pareto curve, for each $i\in \mathbb{N}$, $f_{s_0}(x_i) \in \{ \reward_2(s_0) + f_{s_1}(x_i-\reward_1(s_0)) , \reward_2(s_0) + f_{s_2}(x_i-\reward_1(s_0)) \}$.
  It will be either $f_{s_1}$ or $f_{s_2}$ infinitely often. 
  We write $s'$ for a state such that $f_{s_0}(x_i) = \reward_2(s_0) + f_{s'}(x_i-\reward_1(s_0))$ infinitely often.
  We can extract a subsequence of $x_i$ so that we can assume $f_{s_0}(x_i) = \reward_2(s_0) + f_{s'}(x_i-\reward_1(s_0))$ for all $i$.
  Since this sequence converges to $x$, by continuity of $f_{s'}$ (Lem.~\ref{lem:f-continuity}) and $f_{s_0}$, we have $\reward_2(s_0) + f_{s'}(x-\reward_1(s_0)) = f_{s_0}(x)$.
  Moreover:
  %\begin{align*}
  $\lslope(s_0,x)  = \lim_{x' \to x^-} \frac{f_{s_0}(x') - f_{s_0}(x)}{x' - x} 
     = \lim_{i \to \infty}\frac{f_{s_0}(x_i) - f_{s_0}(x)}{x_i - x} 
     = \lim_{i \to \infty}\frac{f_{s'}(x_i-\reward_1(s_0)) - f_{s'}(x-\reward_1(s_0))}{x_i - x} 
     = \lslope(s',x)$.
%\end{align*}

  %, \frac{1}{2} \cdot \slope(s_1,x) + \frac{1}{2} \cdot \slope(s_2,x) \}$. 
     We now show we can extract $(x_{\tau(i)})_{i\in \mathbb{N}}$, an infinite subsequence of $(x_i)_{i\in\mathbb{N}}$, whose elements all have different slopes.
     Since $\vec{p}$ is a left accumulation point for $f_{s_0}$, for each $x_{\tau(i)}$, there are two extremal points for $f_{s_0}$ in $\openbegin x_{\tau(i)}, x\openend$.
  Because of Lem.~\ref{lem:general-lemma}.\ref{lem:extremal-slopes}, no more than two extremal points can have the same slope, the slope of second extremal point~$\vecprime{p}$ is strictly less than that of $x_{\tau(i)}$.
  Thus, by choosing $\tau(i+1)$ such that $x_{\tau(i+1)} > \firstdim{\vecprime{p}}$, we ensure that the slope of all $x_{\tau(i)}$ is different.
  This shows we can extract an infinite subsequence of $(x_i)_{i\in\mathbb{N}}$ whose all elements have different slopes.
  %there are infinitely many different slopes on $f_{s_0}$.
  To simplify notations, we will still write $(x_i)_{i\in \mathbb{N}}$ for this subsequence, and assume all slopes $\lslope(f_{s_0},x_i)$ are different.

  Because we showed $\lslope(s_0,x) \in \{\lslope(s_1,x-\reward_1(s_0)), \lslope(s_2,x-\reward_1(s_0)) \}$ for all points~$x$ where $f_{s_0}$ is defined, this is also the case for each $x_i$.
  We can extract an infinite subsequence of $(x_i)_{i\in \mathbb{N}}$ such that it always corresponds to the same state.
  We can then assume that there is $s'\in\{s_1,s_2\}$, such that for all $i\in\mathbb{N}$, $\lslope(s_0,x_i) = \lslope(s',x_i-\reward_1(s_0))$.
  Since the slopes of every $x_i$ is different, 
  by Lem.~\ref{lem:general-lemma}.\ref{lem:slope-extremals}, this implies that there is an infinite number of extremal points in the neighbourhood of $\vec{p}$.
  This means that $\vec{p}$ is a left accumulation point for either $s_1$ or $s_2$.

  %% This also means that there are infinitely many different slopes in either $f_{s_1}$ or $f_{s_2}$ because we showed $\lslope(s_0,x) \in \{ \reward_2(s_0) + \lslope(s_1,x-\reward_1(s_0)), \reward_2(s_0) + \lslope(s_2,x-\reward_1(s_0)) \}$\vojta{I don't get this. We proved equality of slope only in $x$, how can we derive that that there are infinitely many slopes?}.
  %% Moreover there is an infinite number of these slopes in the neighbourhood of $\vec{p}$.
  %% By Lem.~\ref{lem:general-lemma}.\ref{lem:slope-extremals}, this implies that there is an infinite number of extremal points in the neighbourhood of $\vec{p}$.
  %% This means that $\vec{p}$ is an accumulation point for either $s_1$ or $s_2$.

  % Point~\ref{lem:slope-intersection-c} is then a consequence of Lem.~\ref{lem:equiv-slope}.
  We now prove point~\ref{lem:slope-intersection-c}.
  %Let $x' < x$.
  If $\vec{p}$ is not a left accumulation point for $s_1$, let $x_1 = \sup\{ x_1 < x \mid (x_1,f_{s_1}(x_1)) \text{ extremal in } f_{s_1} \}$.
  %and let $\vec{q}$ be an extremal point of $s_0$ with $\firstdim{\vec{q}} \in \openbegin x_1, x\openend$.
  %For $x' \in [x_1,x]$, $f_{s_1}(x') = f_{s_1}(x) + (x' - x) \cdot \lslope(f_{s_1},x)$.
  %We have that $f_{s_2}(\firstdim{\vec{q}}) \le f_{s_1}(\firstdim{\vec{q}})$ because $f_{s_0}(\firstdim{\vec{q}}) = \reward_2(s_0) + f_{s_2}(\firstdim{\vec{q}}-\reward_1(s_0))$.
  Since $f_{s_1}$ and $f_{s_2}$ are concave their curves can only intersect twice in $[x_1,x]$.
  %(since $f_{s_1}$ follows a line on this interval).\review{Why is this true? The function $1-x^2$ is concave and intersects the x-axis twice, once in -1 and once in 1. However I guess you can use that it can only intersect twice. I believe that this is fine}
  Since $f_{s_0}$ has an infinite number of extremal points in this interval (shifted by $\reward_1(s_0)$) and $f_{s_1}$ has none, we can deduce from Lem.~\ref{lem:general-lemma}.\ref{lem:slope-extremals} and Lem.~\ref{lem:bound-slope-intersection}
  %it must be 
  that $f_{s_2}$ is below $f_{s_1}$ in a neighbourhood $[\firstdim{\vec{q}},x]$ of $x$.
  Hence, for all $x' \in [\firstdim{\vec{q}} , x]$, $f_{s_0}(x') = \reward_2(s_0) + f_{s_2}(x' - \reward_1(s_0))$.
  So the property we want is satisfied by $s' = s_2$ and $\eta(s_0,\vec{p}) = x - \firstdim{\vec{q}}$.
  
  Similarly, if $\vec{p}$ is not a left accumulation point for $s_2$, then $s' = s_1$ and some $\eta(s_0,\vec{p})$ witness the property.

  If $\vec{p}$ is a left accumulation point for $s_1$ but $\lslope(f_{s_1},\firstdim{\vec{p}}-\reward_1(s_0)) \ne \lslope(f_{s_0},\firstdim{\vec{p}})$, then we can conclude from Lem.~\ref{lem:bound-slope-intersection} that $\lslope(f_{s_1},\firstdim{\vec{p}}) <  \lslope(f_s,\firstdim{\vec{p}}) = \lslope(f_{s_2},\firstdim{\vec{p}})$.
  Hence, on some left neighbourhood of $x$, $f_{s_2}$ is strictly below $f_{s_1}$.
  This means that the curves $f_{s_0}$ and $f_{s_2}$ (shifted by $\reward(s_0)$) coincide on some neighbourhood of $x$.
  So the property we want is satisfied by $s' = s_2$ and some $\eta(s_0,\vec{p})$.
  
  Similarly, if  $\lslope(f_{s_1},\firstdim{\vec{p}}) \ne \lslope(f_{s_0},\firstdim{\vec{p}})$, then $s' = s_1$ and some $\eta(s_0,\vec{p})$ witnesses the property.

  In the other cases, both $s_1$ satisfies point \eqref{lem:slope-intersection-a} and \eqref{lem:slope-intersection-b}.
  By Lem.~\ref{lem:bound-slope-intersection}, there is $s' \in \{s_1,s_2\}$, such that $\rslope(f_{s'},x'-\reward_1(s_0)) \le \rslope(f_{s_0},x')$.
  So the property is satisfied for $s'$ and any $\eta(s_0,\vec{p})$ such that $x - \eta(s_0,\vec{p})$ is still in the domain of $f_{s_0}$.
\qed
\end{proof}

\section{Evolution of the slope in stochastic states (proof of Lem.~\ref{lem:accumulation-stochastic-two})}\label{sec:pareto-stochastic}

\begin{lemma}\label{lem:optimal-stochastic}
  Let $s_0$ be a stochastic state with two successors $s_1$ and $s_2$.
  If $\vec{p}$ is a point of $f_{s_0}$
  then there are $\vec{q} \in f_{s_1}$ and $\vec{r} \in f_{s_2}$ such that $\vec{p} = \tfunction(s_0,s_1) \cdot \vec{q} + \tfunction(s_0,s_2) \cdot \vec{r}$.
  Moreover if $\vecprime{p}$ is another point of $f_{s_0}$ with $\firstdim{\vec{p}} \le \firstdim{\vecprime{p}}$, then we can chose $\vec{q},\vecprime{q}\in f_{s_1}$, $\vec{r}, \vecprime{r}\in f_{s_2}$ such that 
  $\vec{p} = \tfunction(s_0,s_1) \cdot \vec{q} + \tfunction(s_0,s_2) \cdot \vec{r}$, 
  $\vecprime{p} = \tfunction(s_0,s_1) \cdot \vecprime{q} + \tfunction(s_0,s_2) \cdot \vecprime{r}$,
  $\firstdim{\vec{q}}\le \firstdim{\vecprime{q}}$ and $\firstdim{\vec{r}} \le \firstdim{\vecprime{r}}$.
\end{lemma}
\begin{proof}
  Because $\vec{p}$ is achievable in $s_0$, by characterization of the Pareto curve, there are $\vec{q}$ achievable in $s_1$ and $\vec{r}$ achievable in $s_2$ such that $\vec{p} = \tfunction(s_0,s_1) \cdot \vec{q} + \tfunction(s_0,s_2) \cdot \vec{r}$.
  We assume towards a contradiction that $\vec{q}$ is not from $f_{s_1}$ (i.e. it lies strictly below the Pareto curve); the proof works in the same way for $\vec{r}$ in $s_2$.
  %\vojta{$r$ or $r_1$?}
  Then there would be some $\vecprime{q} > \vec{q}$ that is achievable in $s_1$.
  By characterization of the Pareto curve $\vecprime{p} = \tfunction(s_0,s_1) \cdot \vecprime{q} + \tfunction(s_0,s_2) \cdot \vec{r}$ is also achievable. 
  However, we have $\vecprime{p} > \vec{p}$ which is a contradiction with the fact that $\vec{p}\in f_{s_0}$.
  This shows the first part of the lemma.

  Now for the second part of the lemma, assume that $\firstdim{\vec{p}} \le \firstdim{\vecprime{p}}$ and $\firstdim{\vec{q}} > \firstdim{\vecprime{q}}$; the proof would work the same way if $\firstdim{\vec{r}} > \firstdim{\vecprime{r}}$.
  Then we must have $\firstdim{\vec{r}} < \firstdim{\vecprime{r}}$ since $\tfunction(s_0,s_1) \cdot \firstdim{\vec{q}} + \tfunction(s_0,s_2) \cdot \firstdim{\vec{r}} = \firstdim{\vec{p}} \le \firstdim{\vecprime{p}} = \tfunction(s_0,s_1) \cdot \firstdim{\vecprime{q}} + \tfunction(s_0,s_2) \cdot \firstdim{\vecprime{r}} $.
  Let us write $m(\vec{q},\vec{r}) = \tfunction(s_0,s_1) \cdot \vec{q} + \tfunction(s_0,s_2) \cdot \vec{r}$ to simplify notation.
  By monotonicity of $m$, $m(\firstdim{\vec{q}},\firstdim{\vecprime{r}}) > m(\firstdim{\vecprime{q}},\firstdim{\vecprime{r}}) = \firstdim{\vecprime{p}} \ge \firstdim{\vec{p}} = m(\firstdim{\vec{q}},\firstdim{\vec{r}})$.
  By continuity of $m$, there is $x' \in [\firstdim{\vec{r}},\firstdim{\vecprime{r}}]$ such that $m(\firstdim{\vec{q}},x') = \firstdim{\vecprime{p}}$ and similarly $x \in [\firstdim{\vec{r}},\firstdim{\vecprime{r}}]$ such that $m(\firstdim{\vecprime{q}},x) = \firstdim{\vec{p}}$.
  %We let $q_1' = q_1$, $q_2' = q_2$, $r_1' = r_1$.

  Let us show $x - \firstdim{\vec{r}} = \firstdim{\vecprime{r}} - x'$.
  \begin{align*}
    \tfunction(s_0,s_1) \cdot \firstdim{\vecprime{q}} + \tfunction(s_0,s_2) \cdot x & = m(\firstdim{\vec{q}},x') = \firstdim{\vec{p}} = \tfunction(s_0,s_1) \cdot \firstdim{\vec{q}} + \tfunction(s_0,s_2) \cdot \firstdim{\vec{r}}\\
     \tfunction(s_0,s_2) \cdot (x - \firstdim{\vec{r}}) & = \tfunction(s_0,s_1) \cdot (\firstdim{\vec{q}} - \firstdim{\vecprime{q}})
  \end{align*}
  Similarly: 
  \begin{align*}
    \tfunction(s_0,s_1) \cdot \firstdim{\vec{q}} + \tfunction(s_0,s_2) \cdot x' & = \firstdim{\vecprime{p}} 
    = \tfunction(s_0,s_1) \cdot \firstdim{\vecprime{q}} + \tfunction(s_0,s_2) \cdot \firstdim{\vecprime{r}}\\
    \tfunction(s_0,s_2) \cdot (x' - \firstdim{\vecprime{r}}) & = \tfunction(s_0,s_1) \cdot (\firstdim{\vecprime{q}} - \firstdim{\vec{q}})
  \end{align*}
  Let us write $\alpha = \frac{\tfunction(s_0,s_1)}{\tfunction(s_0,s_2)} \cdot (\firstdim{\vec{q}} - \firstdim{\vecprime{q}}) = x - \firstdim{\vec{r}} = \firstdim{\vecprime{r}} - x'$.
  %% \begin{center}
  %%   \begin{tikzpicture}
  %%   \draw (0,3) node[fill,inner sep=2pt] {} node[above] {$\vec{q}_2$} -- (4,3) node[fill,inner sep=2pt] {} node[above] {$\vec{q}_1$};
  %%   \draw (1,0) node[fill,inner sep=2pt] {} node[below] {$\vec{r}_1$} -- (5,0) node[fill,inner sep=2pt] {} node[above] {$\vec{r}_2$};
  %%   \draw[dotted] (0,3) -- (5,0);
  %%   \draw[dotted] (4,3) -- (1,0);
  %%   \draw[dashed] (0,1) -- (6,1);
  %%   \draw[<->] (5.5,3) -- node[right]{$\tfunction(s_0,s_1)$} (5.5,1);
  %%   \draw[<->] (5.5,1) -- node[right] {$\tfunction(s_0,s_2)$} (5.5,0); 
  %%   \draw (2,1) node[fill,inner sep=2pt]{} node [above] {$x_1$};
  %%   \draw (10/3,1) node[fill,inner sep=2pt]{} node [above] {$x_2$};
  %% \end{tikzpicture}
  %% \end{center}
  By concavity of $f_{s_2}$:
  \begin{align*} 
    f_{s_2}(x') & \ge f_{s_2}(\firstdim{\vecprime{r}}) + (x' - \firstdim{\vecprime{r}}) \cdot \frac{f_{s_2}(\firstdim{\vec{r}}) - f_{s_2}(\firstdim{\vecprime{r}})}{\firstdim{\vec{r}}-\firstdim{\vecprime{r}}} \\
    f_{s_2}(x') - f_{s_2}(\firstdim{\vecprime{r}}) & \ge - \alpha \cdot \frac{f_{s_2}(\firstdim{\vec{r}}) - f_{s_2}(\firstdim{\vecprime{r}})}{\firstdim{\vec{r}}-\firstdim{\vecprime{r}}} \\
    m\left(f_{s_1}(\firstdim{\vec{q}}) - f_{s_1}(\firstdim{\vecprime{q}}), f_{s_2}(x') - f_{s_2}(\firstdim{\vecprime{r}})\right)
      & \ge m\left(f_{s_1}(\firstdim{\vec{q}}) - f_{s_1}(\firstdim{\vecprime{q}}), - \alpha \cdot \frac{f_{s_2}(\firstdim{\vec{r}}) - f_{s_2}(\firstdim{\vecprime{r}})}{\firstdim{\vec{r}}-\firstdim{\vecprime{r}}}\right)
  \end{align*}
  Similarly:
  \begin{align*} 
    f_{s_2}(x) & \ge f_{s_2}(\firstdim{\vec{r}}) + (x - \firstdim{\vec{r}}) \cdot \frac{f_{s_2}(\firstdim{\vec{r}}) - f_{s_2}(\firstdim{\vecprime{r}})}{\firstdim{\vec{r}}-\firstdim{\vecprime{r}}} \\
    f_{s_2}(x) - f_{s_2}(\firstdim{\vec{r}}) & \ge \alpha \cdot \frac{f_{s_2}(\firstdim{\vec{r}}) - f_{s_2}(\firstdim{\vecprime{r}})}{\firstdim{\vec{r}}-\firstdim{\vecprime{r}}} \\
      m\left(f_{s_1}(\firstdim{\vecprime{q}}) - f_{s_1}(\firstdim{\vec{q}}), f_{s_2}(x) - f_{s_2}(\firstdim{\vec{r}})\right)
      & \ge m\left(f_{s_1}(\firstdim{\vecprime{q}}) - f_{s_1}(\firstdim{\vec{q}}), \alpha \cdot \frac{f_{s_2}(\firstdim{\vec{r}}) - f_{s_2}(\firstdim{\vecprime{r}})}{\firstdim{\vec{r}}-\firstdim{\vecprime{r}}}\right)\\
      & \ge -m\left(f_{s_1}(\firstdim{\vec{q}}) - f_{s_1}(\firstdim{\vecprime{q}}), - \alpha \cdot \frac{f_{s_2}(\firstdim{\vec{r}}) - f_{s_2}(\firstdim{\vecprime{r}})}{\firstdim{\vec{r}}-\firstdim{\vecprime{r}}}\right)
  \end{align*}
  Hence, either $m\left(f_{s_1}(\firstdim{\vec{q}}) - f_{s_1}(\firstdim{\vecprime{q}}), f_{s_2}(x') - f_{s_2}(\firstdim{\vecprime{r}})\right) \ge 0$ or $m(f_{s_1}(\firstdim{\vecprime{q}}) - f_{s_1}(\firstdim{\vec{q}}),$ $f_{s_2}(x) - f_{s_2}(\firstdim{\vec{r}})) \ge 0$.
  In the first case $ m\left(f_{s_1}(\firstdim{\vec{q}}), f_{s_2}(x')\right) \ge m\left( f_{s_1}(\firstdim{\vecprime{q}}), f_{s_2}(\firstdim{\vecprime{r}})\right) = \seconddim{\vecprime{p}}$.
  So we could have chosen $\vec{q}$ and $\vecsup{r}{\prime\prime}=(x',f_{s_2}(x'))$ instead of $\vecprime{q}$ and $\vecprime{r}$ respectively, and all the properties required in the lemma are satisfied:
  $\vec{p}=m(\vec{q},\vec{r})$, $\vecprime{p} = m(\vec{q}, \vecsup{r}{\prime\prime})$, $\firstdim{\vec{q}} \le \firstdim{\vec{q}}$ and $\firstdim{\vec{r}}\le \firstdim{\vecsup{r}{\prime\prime}}$.

  In the second case
  $m\left(f_{s_1}(\firstdim{\vecprime{q}}), f_{s_2}(x) \right) \ge m\left(f_{s_1}(\firstdim{\vec{q}}), f_{s_2}(\firstdim{\vec{r}})\right) = \seconddim{\vec{p}}$.
  So we could have chosen $\vecprime{q}$ and $\vecsup{r}{\prime\prime}=(x,f_{s_2}(x))$ instead of $\vec{q}$ and $\vec{r}$ respectively, and all the properties required in the lemma are satisfied:
  $\vec{p}=m(\vecprime{q},\vecsup{r}{\prime\prime})$, $\vecprime{p} = m(\vecprime{q}, \vecprime{r})$, $\firstdim{\vecprime{q}} \le \firstdim{\vecprime{q}}$ and $\firstdim{\vecsup{r}{\prime\prime}}\le \firstdim{\vecprime{r}}$.
\qed
\end{proof}

\begin{lemma} \label{lem:slope-optimal-stochastic}
  Let $s_0$ be a stochastic state with two successors $s_1$ and $s_2$.
  If $\vec{p}$ is a point of $f_{s_0}$
  and $f_{s_1}$ has a point $\vec{q}$ and $f_{s_2}$ a point $\vec{r}$ such that $\vec{p} = \tfunction(s_0,s_1) \cdot \vec{q} + \tfunction(s_0,s_2) \cdot \vec{r}$, 
  %\vojta{lhs is a number, rhs is a vector?}, 
  then if $\lslope(s_1,\firstdim{\vec{q}})$ and $\lslope(s_2,\firstdim{\vec{r}})$ are defined then $\lslope(s_0,\firstdim{\vec{p}}) = \min(\lslope(s_1,\firstdim{\vec{q}}), \lslope(s_2,\firstdim{\vec{r}}))$.
\end{lemma}
\begin{proof}
  %\fbox{Check the proof doesn't use extremality}
  We first show that $\lslope(s_0,\firstdim{\vec{p}}) \le \lslope(s_1,\firstdim{\vec{q}})$.
  We have that for all $\varepsilon$: $f_{s_0}(\firstdim{\vec{p}} - \varepsilon) \ge \tfunction(s_0,s_1) \cdot f_{s_1}(\firstdim{\vec{q}} - \frac{\varepsilon}{\tfunction(s_0,s_1)}) + \tfunction(s_0,s_2) \cdot f_{s_2}(\firstdim{\vec{r}})$ because of the characterisation of the Pareto curve and the fact that $\firstdim{\vec{p}} - \varepsilon = \tfunction(s_0,s_1) \cdot (\firstdim{\vec{q}} - \frac{\varepsilon}{\tfunction(s_0,s_1)}) + \tfunction(s_0,s_2) \cdot \firstdim{\vec{r}}$.
  Therefore:
  \begin{align*}
    \lim_{\varepsilon \to 0^+} \frac{f_{s_0}(\firstdim{\vec{p}})- f_{s_0}(\firstdim{\vec{p}} - \varepsilon)}{\varepsilon} 
    & = \lim_{\varepsilon \to 0^+} \frac{\tfunction(s_0,s_1) \cdot f_{s_1}(\firstdim{\vec{q}}) + \tfunction(s_0,s_2) \cdot f_{s_2}(\firstdim{\vec{r}}) - f_{s_0}(\firstdim{\vec{p}} - \varepsilon)}{\varepsilon} \\
    & \le \lim_{\varepsilon \to 0^+} \frac{\tfunction(s_0,s_1) \cdot (f_{s_1}(\firstdim{\vec{q}}) - f_{s_1}(\firstdim{\vec{q}} - \frac{\varepsilon}{\tfunction(s_0,s_1)}))}{\varepsilon}\\
    & \le \lim_{\varepsilon \to 0^+} \frac{ (f_{s_1}(\firstdim{\vec{q}}) - f_{s_1}(\firstdim{\vec{q}} - \frac{\varepsilon}{\tfunction(s_0,s_1)}))}{\frac{\varepsilon}{\tfunction(s_0,s_1)} }\\
    & \le \lim_{\varepsilon \to 0^+} \frac{f_{s_1}(\firstdim{\vec{q}}) - f_{s_1}(\firstdim{\vec{q}} - \varepsilon)}{\varepsilon}
  \end{align*}
  %The proof works similarly for $\varepsilon \to 0^+$ and 
  Therefore, $\lslope(s_0,\firstdim{\vec{p}}) \le \lslope(s_1,\firstdim{\vec{q}})$.
  The proof also works for $\vec{r}$ and therefore $\lslope(s_0,\firstdim{\vec{p}}) \le \min(\lslope(s_1,\firstdim{\vec{q}}),\lslope(s_2,\firstdim{\vec{r}}))$.

  We now show that $\lslope(s_0,\firstdim{\vec{p}}) \ge \min(\lslope(s_1,\firstdim{\vec{q}}),  \lslope(s_2,\firstdim{\vec{r}}))$.
  Assume towards a contradiction that $\lslope(s_0,\firstdim{\vec{p}}) < \min(\lslope(s_1,\firstdim{\vec{q}}),  \lslope(s_2,\firstdim{\vec{r}}))$.
  By definition of $\lslope$ there exists $x' < \firstdim{\vec{p}}$ such that:
  \[ \frac{f_{s_0}(x') - f_{s_0}(\firstdim{\vec{p}})}{x' - \firstdim{\vec{p}}} < \min(\lslope(s_1,\firstdim{\vec{q}}),  \lslope(s_2,\firstdim{\vec{r}}))\]
  By Lem.~\ref{lem:optimal-stochastic}, there are $\vecprime{q}\in f_{s_1}$ and $\vecprime{r}\in f_{s_2}$ such that $(x',f_{s_0}(x')) = \tfunction(s_0,s_1) \cdot \vecprime{q} + \tfunction(s_0,s_2) \cdot \vecprime{r}$ and $\firstdim{\vecprime{q}} < \firstdim{\vec{q}}, \firstdim{\vecprime{r}} < \firstdim{\vec{r}}$.
  This gives:
  \begin{align*} 
    \min(\lslope(s_1,\firstdim{\vec{q}}),  \lslope(s_2,\firstdim{\vec{r}})) & > \frac{f_{s_0}(x') - f_{s_0}(\firstdim{\vec{p}})}{x' - \firstdim{\vec{p}}} \\ 
    & > \frac{\tfunction(s_0,s_1) \cdot \seconddim{\vecprime{q}} + \tfunction(s_0,s_2) \cdot \seconddim{\vecprime{r}} - \tfunction(s_0,s_1) \cdot \seconddim{\vec{q}} - \tfunction(s_0,s_2) \cdot \seconddim{\vec{r}}}{\tfunction(s_0,s_1) \cdot \firstdim{\vecprime{q}} + \tfunction(s_0,s_2) \cdot \firstdim{\vecprime{r}} - \tfunction(s_0,s_1) \cdot \firstdim{\vec{q}} - \tfunction(s_0,s_2) \cdot \firstdim{\vec{r}}} \\
    & > \frac{\tfunction(s_0,s_1) \cdot (\seconddim{\vecprime{q}} - \seconddim{\vec{q}}) + \tfunction(s_0,s_2) \cdot (\seconddim{\vecprime{r}} - \seconddim{\vec{r}})}{\tfunction(s_0,s_1) \cdot (\firstdim{\vecprime{q}} - \firstdim{\vec{q}}) + \tfunction(s_0,s_2) \cdot (\firstdim{\vecprime{r}} - \firstdim{\vec{r}})} 
  \end{align*}
  By concavity of the Pareto curves:
  %\vojta{concavity?} 
  \begin{align*}
    \frac{\seconddim{\vecprime{q}} - \seconddim{\vec{q}}}{\firstdim{\vecprime{q}} - \firstdim{\vec{q}}} & \ge \lslope(s_1,\firstdim{\vec{q}}) \\
    \seconddim{\vecprime{q}} - \seconddim{\vec{q}} & \le \lslope(s_1,\firstdim{\vec{q}}) \cdot (\firstdim{\vecprime{q}} - \firstdim{\vec{q}}) & \text{because $\firstdim{\vecprime{q}} - \firstdim{\vec{q}}$ is negative} \\
    \frac{\seconddim{\vecprime{r}} - \seconddim{\vec{r}}}{\firstdim{\vecprime{r}} - \firstdim{\vec{r}}} & \ge \lslope(s_2,\firstdim{\vec{r}}) \\
    \seconddim{\vecprime{r}} - \seconddim{\vec{r}} & \le \lslope(s_2,\firstdim{\vec{r}}) \cdot (\firstdim{\vecprime{r}} - \firstdim{\vec{r}})
  \end{align*}
  And since $\tfunction(s_0,s_1) \cdot (\firstdim{\vecprime{q}} - \firstdim{\vec{q}}) + \tfunction(s_0,s_2) \cdot (\firstdim{\vecprime{r}} - \firstdim{\vec{r}})$ is negative:
  \begin{align*}
    & \frac{\tfunction(s_0,s_1) \cdot (\seconddim{\vecprime{q}} - \seconddim{\vec{q}}) + \tfunction(s_0,s_2) \cdot (\seconddim{\vecprime{r}} - \seconddim{\vec{r}})}{\tfunction(s_0,s_1) \cdot (\firstdim{\vecprime{q}} - \firstdim{\vec{q}}) + \tfunction(s_0,s_2) \cdot (\firstdim{\vecprime{r}} - \firstdim{\vec{r}})}\\
    &\ge 
    \frac{\tfunction(s_0,s_1) \cdot \lslope(s_1,\firstdim{\vec{q}}) \cdot (\firstdim{\vecprime{q}} - \firstdim{\vec{q}}) + \tfunction(s_0,s_2) \cdot \lslope(s_2,\firstdim{\vec{r}}) \cdot (\seconddim{\vecprime{r}} - \seconddim{\vec{r}})}{\tfunction(s_0,s_1) \cdot (\firstdim{\vecprime{q}} - \firstdim{\vec{q}}) + \tfunction(s_0,s_2) \cdot (\firstdim{\vecprime{r}} - \firstdim{\vec{r}})} \\
    & \ge 
    \frac{\tfunction(s_0,s_1) \cdot \min(\lslope(s_1,\firstdim{\vec{q}}),\lslope(s_2,\firstdim{\vec{r}})) \cdot (\firstdim{\vecprime{q}} - \firstdim{\vec{q}})}{\tfunction(s_0,s_1) \cdot (\firstdim{\vecprime{q}} - \firstdim{\vec{q}}) + \tfunction(s_0,s_2) \cdot (\firstdim{\vecprime{r}} - \firstdim{\vec{r}})} \\ & ~~~ + \frac{\tfunction(s_0,s_2) \cdot \min(\lslope(s_1,\firstdim{\vec{q}}),\lslope(s_2,\firstdim{\vec{r}})) \cdot (\seconddim{\vecprime{r}} - \seconddim{\vec{r}})}{\tfunction(s_0,s_1) \cdot (\firstdim{\vecprime{q}} - \firstdim{\vec{q}}) + \tfunction(s_0,s_2) \cdot (\firstdim{\vecprime{r}} - \firstdim{\vec{r}})} \\
    & \ge  \min(\lslope(s_1,\firstdim{\vec{q}}),\lslope(s_2,\firstdim{\vec{r}}))
  \end{align*}
  This would imply 
  $\min(\lslope(s_1,\firstdim{\vec{q}}),  \lslope(s_2,\firstdim{\vec{r}})) > \min(\lslope(s_1,\firstdim{\vec{q}}),  \lslope(s_2,\firstdim{\vec{r}})) $, hence a contradiction.
  Therefore, $\lslope(s_0,\firstdim{\vec{p}}) \ge \min(\lslope(s_1,\firstdim{\vec{q}}),  \lslope(s_2,\firstdim{\vec{r}}))$, which concludes the proof.
\qed
\end{proof}

\begin{lemma}\label{lem:slope-stochastic}
  Let $s_0$ be a stochastic state with two successors $s_1$ and $s_2$.
  If $\vec{p}$ is an extremal point in ${s_0}$
  then $s_1$ has an extremal point $\vec{q}$ and $s_2$ has an extremal point $\vec{r}$ such that $\vec{p} = \tfunction(s_0,s_1) \cdot \vec{q} + \tfunction(s_0,s_2) \cdot \vec{r}$.
  Moreover $\vec{q}$ and $\vec{r}$ are uniquely defined and if $\lslope(s_1,\firstdim{\vec{q}})$ and $\lslope(s_2,\firstdim{\vec{r}})$ are defined then $\lslope(s_0,\firstdim{\vec{p}}) = \min(\lslope(s_1,\firstdim{\vec{q}}), \lslope(s_2,\firstdim{\vec{r}}))$.
  If only one of the two slopes is defined then $\lslope(s_0,\firstdim{\vec{p}})$ is equal to that slope, and if none of the two is defined then $\lslope(s_0,\firstdim{\vec{p}})$ is undefined.
\end{lemma}
\begin{proof}
  By the fixpoint characterisation of the achievable points, $\vec{p} = \tfunction(s_0,s_1) \cdot \vec{q} + \tfunction(s_0,s_2) \cdot \vec{r}$ for some $\vec{r}$ achievable in $s_2$ and $\vec{q}$ achievable in $s_1$.

  We first show that $\vec{q}$ is extremal in $s_1$ and $\vec{r}$ is extremal in $s_2$.
  Assume towards a contradiction that there are $\vecsup{q}{1}, \vecsup{q}{2}$ achievable from $s_1$ such that $\vec{q} \in \conv(\vecsup{q}{1},\vecsup{q}{2})$, then $\vec{p} \in \conv( \tfunction(s_0,s_1) \cdot \vecsup{q}{1} + \tfunction(s_0,s_2) \cdot \vec{r},  \tfunction(s_0,s_1) \cdot \vecsup{q}{2} + \tfunction(s_0,s_2) \cdot \vec{r})$ and both these points are achievable from $s_0$ because of the characterization of the Pareto curve.
  %\vojta{what is $r_i$?}
  This contradicts that $\vec{p}$ is an extremal point.
  The proof works similarly for $\vec{r}$, and therefore $\vec{q}$ and $\vec{r}$ are extremal in $f_{s_1}$ and $f_{s_2}$ respectively.

  We now prove that $\vec{q}$ and $\vec{r}$ are uniquely defined. 
  Assume towards a contradiction that there are $(\vecsup{q}{1},\vecsup{r}{1}) \ne (\vecsup{q}{2},\vecsup{r}{2})$ 
  %\vojta{should we also consider the case when only one takes place?} $\vecsup{r}{1} \ne \vecsup{r}{2}$ 
  with $\vecsup{q}{1}$ and $\vecsup{q}{2}$ extremal in ${s_1}$ and $\vecsup{r}{1}$ and $\vecsup{r}{2}$ extremal in ${s_2}$ such that $\vec{p} = \tfunction(s_0,s_1) \cdot \vecsup{q}{1} + \tfunction(s_0,s_2) \cdot \vecsup{r}{1}$ and $\vec{p} = \tfunction(s_0,s_1) \cdot \vecsup{q}{2} + \tfunction(s_0,s_2) \cdot \vecsup{r}{2}$.
  Note that $\vecsup{q}{1}\ne \vecsup{q}{2}$ and $\vecsup{r}{1}\ne \vecsup{r}{2}$.
  %  Assume towards a contradiction that there are $\vecsup{q}{1} \ne \vecsup{q}{2}$ and $\vecsup{r}{1} \ne \vecsup{r}{2}$ that are extremal in ${s_1}$ and ${s_2}$ respectively such that $\vec{p} = \tfunction(s_0,s_1) \cdot \vecsup{q}{1} + \tfunction(s_0,s_2) \cdot \vecsup{r}{1}$ and $\vec{p} = \tfunction(s_0,s_1) \cdot \vecsup{q}{2} + \tfunction(s_0,s_2) \cdot \vecsup{r}{2}$.
  Then we would have $\vec{p} = \frac{1}{2} \cdot (\tfunction(s_0,s_1) \cdot \vecsup{q}{1} + \tfunction(s_0,s_2) \cdot \vecsup{r}{2}) + \frac{1}{2} \cdot (\tfunction(s_0,s_1) \cdot \vecsup{q}{2} + \tfunction(s_0,s_2) \cdot \vecsup{r}{1})$.
  This means that $\vec{p} \in \conv(\tfunction(s_0,s_1) \cdot \vecsup{q}{1} + \tfunction(s_0,s_2) \cdot \vecsup{r}{2}, \tfunction(s_0,s_1) \cdot \vecsup{q}{2} + \tfunction(s_0,s_2) \cdot \vecsup{r}{1})$, and both these points are achievable, which contradicts the fact that $\vec{p}$ is extremal.
  Therefore, $\vec{q}$ and $\vec{r}$ are uniquely defined.

If both the slopes of $f_1$ and $f_2$ are defined then Lem.~\ref{lem:slope-optimal-stochastic} implies the desired property of the slope.
Assume now that exactly one of them is not defined; without loss of generality we assume $\lslope(f_1)$ is defined and $\lslope(f_2)$ is not.
Then this means that $f_2$ is not defined at the left of $\firstdim{\vec{r}}$.
Let $x' < \firstdim{\vec{p}}$ such that $f_0(x')$ is defined.
By Lem.~\ref{lem:slope-optimal-stochastic}, there are $\vecprime{q},\vecprime{r}$ such that $(x',f_0(x')) = \tfunction(s_0,s_1)\cdot \vecprime{q} + \tfunction(s_0,s_2) \cdot \vecprime{r}$ and $\firstdim{\vecprime{r}} \le \firstdim{\vec{r}}$.
Since $f_2$ is not defined at the left of $\firstdim{\vec{r}}$, we have that $\firstdim{\vec{r}} = \firstdim{\vec{r}}'$ and so $\vecprime{r} = \vec{r}$.
Therefore, at the left of $x$, $f_0(x') = \tfunction(s_0,s_1) \cdot f_1\left(\firstdim{\vec{q}} + \frac{x' - x}{\tfunction(s_0,s_1)}\right) + \tfunction(s_0,s_2) \cdot \vec{r}$.
Thus, $\lslope(f_0,x)= \tfunction(s_0,s_1) \cdot \frac{1}{\tfunction(s_0,s_1)} \cdot \lslope(f_1,\firstdim{\vec{q}}) = \lslope(f_1,\firstdim{\vec{q}})$.

Now in the case where none of the two slopes are defined, this means that $f_1$ and $f_2$ are not defined at the left of $\firstdim{\vec{q}}$ and $\firstdim{\vec{r}}$ respectively.
By Lem.~\ref{lem:slope-optimal-stochastic}, we can deduce that there is no point of $f_0$ at the left of $\vec{p}$, hence $\lslope(f_0,\firstdim{\vec{p}})$ is also not defined.
\qed
\end{proof}

For the proof of Lem.~\ref{lem:slope-stochastic-two} below we will need
several properties of the dot product. These are formalised in the lemmas below.

\begin{lemma}\label{lem:maximise-scalar-implies-optimal}
  Let $\vec{n} \in \mathbb{R}^2$ be a vector with $\firstdim{\vec{n}} > 0$ and $\seconddim{\vec{n}} > 0$.
  If $\vec{p}$ maximises the dot product $\vec{n} \cdot \vec{p}$ among the achievable vectors $\achievable_s$, then $\vec{p}\in f_s$ (i.e. $\vec{p}$ is a maximal point in $\achievable_s$).
\end{lemma}
\begin{proof}
  Assume towards a contradiction that there is $\vecprime{p} \ge \vec{p}$ which is achievable.
  Either $\firstdim{\vecprime{p}} > \firstdim{\vec{p}}$ and $\seconddim{\vecprime{p}} \ge \seconddim{\vec{p}}$, or  $\firstdim{\vecprime{p}} \ge \firstdim{\vec{p}}$ and $\seconddim{\vecprime{p}} > \seconddim{\vec{p}}$.
  Then as $\firstdim{\vec{n}} > 0$ and $\seconddim{\vec{n}} > 0$, $\vec{n} \cdot \vecprime{p} > \vec{n} \cdot \vec{p}$.
  This contradicts that $\vec{p}$ maximises the dot product.
\qed
\end{proof}

\begin{lemma}\label{lem:scalar-gives-slope}
  Let $\vec{n} \in \mathbb{R}^2$ be a vector with $\firstdim{\vec{n}} > 0$ and $\seconddim{\vec{n}} > 0$.
  If $\vec{p}$ maximises the dot product $\vec{n} \cdot \vec{p}$ among achievable vectors, then $\lslope(f_s,\firstdim{\vec{p}}) \ge \frac{- \firstdim{\vec{n}}}{\seconddim{\vec{n}}} \ge \rslope(f_s,\firstdim{\vec{p}})$.
\end{lemma}
\begin{proof}
  Assume towards a contradiction that $\lslope(f_s,\firstdim{\vec{p}}) < \frac{- \firstdim{\vec{n}}}{\seconddim{\vec{n}}}$.
  Then there is $x' < \firstdim{\vec{p}}$ such that:
  \begin{align*}
    \frac{f_s(x') - f_s(\firstdim{\vec{p}})}{x' - \firstdim{\vec{p}}} & < \frac{- \firstdim{\vec{n}}}{\seconddim{\vec{n}}} \\
    \seconddim{\vec{n}} \cdot (f_s(x') - f_s(\firstdim{\vec{p}})) & > - \firstdim{\vec{n}} \cdot (x' - \firstdim{\vec{p}})  \hspace{1cm}\text{because $x' - \firstdim{\vec{p}} < 0$} \\
    \firstdim{\vec{n}} \cdot x' + \seconddim{\vec{n}} \cdot f_s(x') & > \firstdim{\vec{n}} \cdot \firstdim{\vec{p}} + \seconddim{\vec{n}} \cdot f_s(\firstdim{\vec{p}}) 
  \end{align*}
  This contradicts that $\vec{p}$ maximises the dot product.
\qed
\end{proof}

\begin{lemma}\label{lem:optimal-implies-maximise}
  If $\vec{p}\in f_{s}$ then it maximises the dot product with $(-\lslope(f_s,\firstdim{\vec{p}}),1)$ and with $(-\rslope(f_s,\firstdim{\vec{p}}),1)$ among achievable vectors $\achievable_s$.
\end{lemma}
\begin{proof}
  %Assume towards a contradiction that there is an achievable vector $\vec{p}' = (x',y')$ for which the dot product with $(-\lslope(f_s,x),1)$ is strictly greater than that of $\vec{p}$.
  %We can choose $\vec{p}'$ such that it maximises this product.
  %Then by Lem.~\ref{lem:scalar-gives-slope}, $\lslope(f_{s},x') \ge \lslope(f_s,x)$.
  %So by Lem.~\ref{lem:general-lemma}.\ref{lem:slope-decrease}, $x' \le x$.
  Let $\vecprime{p}$ be a point of $f_s$.
  Since $f_s$ is concave, $\vecprime{p}$ is below the line $\{ (\firstdim{\vec{p}},f_s(\firstdim{\vec{p}})) + t \cdot (1,\lslope(f_s,\firstdim{\vec{p}})) \mid t \in \mathbb{R} \}$.
  So $\seconddim{\vecprime{p}} \le f_s(\firstdim{\vec{p}}) + (\firstdim{\vecprime{p}} - \firstdim{\vec{p}}) \cdot \lslope(f_s,\firstdim{\vec{p}})$.
  Then, looking at the dot product:
  \begin{align*}
    (-\lslope(f_s,\firstdim{\vec{p}}),1) \cdot \vecprime{p} &= \seconddim{\vecprime{p}} - \firstdim{\vecprime{p}} \cdot \lslope(f_s,\firstdim{\vec{p}}) \\
    & \le f_s(\firstdim{\vec{p}}) + (\firstdim{\vecprime{p}} - \firstdim{\vec{p}}) \cdot \lslope(f_s,\firstdim{\vec{p}})  - \firstdim{\vecprime{p}} \cdot \lslope(f_s,\firstdim{\vec{p}}) \\
    & \le f_s(\firstdim{\vec{p}}) - \firstdim{\vec{p}} \cdot  \lslope(f_s,\firstdim{\vec{p}}) \\
    & \le f_s(\firstdim{\vec{p}}) - \firstdim{\vec{p}} \cdot  \lslope(f_s,\firstdim{\vec{p}}) \\
    & \le (-\lslope(f_s,\firstdim{\vec{p}}) , 1) \cdot (\firstdim{\vec{p}},f_s(\firstdim{\vec{p}}))\\
    & \le (-\lslope(f_s,\firstdim{\vec{p}}) , 1) \cdot p
  \end{align*}

  The proof works similarly for $(-\rslope(f_s,\firstdim{\vec{p}}),1)$.
\qed
\end{proof}

\begin{lemma}\label{lem:dot-product-sum}
  Let $\vec{n}$ be a vector in $\mathbb{R}^2$, $Y, Z \subseteq \mathbb{R}^2$, $\lambda_1, \lambda_1 \in \mathbb{R}_{> 0}$, and $X = \lambda_1 Y + \lambda_2 Z$.
  Let $(\vec{y},\vec{z}) \in Y\times Z$ and $\vec{x} = \lambda_1 \vec{y} + \lambda_2 \vec{z}$.
  Then $\vec{y}$ and $\vec{z}$ maximize the dot product with $\vec{n}$ among vectors of $Y$ and $Z$ respectively, if and only if, $\vec{x}$ maximizes the dot product with $\vec{n}$ among vectors of $X$.
\end{lemma}
\begin{proof}
  Assume that $\vec{y}$ and $\vec{z}$ maximise the dot product with $\vec{n}$ and let $\vec{a} \in \lambda_1  Y + \lambda_2 Z$.
  We have that $\vec{a} \cdot \vec{n} = \lambda_1 \vec{b} \cdot \vec{n} + \lambda_2 \vec{c} \cdot \vec{n}$ for some $(\vec{b},\vec{c}) \in Y \times Z$.
  This is smaller than $\lambda_1  \vec{y} \cdot \vec{n} + \lambda_2 \vec{z} \cdot \vec{n}$, since they maximize the dot product with $\vec{n}$ within their respective sets.
  It is therefore smaller than $\vec{x} \cdot \vec{n}$.

  Reciprocally, assume that $\vec{x}$ maximises the dot product with $\vec{n}$.
  If there is $\vec{u} \in Y$ such that $\vec{u} \cdot \vec{n} > \vec{y} \cdot \vec{n}$ then $\lambda_1 \vec{u} + \lambda_2 \vec{z}$ has a greater dot product with $\vec{n}$ than $\vec{x}$, which is a contradiction.
\qed
\end{proof}

\begin{lemma}\label{lem:slope-stochastic-two}
  Let $s_0$ be a stochastic state with two successors $s_1$ and $s_2$.
  Let $\vec{p}\in f_{s_0}$, $\vec{q} \in f_{s_1}$, and $\vec{r}\in f_{s_2}$ such that $\vec{p} = \tfunction(s_0,s_1) \cdot \vec{q} + \tfunction(s_0,s_2) \cdot \vec{r}$.
  For all $\varepsilon > 0$, 
  there are $\varepsilon_1, \varepsilon_2$ such that $\rslope(f_{s_0},\firstdim{\vec{p}} - \varepsilon) \ge \rslope(f_{s_1},\firstdim{\vec{q}} - \varepsilon_1)$, $\rslope(f_{s_0},\firstdim{\vec{p}} - \varepsilon) \ge \rslope(f_{s_2},\firstdim{\vec{r}} - \varepsilon_2)$, and $\varepsilon = \tfunction(s_0,s_1) \cdot \varepsilon_1 + \tfunction(s_0,s_2) \cdot \varepsilon_2$.
\end{lemma}

%\startpara{Intermediary lemmas for proving Lem.~\ref{lem:slope-stochastic-two}}
%To prove this lemma we will first need the following properties. 
%-> now in subsection Dot product
% We now proceed to the proof of Lem.~\ref{lem:slope-stochastic-two}.
\begin{proof}%[Proof of Lem.~\ref{lem:slope-stochastic-two}]
  By Lem.~\ref{lem:optimal-stochastic}, there are $\vecprime{q}$ and $\vecprime{r}$ such that $(\firstdim{\vec{p}} -\varepsilon, f_{s_0}(\firstdim{\vec{p}} - \varepsilon)) = \tfunction(s_0,s_1) \cdot \vecprime{q} + \tfunction(s_0,s_2) \cdot \vecprime{r}$ and
  $\lslope(f_{s_0},\firstdim{\vec{p}} - \varepsilon) = \min \{ \lslope(f_{s_1},\firstdim{\vecprime{q}}) , \lslope(f_{s_2}, \firstdim{\vecprime{r}}) \}$.
  We let $\varepsilon_1 = \firstdim{\vec{r}} - \firstdim{\vecprime{r}}$ and $\varepsilon_2 = \firstdim{\vec{q}} - \firstdim{\vecprime{q}}$.

  %\begin{center} \begin{tikzpicture}[yscale=0.6] \draw (3,2) node[fill=black,inner sep=0pt,minimum size=1mm] (QP) {} node [above] {$\vec{q}'$}; \draw (1.5,1) node[fill=black,inner sep=0pt,minimum size=1mm] (PP) {} node [above] {$\vec{p}'$}; \draw (2,0) node[fill=black,inner sep=0pt,minimum size=1mm] (PS){} node [above]  {$\vec{p}''$}; \draw (0,0) node[fill=black,inner sep=0pt,minimum size=1mm] (RP){} node [above]  {$\vec{r}'$};  \draw (1,-2) node[fill=black,inner sep=0pt,minimum size=1mm] (RS){} node [above]  {$\vec{r}''$}; \draw (2.5,3) -- (3.5,1) node[right]{$f_{s_1}$}; \draw (1,2) -- (2.5,-1) node[right]{$f_{s_0}$}; \draw (RP) -- (RS) node[right]{$f_{s_2}$}; \end{tikzpicture} \end{center}

  Let $n = (-\rslope(f_{s_0},\firstdim{\vec{p}} - \varepsilon), 1)$ be a vector that follows the normal to the slope in $\firstdim{\vec{p}} - \varepsilon$.
  By Lem.~\ref{lem:optimal-implies-maximise}, $\vecprime{p} = (\firstdim{\vec{p}} - \varepsilon, f_{s_0}(\firstdim{\vec{p}} - \varepsilon))$ maximises the dot product with $n$ on the curve of $f_{s_0}$.
  By Lem.~\ref{lem:dot-product-sum}, it is also the case of $\vecprime{q}$ and $\vecprime{r}$ on the curve of $f_{s_1}$ and $f_{s_2}$ respectively.
  Therefore, by Lem.~\ref{lem:scalar-gives-slope}, $\lslope(f_{s_1},\firstdim{\vec{q}} - \varepsilon_1) \ge \rslope(f_{s_0},\firstdim{\vec{p}} - \varepsilon) \ge \rslope(f_{s_1},\firstdim{\vec{q}} - \varepsilon_1)$ and similarly $\rslope(f_{s_0},\firstdim{\vec{p}} - \varepsilon) \ge \rslope(f_{s_2},\firstdim{\vec{r}} - \varepsilon_2)$.
\qed
\end{proof}

Intuitively the next lemma says that for a stochastic state $s_0$ with successors $s_1$ and $s_2$, if $s_1$ has no left accumulation point, then the slopes decrease faster in $s_2$ than in $s_0$.

\begin{reflemma}{lem:accumulation-stochastic-two}
  Let $s_0$ be a stochastic state with two successors $s_1$ and $s_2$, and $\vec{p}$ a left accumulation point of $f_{s_0}$.
  There are points $\vec{q}$ and $\vec{r}$ on $f_{s_1}$ and $f_{s_2}$ respectively such that $\vec{p} = \tfunction(s_0,s_1) \cdot \vec{q} + \tfunction(s_0,s_2) \cdot \vec{r}$.
  Moreover:
  \begin{enumerate}
  \item there is $(s',\vecprime{p}) \in \{ (s_1,\vec{q}), (s_2,\vec{r}) \}$ such that $\vecprime{p}$ is a left accumulation point of $f_{s'}$ and
    $\lslope(f_{s_0},\firstdim{\vec{p}}) = \lslope(f_{s'},\firstdim{\vecprime{p}})$;
  \item there is $\eta(s_0,\vec{p}) > 0$ such that for all $\varepsilon \in [0,\eta(s_0,\vec{p}\ ) \openend$:
    \begin{itemize}
      \item there are $\varepsilon_1, \varepsilon_2$ such that $\rslope(f_{s_0},\firstdim{\vec{p}} - \varepsilon) \ge \rslope(f_{s_1},\firstdim{\vec{q}} - \varepsilon_1)$, $\rslope(f_{s_0},\firstdim{\vec{p}} - \varepsilon) \ge \rslope(f_{s_2},\firstdim{\vec{r}} - \varepsilon_2)$, and $\varepsilon = \tfunction(s_0,s_1) \cdot \varepsilon_1 + \tfunction(s_0,s_2) \cdot \varepsilon_2$.
    %$\rslope(f_{s_0},\firstdim{\vec{p}} - \varepsilon) \ge \rslope(f_{s'},\firstdim{\vecprime{p}} - \varepsilon)$;
      \item if $\vec{r}$ is not a left accumulation point in $f_{s_2}$, or $\lslope(f_{s_0},\firstdim{\vec{p}}) \ne \lslope(f_{s_2},\firstdim{\vec{r}})$, then 
        $f_{s_0}(\firstdim{\vec{p}} - \varepsilon) = \tfunction(s_0,s_1)\cdot f_{s_1}\left(\frac{\firstdim{\vec{p}} - \varepsilon - \tfunction(s_0,s_2) \cdot \firstdim{\vec{r}}}{\tfunction(s_0,s_1)}\right) + \tfunction(s_0,s_2) \cdot \seconddim{\vec{r}}$.
      \item symmetrically, if $\vec{q}$ is not a left accumulation point in $f_{s_1}$, or $\lslope(f_{s_0},\firstdim{\vec{p}}) \ne \lslope(f_{s_1},\firstdim{\vec{q}})$, then 
        $f_{s_0}(\firstdim{\vec{p}} - \varepsilon) = \tfunction(s_0,s_1)\cdot \seconddim{\vec{q}} + \tfunction(s_0,s_2) \cdot f_{s_1}\left(\frac{\firstdim{\vec{p}} - \varepsilon - \tfunction(s_0,s_1) \cdot \firstdim{\vec{q}}}{\tfunction(s_0,s_2)}\right)$.
    \end{itemize}
  \end{enumerate}
\end{reflemma}
\begin{proof}
  Let $(\vecsup{p}{i})_{i \in \mathbb{N}}$ be a sequence of extremal points in ${s_0}$ with increasing first coordinate which converges towards $\vec{p}$.
  By Lem.~\ref{lem:slope-stochastic}, there are $\vecsup{q}{i}$ and $\vecsup{r}{i}$  extremal in $f_{s_1}$ and $f_{s_2}$ respectively, such that $\vecsup{p}{i} = \tfunction(s_0,s_1) \cdot \vecsup{q}{i} + \tfunction(s_0,s_2) \cdot \vecsup{r}{i}$.
%  By Lem.~\ref{lem:slope-stochastic}, there are $\vecsup{q}{i}$ and $\vecsup{r}{i}$ extremal in $X_{s_1}$ and $X_{s_2}$ respectively, such that $\vecsup{p}{i} = \tfunction(s_0,s_1) \cdot \vecsup{q}{i} + \tfunction(s_0,s_2) \cdot \vecsup{r}{i}$.
 % By Lem.~\ref{lem:optimal-stochastic} we can choose $\vecsup{q}{i}$ and $\vecsup{r}{i}$ such that the sequence of their $x$-coordinate is increasing.\romain{this is not so clear.}
  Lem.~\ref{lem:optimal-stochastic} tells use that for a particular index $i$, we could choose $\vecsup{q}{i}$, $\vecsup{r}{i}$, $\vecsup{q}{i+1}$ and $\vecsup{r}{i+1}$ such that $\firstdim{\vecsup{q}{i}} \le \firstdim{\vecsup{q}{i+1}}$ and $\firstdim{\vecsup{r}{i}} \le \firstdim{\vecsup{r}{i+1}}$.
  But since Lem.~\ref{lem:slope-stochastic} shows that $\vecsup{q}{i}$ and $\vecsup{r}{i}$ are uniquely defined, the sequence indeed satisfies the fact that $x$-coordinates are increasing.
  %\romain{this is not so clear. Vojta: I don't see why we need this, can we just take any converging subsequence? Also, couldn't we get item 1 directly if we know that accum point is extremal?}
  %\point{\ref{lem:accumulation-stochastic-two-a}} 
  The sequences $\vecsup{q}{i}$ and $\vecsup{r}{i}$ converge because their first coordinate are increasing and bounded.
  The limits $\vec{q}$ and $\vec{r}$ are such that $\vec{p} = \tfunction(s_0,s_1) \cdot \vec{q} + \tfunction(s_0,s_2) \cdot \vec{r}$.
  %By unicity as proved in Lem.~\ref{lem:slope-stochastic} this means the limits are equal to $\vec{q}$ and $\vec{r}$ respectively.

  \point{\ref{lem:accumulation-stochastic-two-lap}}
%  Since $\vecsup{p}{i}$ contains an infinite number of different points, one of $\vecsup{q}{i}$ and $\vecsup{r}{i}$ contains an infinite number of points. Let say that it is $\vecsup{q}{i}$. Since points of the sequence are optimals and converge to $\vec{q}$, $\vec{q}$ is an accumulation point. Similarly, if the $\vecsup{r}{i}$ contains infinitely many different points, then $\vec{r}$ is an accumulation point.
  %\point{\ref{lem:accumulation-stochastic-two-b}}
  %% We prove that $\vec{p}$ is extremal.
  %% Assume towards a contradiction that $\vec{p}= \lambda_1 \cdot \vecsup{p}{\prime} + \lambda_2 \cdot \vecsup{p}{\prime\prime}$ with $\vecsup{p}{\prime} \ne \vecsup{p}{\prime\prime}$ both optimal, $\lambda_1,\lambda_2 \in [0,1]$ and $\lambda_1 + \lambda_2 = 1$.
  %% Assume that $\vecsup{p}{\prime}$ is on the left of $\vec{p}$, then there is $i$ such that $\firstdim{p^i} > \firstdim{\vecsup{p}{\prime}}$.
  %% By convexity, the line segment $[\vecsup{p}{\prime},\vec{p}]$ is below $\vecsup{p}{i}$ (strictly because $\vecsup{p}{i}$ is extremal) and $[\vecsup{p}{i},\vecsup{p}{\prime\prime}]$ is below $\vec{p}$.
  %% This implies that $[\vecsup{p}{\prime},\vecsup{p}{\prime\prime}]$ is strictly below $\vec{p}$ which is a contradiction with the fact that it is a linear combination of the two.
  %% Therefore the second point is a consequence of Lem.~\ref{lem:slope-stochastic}.\romain{This lem only shows it is equal to the min}
  Since $\vecsup{p}{i}$ contains an infinite number of different points, by Lem.~\ref{lem:general-lemma}.\ref{lem:extremal-slopes}, there should also be an infinite number of different slopes (no more than two extremal points can have the same slope).
  By Lem.~\ref{lem:slope-stochastic}, we have that for all index~$i$, $\lslope(s_0,\vecsup{p}{i}) = \min(\lslope(s_1,\vecsup{q}{i}), \lslope(s_2,\vecsup{r}{i}))$.
  This means one of $\vecsup{q}{i}$ and $\vecsup{r}{i}$ gives an infinite number of slopes, and therefore also an infinite number of points.
  Let say that it is $\vecsup{q}{i}$.
  Since the points of the sequence lie on the Pareto curve and converge to $\vec{q}$, $\vec{q}$ is a left accumulation point.
  Moreover by Lem.~\ref{lem:general-lemma}.\ref{lem:limit-slope}, $\lslope(f_{s_0},\vec{p}) = \lslope(f_{s_1},\vec{q})$.
  Similarly, if the $\vecsup{r}{i}$ contains infinitely many different points, then $\vec{r}$ is a left accumulation point and $\lslope(f_{s_0},\vec{p}) = \lslope(f_{s_2},\vec{r})$.

  \point{\ref{lem:accumulation-stochastic-two-eps}}
  Note that $\lslope(f_{s_0},\cdot)$ is defined at a neighbourhood on the left of $\firstdim{\vec{p}}$ because it is a left accumulation point.
  We let $\eta(s_0,\vec{p})$ be such that $\firstdim{\vec{p}} - \eta(s_0,\vec{p})$ is included in that neighbourhood.
  By Lem.~\ref{lem:slope-stochastic-two}, we have that for all $\varepsilon \in [0, \eta(s_0,\vec{p})\openend$, there are $\varepsilon_1, \varepsilon_2$ such that $\rslope(f_{s_0},\firstdim{\vec{p}} - \varepsilon) \ge \rslope(f_{s_1},\firstdim{\vec{q}} - \varepsilon_1)$, $\rslope(f_{s_0},\firstdim{\vec{p}} - \varepsilon) \ge \rslope(f_{s_2},\firstdim{\vec{r}} - \varepsilon_2)$, and $\varepsilon = \tfunction(s_0,s_1) \cdot \varepsilon_1 + \tfunction(s_0,s_2) \cdot \varepsilon_2$.
    %there is $(s',x') \in \{ (s_1,\firstdim{\vec{q}}); (s_2,\firstdim{\vec{r}}) \}$ such that $\rslope(f_{s_0},\firstdim{\vec{p}} - \varepsilon) \ge \rslope(f_{s'},x' - \varepsilon)$.

%  \point{\ref{lem:accumulation-stochastic-two-d}}

    We now look at the cases where $\vec{r}$ is not a left accumulation point of $s_2$ or $\lslope(f_{s_0},\firstdim{\vec{p}}) \ne \lslope(f_{s_2},\firstdim{\vec{r}})$.
    In the first case, it is clear that the sequence of $\vecsup{r}{i}$ does not contain infinitely many different points.
    In the second case, by Lem.~\ref{lem:slope-stochastic}, $\lslope(f_{s_0},\firstdim{\vec{p}}) < \lslope(f_{s_2},\firstdim{\vec{r}})$.
    By Lem.~\ref{lem:general-lemma}.\ref{lem:limit-slope}, the slope in $\vec{p}$ is the limit of the slopes in the points $\vecsup{p}{i}$,
    So there is an index $j$ after which $\lslope(f_{s_0},\firstdim{\vecsup{p}{i}}) < \lslope(f_{s_2},\firstdim{\vec{r}})$ for $i \ge j$.
    By Lem.~\ref{lem:optimal-implies-maximise}, $\vec{r}$ maximises the dot product with $(-\lslope(f_{s_2},\firstdim{\vec{r}}),1)$ on the curve of $f_{s_2}$.
    By the same lemma, $\vec{p}$ maximises the dot product with $(-\lslope(f_{s_0},\firstdim{\vec{p}}),1)$ on the curve of $f_{s_0}$.
    Since $\vec{p} = \tfunction(s_0,s_1) \cdot \vec{q} + \tfunction(s_0,s_2) \cdot \vec{r}$, Lem.~\ref{lem:dot-product-sum} implies that $\vec{r}$ also maximises the dot product with $(-\lslope(f_{s_0},\firstdim{\vec{p}}),1)$ on the curve of $f_{s_2}$.
    The point $\vec{r}$ therefore maximises the dot product with all vectors $(\alpha,1)$ with $\alpha \in [-\lslope(f_{s_0},\firstdim{\vec{p}}) , -\lslope(f_{s_0},\firstdim{\vec{p}})]$.
    This is in particular the case for $(-\lslope(f_{s_0},\firstdim{\vecsup{p}{i}}),1)$ where $i \ge j$.
    By Lem.~\ref{lem:dot-product-sum}, this implies that $\vecsup{r}{i} = \vec{r}$.
    %and the slope in $\vec{r}$ is the limit of the slopes in the points $\vecsup{r}{i}$.
    %\fbox{\dots}
    Hence, in both cases, we can extract a subsequence of $\vecsup{p}{i}$ where $\vecsup{r}{i}$ is a constant $\vecsup{r}{0}$.
  %\vojta{should not use $\vec{r}_0$, sounds like 0th dim of vector} 
  Since the sequence $\vecsup{r}{i}$ converges to $\vec{r}$, we get $\vecsup{r}{0} = \vec{r}$.
  We have that $\vecsup{q}{i} = \frac{p^i + \tfunction(s_0,s_2) \cdot \vec{r}}{\tfunction(s_0,s_1)}$.

  Let us show that for $\vecprime{p}$ on the Pareto curve of $s_0$ close enough at the left to $\vec{p}$ there is $\vecprime{q}$ on the Pareto curve in ${s_1}$ such that $\vecprime{p} = \tfunction(s_0,s_1)\cdot \vecprime{q} + \tfunction(s_0,s_2) \cdot \vec{r}$.
  Let $\vecprime{p}$ such that $\firstdim{\vecprime{p}} > \firstdim{p^0}$, we will show that such a $\vecprime{q}$ exists.
  By Lem.~\ref{lem:optimal-stochastic} there are $\vecsup{q}{1}$ and $\vecsup{r}{1}$ such that $\vecprime{p} = \tfunction(s_0,s_1) \cdot \vecsup{q}{1} + \tfunction(s_0,s_2) \vecsup{r}{1}$ and because $\vecsup{p}{0} = \tfunction(s_0,s_1) \cdot \vecsup{q}{0} + \tfunction(s_0,s_2) \cdot \vec{r}$, we can choose $\firstdim{\vecprime{r}} \ge \firstdim{r^0} = \firstdim{\vec{r}}$ 
  Since $\vec{p} = \tfunction(s_0,s_1) \cdot \vecsup{q}{0} + \tfunction(s_0,s_2) \cdot \vec{r}$, we can also choose $\vecsup{q}{2}$ and $\vecsup{r}{2}$ such that $\vecprime{p} = \tfunction(s_0,s_1) \cdot \vecsup{q}{2} + \tfunction(s_0,s_2) \cdot \vecsup{r}{2}$ and $\firstdim{r^2} \le \firstdim{r^0} = \vec{r}$.
  There exists $\lambda_1, \lambda_2 \in [0,1]$ whose sum is $1$ and such that $\lambda_1 \cdot \vecsup{r}{1} + \lambda_2 \cdot \vecsup{r}{2} = \vec{r}$.
  Moreover letting $\vecprime{q} = \lambda_1 \cdot \vecsup{q}{1} + \lambda_2 \cdot \vecsup{q}{2}$, we obtain that 
  $\vecprime{p} = \tfunction(s_0,s_1)\cdot \vecprime{q} + \tfunction(s_0,s_2) \cdot \vec{r}$.

  Therefore, for $x' \in [\firstdim{p^0},\firstdim{\vec{p}}\openend$, $f_0(x') = \tfunction(s_0,s_1)\cdot f_1(\firstdim{\vecprime{q}}) + \tfunction(s_0,s_2) \cdot f_2(\firstdim{\vec{r}})$ where $\firstdim{\vecprime{q}}$ is such that $x' = \tfunction(s_0,s_1)\cdot \firstdim{\vecprime{q}} + \tfunction(s_0,s_2) \cdot \firstdim{\vec{r}}$.
    That means $f_0(x') = \tfunction(s_0,s_1)\cdot f_1\left(\frac{x' + \tfunction(s_0,s_2)}{\tfunction(s_0,s_1)}\right) + \tfunction(s_0,s_2) \cdot f_2(\firstdim{\vec{r}})$.
\qed
\end{proof}

\section{Inverse betting game}\label{sec:epsilon-game}

\subsection{Proof of Thm.~\ref{thm:inverse-betting}}
\begin{reftheorem}{thm:inverse-betting}
  Let $\langle V, E, (v_0,c_0), w \rangle$ be a inverse betting game.
  Let $T \subseteq V$ be a target set and $B \in \mathbb{R}$ a bound.
  If from every vertex $v\in V$, \eve has a strategy to ensure visiting $T$ then she has one to ensure visiting it with a valuation of the counter $c \ge 1$ or to exceed the bound, that is she can force a configuration in $(T \times [c_0,+\infty\openend) \cup (V \times [B, +\infty\openend)$.
\end{reftheorem}
\begin{proof}
  Assuming \eve has a strategy to ensure visiting $T$, then she has a memoryless strategy to do so (see for example~\cite{GTW03}).
  We write $\sigma \colon V \to V$ for the function on states associated to this memoryless strategy that ensures visiting $T$ from $v$ (it is easy to recover the full strategy from there: $h \mapsto \sigma(\LAST(h))$).
  We also write $a(v)$ for the length of the longest path from $v$ compatible with $\sigma$ that does not reach $T$.
  Note that $a(v)$ is bounded by $|V|$ and decrease with each step compatible with $\sigma$.

  We define a potential function over configurations: 
  %$p(v,c) = c + \delta(|V|) - \delta(a(v))$ where 
  $p(v,c) = c + W^{a(v)} - W^{|V|}$.
  %where $\delta \colon \mathbb{N} \to \mathbb{R}$ is defined by $\delta(n) = - W^{n}$ and $W$ is the minimum weight that appears in the game.
  Note that because $a$ is bounded, when $p$ goes to infinity, $c$ also goes to infinity.
  %% Note that:
  %% \begin{enumerate}
  %% \item\label{eq:Wdelta} $W \cdot \delta(n) = \delta(n+1)$;
  %%   %W \frac{W}{W-1} - W ^{n+1} = W (1 + \frac{1}{W-1}) - W^{n+1} = \delta(n+1) + W$;
  %% \item\label{eq:delta-monotonic} $\delta$ is strictly monotonic: $\delta(n+1) = \delta(n) + W^{n} \cdot (1-W) > \delta(n)$ because $W \in \openbegin 0, 1\openend$;
  %% %\item $\frac{W}{W-1} = 1 + \frac{1}{W-1} \le 1$ because $W-1 < 0$, hence $\delta(0) = 1 - (1 + \frac{1}{W-1}) \ge 0$;
  %% \item\label{eq:delta-bounded} by the previous point and the fact that $a(v)$ is bounded by $|V|$, for all $v \in V$, $\delta(a(v)) \le \delta(|V|)$;
  %% \item
  %%   because $a$ is bounded, when $p$ goes to infinity, $c$ also goes to infinity.
  %% \end{enumerate}

  The idea for our strategy is to never make this potential decrease.
  We show that it is possible to do so in each configuration that is not a target.
  Given a configuration $(v,c)$, let us write $v_1$ and $v_2$ the successors of $v$ and $c_1$ and $c_2$ the respective valuations of these successors chosen by \adam.
  One of the successors is closer to $T$ with respect to $\sigma$, so without loss of generality we assume that $a(v_1) \le a(v) - 1$.
  We have that $c = w(v,v_1) \cdot c_1 + w(v,v_2) \cdot c_2$.
  \begin{align*}
    p(v_1,c_1) &= c_1 + W^{a(v_1)} - W^{|V|} \\
    & \ge c_1 + W^{a(v)-1} - W^{|V|} \tag{as $a(v_1) \le a(v) - 1$ and $W \le 1$} \\
    p(v_2,c_2) &= c_2 + W^{a(v_2)} - W^{|V|} \\
   & \ge c_2  \tag{as $a(v_2) \le |V|$ and $W \le 1$} \\
    w(v,v_1) \cdot p(v_1,c_1) & + w(v,v_2) \cdot p(v_2,c_2) \\
    & \ge w(v,v_1) \cdot c_1 + w(v,v_1) \cdot (W^{a(v)-1} - W^{|V|}) + w(v,v_2) \cdot c_2  \\
    & \ge c + w(v,v_1) \cdot (W^{a(v)-1} - W^{|V|}) \\
    & \ge c + W \cdot (W^{a(v)-1} - W^{|V|}) \tag{as $w(v,v_1) \ge W$} \\
    & \ge c + W^{a(v)} - W^{|V|+1} \\
    & \ge p(v,c) + W^{|V|} - W^{|V|+1}
  \end{align*}

  %%%% VERSION WITH DELTA:
  %% \begin{align*}
  %%   p(v_1,c_1) &= c_1 + \delta(|V|) - \delta(a(v_1)) \\
  %%   & \ge c_1 + \delta(|V|) - \delta(a(v)-1) ~\text{by monocity of $\delta$~\eqref{eq:delta-monotonic}} \\
  %%   p(v_2,c_2) &= c_2 + \delta(|V|)- \delta(a(v_2)) \\
  %%   & \ge c_2 ~\text{because $\delta(a(v_2)) \le \delta(|V|)$~\eqref{eq:delta-bounded}} \\
  %%   w(v,v_1) \cdot p(v_1,c_1) & + w(v,v_2) \cdot p(v_2,c_2) \\
  %%   & \ge w(v,v_1) \cdot c_1 + w(v,v_1) \cdot (\delta(|V|) - \delta(a(v)-1)) + w(v,v_2) \cdot c_2  \\
  %%   & \ge c + w(v,v_1) \cdot (\delta(|V|) - \delta(a(v)-1)) \\
  %%   & \ge c + W \cdot ( \delta(|V|) - \delta(a(v)-1)) ~ \text{ because $w(v,v_1) \ge W$} \\
  %%   & \ge c + \delta(|V|+1) - \delta(a(v)) \hspace{1cm} \text{ by \eqref{eq:Wdelta}}\\
  %%   & \ge p(v,c) + \delta(|V|+1) - \delta(|V|)
  %% \end{align*}

  Since $w(v,v_1) + w(v,v_2) = 1$, either $p(v_1,c_1) \ge p(v,c) + W^{|V|} - W^{|V|+1}$ or $p(v_2,c_2) > p(v,c) + W^{|V|} - W^{|V|+1}$.
  We define $\sigma'$, to choose $(v_1,c_1)$ in the first case and $(v_2,c_2)$ in the second one.
  Along any path compatible with this strategy the potential at each step increases by at least $W^{|V|} - W^{|V|+1}$, which is strictly positive.
  This means that either it will reach a target (then $a(v)$ can no longer decrease) or it goes to infinity, and so does $c$.
\end{proof}

\subsection{Following a point close to the left accumulation point (proof of Lem.~\ref{lem:following-close}}

We consider a sequence of points that are $\theta({s_0,\vec{p}_0})$ close to $\vec{p}_0$ and with a slope that is decreasing at least as fast as that of their predecessors.

\begin{reflemma}{lem:following-close}
  For stopping games, given $s_0, \vecsup{p}{0}, \varepsilon_0$, such that $\varepsilon_0 < \theta({s_0,\vecsup{p}{0}})$, there exists a finite sequence $\pi(s_0,\vecsup{p}{0},\varepsilon_0) = (s_i,\vecsup{p}{i},\varepsilon_i)_{i \le j}$ such that:
  \begin{itemize}
  \item $(s_i,\vecsup{p}{i})_{i \le j}$ is a path in $T_{s_0,\vecsup{p}{0}}$;
  \item for all $i\le j$, $\rslope(f_{s_i},\firstdim{p^i} - \varepsilon_i) \ge \rslope(f_{s_{i+1}},\firstdim{p^{i+1}} - \varepsilon_{i+1})$.
  \item either $s_j \in U_{s_0,\vecsup{p}{0}}$ and $\varepsilon_j \ge \varepsilon_0$ or $\varepsilon_j \ge \theta({s_0,\vecsup{p}{0}})$.
  \end{itemize}
\end{reflemma}
\begin{proof}
  To construct this path, we will invoke results on inverse betting games presented in Sec.~\ref{subsec:inverse-betting}.
%\ref{sec:epsilon-game}.
  Consider the inverse betting game given by $T_{s_0,\vecsup{p}{0}}$ in Sec.~\ref{subsec:inverse-betting}.
  %with $V_\exists = \statesprob$, $V_\forall = \statesone \cup \statestwo$, $w((s,\vec{p}) , (s',\vec{p}')) = \tfunction(s,s')$ and initial configuration $((s_0,\vecsup{p}{0}),\varepsilon_0)$.

  We show in every configuration $((s_i,\vecsup{p}{i}),c_i)$ with $c_i \le \theta({s_0,\vecsup{p}{0}})$, \adam has a choice in its action such that the successor $((s_{i+1},\vec{p}^{i+1}),c_{i+1})$ will be such that $\lslope(f_{s_i},\firstdim{p^i} - c_i) \ge \lslope(f_{s_{i+1}},\firstdim{p^{i+1}} - c_{i+1})$.
  \begin{itemize}
  \item For \PONE states, this is thanks to Lem.~\ref{lem:slope-convex-union}:
    %\fbox{change to talk about $\rslope$}:
    since $\varepsilon_i \le \theta({s_0,\vecsup{p}{0}}) \le \eta(s_i,\vecsup{p}{i})$, we have that there is $s'$ in $\tfunction(s_i)$ that is a successor of $s_i$ in $T_{s_0,\vecsup{p}{0}}$ (because it has a left accumulation point $\vecsup{p}{i}$ and $\lslope(s_i,\firstdim{p^i}) = \lslope(s',\firstdim{p^i})$) and such that $\rslope(f_{s_i},\firstdim{p^i} -\varepsilon_i) \ge \rslope(f_{s'},\firstdim{p^i} - \varepsilon_i)$.
    Since \adam controls the configurations corresponding to \PONE states, he can chose the appropriate successor.
  \item For \PTWO states, thanks to Lem.~\ref{lem:slope-intersection}:
    %\fbox{change to talk about $\rslope$}: 
    %\romain{Need to add the reward for \PTWO states}
since $\varepsilon_i \le \theta({s_0,\vecsup{p}{0}}) \le \eta(s_i,\vecsup{p}{i})$, we have that there is $s'$ in $\tfunction(s_i)$ that is a successor of $s_i$ in $T_{s_0,\vecsup{p}{0}}$ (because it has a left accumulation point $\vecsup{p}{i}$ and $\lslope(s_i,\firstdim{p^i}+\reward_1(s_i)) = \lslope(s',\firstdim{p^i})$) and such that $\rslope(f_{s_i},\firstdim{p^i} + \reward_1(s_i)-\varepsilon_i) = \rslope(f_{s'},\firstdim{p^i} - \varepsilon_i)$.
Since \adam controls the configurations corresponding to \PTWO states, he can chose the appropriate successor.
\item If $s_i$ is a stochastic state,
  % not in $U_{s_0,\vecsup{p}{0}}$ 
  then by Lem.~\ref{lem:slope-stochastic-two} there are $c_1, c_2 \in \mathbb{R}$ such that the successors $(s_1,\vec{q})$ and $(s_2,\vec{r})$ of $(s_i,\vecsup{p}{i})$ in $T_{s_0,\vecsup{p}{0}}$, are such that $c_i = \tfunction(s_0,s_i) \cdot c_1 + \tfunction(s_i,s_2) \cdot c_2$ and $\rslope(f_{s_i},\firstdim{p^{i}} - c_i) \ge \rslope(f_{s_1},\firstdim{\vec{q}} - c_1)$ and  $\rslope(f_{s_i},\firstdim{p^{i}} - c_i) \ge \rslope(f_{s_2},\firstdim{\vec{r}} - c_2)$.
  So by choosing $c_1$ for $s_1$ and $c_2$ for $s_2$, \adam ensures that for all choices of \eve, $\rslope(f_{s_i},\firstdim{p^i} -\varepsilon_i) \ge \rslope(f_{s_{i+1}},\firstdim{p^{i+1}} - \varepsilon_{i+1})$.
  \end{itemize}
  With such choices for \adam, there is a strategy $\sigma_\forall$ that ensures that $\rslope(f_{s_i},\firstdim{p^i} - c_i)$ is decreasing along the outcome of the game.
  
  By Cor.~\ref{cor:inverse-betting},
  %Lem.~\ref{lem:stopping-implies-one-successor}, there is a strategy for \eve to reach $U_{s_0,\vecsup{p}{0}}$.
  %Therefore we deduce from Thm.~\ref{thm:inverse-betting}, that 
  for any bound $B$ there is a strategy for \eve to ensure we reach a configuration in $(U_{s_0,\vecsup{p}{0}} \times [c_0,+\infty\openend) \cup (V \times [B, +\infty\openend)$.
  This is in particular the case for $B = \theta({s_0,\vecsup{p}{0}})$, and we write $\sigma_\exists$ the corresponding strategy.

  The outcome $\rho$ of $(\sigma_\exists, \sigma_\forall)$ has both properties.
  We now distinguish two types of paths:
  \begin{itemize}
  \item If $\rho$ reaches a configuration with credit greater than $\theta({s_0,\vecsup{p}{0}})$, then let $j$ be the first index where this happen.
    We have that for all $i < j$, 
    $\rslope(f_{s_i},\firstdim{p^i} - \varepsilon_i) \ge \rslope(f_{s_{i+1}},\firstdim{p^{i+1}} - \varepsilon_{i+1})$, thanks to the construction of strategy~$\sigma_\forall$.
  \item Otherwise, we have for all $i$ that $\rslope(f_{s_i},\firstdim{p^i} - \varepsilon_i) \ge \rslope(f_{s_{i+1}},\firstdim{p^{i+1}} - \varepsilon_{i+1})$, thanks to the construction of strategy~$\sigma_\forall$.
    Moreover since $\sigma_\exists$ is winning and we do not get to a configuration in $V \times [\theta({s_0,\vecsup{p}{0}}),+\infty\openend$, $\rho$ reaches $U_{s_0,\vecsup{p}{0}}$ with a credit greater than the initial credit that was $\varepsilon_0$.
      Let $j$ be the first index where this happens.
  \end{itemize}
  In both case we have that $\rho_{\le j}$ %\fbox{defined?} 
  is a witness of the property.
\qed
\end{proof}

\subsection{Proof of Lem.~\ref{lem:constant-slope}}
\begin{reflemma}{lem:constant-slope}
  For all states $s$ with a left accumulation point $\vec{p}$ and for all $\varepsilon < \theta({s,\vec{p}})$, there is some $(s', \vec{p}')$ reachable in $T_{s,\vec{p}}$ such that $\rslope(f_{s'},\firstdim{\vecprime{p}}- \theta({s_0,\vecsup{p}{0}})) \le \rslope(f_s,\firstdim{\vec{p}} - \varepsilon)$.
\end{reflemma}
\begin{proof}
  Consider the sequence $\pi(s,\vec{p},\varepsilon)$ as defined in Lem.~\ref{lem:following-close}.
  Either for the last configuration, $\varepsilon_j$ is greater than $\theta({s,\vec{p}})$ in which case we directly get the property for $(s',\vec{p}') = (s_j,\vecsup{p}{j})$; or we reach $U_{s_0,\vecsup{p}{0}}$. 
  In this case, we have that $\varepsilon_j \ge \varepsilon$, 
  and by Lem.~\ref{lem:accumulation-stochastic-two}.\ref{lem:accumulation-stochastic-two-eps}, there is a successor $(s_{j+1},\vec{p}^{j+1})$ in $T_{s,\vec{p}}$ such that for all $\varepsilon \le \theta({s_0,\vecsup{p}{0}})$, $f_{s_j}(\firstdim{p^j} - \varepsilon) = \tfunction(s_j,s_{j+1})\cdot f_{s_{j+1}}\left(\frac{\firstdim{p^{j}} - \varepsilon - \tfunction(s_j,s') \cdot \vec{r}_1}{\tfunction(s_j,s_{j+1})}\right) + \tfunction(s_j,s') \cdot f_{s'}(\vec{r}_1)$ for some state $s'$ and real $\vec{r}_1$.
  This gives that $\rslope(f_{s_j},\firstdim{p^j} - \varepsilon_j) = \rslope(f_{s_{j+1}},\firstdim{p^{j+1}} - \frac{\varepsilon_j}{\tfunction(s_j,s_{j+1})})$.
  We write $\delta = \max(\{ \tfunction(s,s') \mid s,s' \in \states\} \setminus \{ 1 \})$.
  Since $\tfunction(s_j,s_{j+1}) \le \delta$ and the slope is decreasing:
  $\rslope(f_{s_j},\firstdim{p^j} - \varepsilon_j) \ge \rslope(f_{s_{j+1}},\firstdim{p^{j+1}} - \frac{\varepsilon_j}{\delta})$.
  
  Note that each time we repeat this, $\varepsilon_j$ is multiplied by at least $\frac{1}{\delta}$.
  Hence, after finitely many steps we will reach a value greater than $\theta({s,\vec{p}})$, which shows the property.
\qed
\end{proof}

%\input{non-stopping.tex}

%% I think this is no longer necessary:
%\input{discounted.tex}

\section{Challenges for generalisation of our results}\label{sec:difficulties}

This section enumerates the reasons why we were unable to generalize results to multiple dimensions.

%\startpara{Multi-dimension}
The notions of left-slope and right-slope still makes sense in three dimensions (and we could think of generalising them: front slope, back slope and other directions).
However, the properties that we used in the two-dimensional case are now longer true, as we illustrate in the following lemmas.
We write $\lslope_i$ for the slope in the direction of decreasing $i$-th dimension.
\begin{lemma}
  There is a bounded convex set $X \subset \mathbb{R}^3$, for which there exists two extremal points $(x,y,z)$ and $(x',y',z')$ such that $\lslope_1(f,(x,y)) = \lslope_1(f,(x',y'))$ where $f_X$ is the function defined by $f(x,y)=\sup\{z \mid (x,y,z) \in X\}$.
\end{lemma}
\begin{proof}
  Let $X = \{ (x,y,z) \mid x+y+z \le 1 \land x \le 1 \land y \le 1 \}$.
  The points $a=(1,0,0)$ and $b=(0,1,0)$ are extremal.
  $f_X(1-\varepsilon,0) = \varepsilon$ and $f_X(0-\varepsilon,1) = \varepsilon$ therefore $\lslope_1(f_X,a) = -1 = \lslope_1(f_X,b)$.
\end{proof}
Thus to generalise Lem.~\ref{lem:general-lemma-bis}.\ref{lem:extremal-slopes-bis} to dimension~$n$, we would need to consider more than $n$ directions.
We could for instance consider a property of this kind:

\smallskip
\noindent\textbf{Conjecture: }
  If $X \subset \mathbb{R}^n$ is a bounded convex set and $\vec{p} \ne \vec{p'}$ are extremal points of $X$, then $\exists i.\ \lslope_i(f_X,p) \ne \lslope_i(f_X,p)$.

\smallskip
Assuming this conjecture was true, there would still be the problem of how to follow an accumulation point.
In the two-dimensional case, we chose at each step an accumulation point with the same slope.
Now in higher dimension, we may not be able to follow a accumulation point that has the same slope in all directions (and if the slope is not preserved in all directions, our conjecture cannot be used).
We illustrate this problem, with the following lemma that shows that we could not extend the techniques used in Lem.~\ref{lem:slope-intersection}.
%\ref{lem:slope-convex-union}.
\begin{lemma}
  There are bounded convex sets $X, Y, Z$ such that $X = Y \cap Z$, $\lslope_1(f_X,p) \ne \lslope_1(f_Y,p)$ and $\lslope_2(f_X,p) \ne \lslope_2(f_Z,p)$.
\end{lemma}
\begin{proof}
  Consider $Y = \{ (x,y,z) \in [0,1]^3 \mid z \le 1 - x \}$ and $Z = \{ (x,y,z) \in [0,1]^3 \mid z \le 1 - y \}$ and the point $p = (\frac{1}{2},\frac{1}{2},\frac{1}{2})$.
  Let $\varepsilon \in [0,\frac{1}{2}]$, $f_X(p - (\varepsilon,0,0)) = \frac{1}{2}$, so $\lslope_1(f_X,p) = 0$ and similarly $\lslope_2(f_X,p) = 0$. 
  However $f_Y(p - (\varepsilon,0,0)) = \frac{1}{2} + \varepsilon$, so $\lslope_1(f_Y,p) = -1$ and similarly $\lslope_2(f_Z,p) = -1$. 
\end{proof}

This shows that the idea of following an accumulation point with the same slope cannot be generalised easily to higher number of dimensions.

%\startpara{Non stopping games}\fbox{TODO}

%\section{Existence of persistent strategies}
%\label{sec:complexity}
%\input{memoryless.tex}

\end{document}